\documentclass{article}
\usepackage[utf8]{inputenc}

\usepackage[margin=1in]{geometry}

\usepackage[backref=page]{hyperref}
\usepackage{float,graphicx,verbatim,fullpage,amssymb,amsmath,amsthm,enumerate,multicol, xspace,xcolor,mathtools,thmtools,thm-restate,cleveref,xspace,tikz,caption,subcaption,wasysym,enumitem}
\usepackage{algorithm, cite}
\usepackage[noend]{algpseudocode}
\usepackage{enumitem}
\usepackage{comment}
\usepackage{mdframed}

\usetikzlibrary{calc}

\providecommand{\nnreals}{\mathbb{R}_{\geq 0}}

\definecolor{darkpastelgreen}{rgb}{0.01, 0.75, 0.24}
\definecolor{bleudefrance}{rgb}{0.19, 0.55, 0.91}

\usepackage{mleftright}
\usepackage{hyperref}
    \hypersetup{ colorlinks=true, linkcolor=bleudefrance, filecolor=magenta, urlcolor=bleudefrance, citecolor=darkpastelgreen}

\usetikzlibrary{arrows.meta}
\tikzset{>={Latex[width=1.5mm,length=1.5mm]}}

\newfloat{procedure}{htbp}{loa}
\floatname{procedure}{Procedure}

\def\R{\mathbb{R}}
\def\E{\mathbb{E}}

\newcommand{\opt}{\textsf{OPT}}
\newcommand{\lp}{\textsf{LP}}
\newcommand{\apx}{\textsf{APX}}
\newcommand{\adv}{\textsf{ADV}}

\newtheorem{theorem}{Theorem}[section]

\newtheorem{lemma}[theorem]{Lemma}
\newtheorem{claim}[theorem]{Claim}
\newtheorem{corollary}[theorem]{Corollary}

\theoremstyle{definition}
\newtheorem{definition}[theorem]{Definition}

\usepackage{environ}

\newcommand{\eps}{\varepsilon}

\def\*#1{\mathbf{#1}}
\def\+#1{\mathcal{#1}}

\DeclareMathOperator{\trace}{trace}

\DeclareMathOperator{\argmin}{argmin}
\DeclareMathOperator{\poly}{poly}

\NewEnviron{problem}[1]{
	\begin{center}\fbox{\parbox{6in}{
				{\centering\scshape #1\par}
				\parskip=1ex
				\everypar{\hangindent=1em}
				\BODY
}}\end{center}}

\makeatletter
\newcommand*{\inlineequation}[2][]{
  \begingroup
    \refstepcounter{equation}
    \ifx\\#1\\
    \else
      \label{#1}
    \fi
    \relpenalty=10000 
    \binoppenalty=10000 
    \ensuremath{
      #2
    }
    ~\@eqnnum
  \endgroup
}
\makeatother

\makeatletter
\renewcommand*{\@fnsymbol}[1]{\textcolor{darkpastelgreen}{\ensuremath{\ifcase#1\or *\or \dagger\or \ddagger\or
 \mathsection\or \triangledown\or \mathparagraph\or \|\or **\or \dagger\dagger
   \or \ddagger\ddagger \else\@ctrerr\fi}}}
\makeatother

\def\final{0}  
\def\iflong{\iffalse}
\ifnum\final=0  
\newcommand{\rnote}[1]{{\color{red}[{Raymond: \bf #1}]\marginpar{\color{red}*}}}
\newcommand{\enote}[1]{{\color{green}[{\small Elena: \bf #1}]\marginpar{\color{red}*}}}
\newcommand{\ynote}[1]{{\color{purple}[{Young-San: \bf #1}]\marginpar{\color{red}*}}}
\newcommand{\samson}[1]{{\color{blue}[{Samson: \bf #1}]\marginpar{\color{red}*}}}

\else 
\newcommand{\rnote}[1]{}
\newcommand{\enote}[1]{}
\newcommand{\ynote}[1]{}
\fi

\providecommand{\email}[1]{\href{mailto:#1}{\nolinkurl{#1}\xspace}}

\title{Learning-Augmented Algorithms for Online Linear and Semidefinite Programming}

 \author{\hspace{0.6in}
 Elena Grigorescu\thanks{Purdue University. Supported in part by NSF CCF-1910659, NSF CCF-1910411 and NSF CCF-2228814.
 E-mail: \email{elena-g@purdue.edu}}
 \and
 Young-San Lin\thanks{University of Melbourne. Work done while at Purdue University. Supported in part by NSF CCF-1910659, NSF CCF-1910411 and NSF CCF-2228814.
 E-mail: \email{nilnamuh@gmail.com}}
 \and
 Sandeep Silwal\thanks{Massachusetts Institute of Technology. Supported by an NSF Graduate Research Fellowship under Grant No. 1745302, NSF TRIPODS program (award DMS-2022448), and Simons Investigator Award. 
 E-mail: \email{silwal@mit.edu}}\hspace{0.6in}
 \and
 Maoyuan Song\thanks{Purdue University.  Supported in part by NSF CCF-1910659, NSF CCF-1910411, NSF CCF-2127806, NSF CCF-2228814 and a Ross-Lynn Award. 
 E-mail: \email{song683@purdue.edu}}
 \and
 Samson Zhou\thanks{UC Berkeley and Rice University. Work supported by a Simons Investigator Award and by NSF CCF-1815840, and done in part while at Carnegie Mellon University, USA. 
 E-mail: \email{samsonzhou@gmail.com}}
 }

\date{\today}

\allowdisplaybreaks

\begin{document}

\maketitle

\begin{abstract}

Semidefinite programming (SDP) is a unifying framework that generalizes both linear programming and quadratically-constrained  quadratic programming, while also yielding efficient solvers, both in theory and in practice. However, there exist known impossibility results for approximating the optimal solution when constraints for covering SDPs arrive in an online fashion. In this paper, we study online covering linear and semidefinite programs in which the algorithm is augmented with advice from a possibly erroneous predictor. We show that if the predictor is accurate, we can efficiently bypass these impossibility results and achieve a constant-factor approximation to the optimal solution, i.e., consistency. On the other hand, if the predictor is inaccurate, under some technical conditions, we achieve results that match both the classical optimal upper bounds and the tight lower bounds up to constant factors, i.e., robustness. 

More broadly, we introduce a framework that extends both (1) the online set cover problem augmented with machine-learning predictors, studied by Bamas, Maggiori, and Svensson (NeurIPS 2020), and (2) the online covering SDP problem, initiated by Elad, Kale, and Naor (ICALP 2016).  Specifically, we obtain general online learning-augmented algorithms for covering linear programs with fractional advice and constraints, and initiate the study of learning-augmented algorithms for covering SDP problems. 

Our techniques are based on the primal-dual framework of Buchbinder and Naor (Mathematics of Operations Research, 34, 2009) and can be further adjusted to handle constraints where the variables lie in a bounded region, i.e., box constraints.

\end{abstract}

\section{Introduction}
In the classical online model, an input is iteratively given to an algorithm that must make irrevocable decisions at each point in time, while satisfying a number of changing constraints and optimizing a fixed predetermined objective.
A common metric for evaluating the quality of an online algorithm is the competitive ratio, which is the ratio between the ``cost'' of the algorithm and the best cost in hindsight, i.e., an optimal offline algorithm given the entire input sequence in advance.
In the context of the minimization problems we study in this paper, an online algorithm is $c$-competitive if its cost is at most a multiplicative $c$ factor more than the cost of the optimal solution.
Due to the irrevocable decisions, the changing constraints, or the number of possible different worst-case inputs, many online algorithms have undesirable competitive ratios that are impossible to improve upon without additional assumptions, e.g.,~\cite{alon2009online}.

Due to advances in the predictive ability of machine learning models, a natural approach to overcome these computational barriers is to incorporate models with predictions, e.g. models that predict outcomes based on historical data. 
Unfortunately, due to the lack of provable guarantees on worst-case input, these predictions can be embarrassingly inaccurate
when attempting to generalize to unfamiliar inputs,  as shown in \cite{SzegedyZSBEGF13}, or simply not even satisfy the given constraints~\cite{BamasMS20}. 
Thus, rather than blindly following an erroneous machine learning predictor, recent focus has shifted to studying algorithms that use the output of these models as \emph{advice}, and guarantee good competitive ratios both when the predictions are accurate, i.e., consistency, and when the predictions are poor, i.e., robustness.

Recently, \cite{BamasMS20} studied the learning-augmented online set cover problem and related problems, using their linear programming (LP) formulation to incorporate additional advice through a primal-dual approach. One drawback of their seminal work however, is that they assume both \emph{integral} constraints, as well as integral advice, which restricts the modeling capabilities of the framework; it is natural to ask how an online algorithm can be improved when the advice is given in terms of probability distribution or some other meaningful \emph{fractional} values. For example, for the online set cover problem, fractional advice can indicate how likely a set should be chosen instead of the binary decision of whether a set should be chosen or not; for the ski rental problem, the advice can be presented as a probability distribution over the total number of vacation days; in online network connectivity problems, the advice can indicate how likely an edge should be chosen. 
 
In addition, linear programs cannot handle quadratic constraints and thus often fail to capture important aspects of fundamental optimization problems, which motivates the study of more general programs, such as semidefinite programming (SDP). 
SDP is a unifying framework that generalizes both linear programs and quadratically constrained quadratic programming (QCQP), while also yielding very efficient solvers, both in theory and in practice~\cite{VandenbergheB96}.

\paragraph{Preliminaries.}

For a learning-augmented problem, we are given a confidence parameter $\lambda\in[0,1]$, where lower values of $\lambda$ denote higher confidence in the advice, and higher values denote lower confidence. 
An advice is a suggested solution for the online problem that is given. In the content of optimization problems including linear and semidefinite programming, a solution is a vector consisting of real numbers. We denote $\apx$ as the objective value of the online solution obtained by an online algorithm, and compare it with (1) the objective value of the advice, denoted as $\adv$, and (2) the objective value of the offline optimal solution, denoted as $\opt$.
The consistency and robustness of an online algorithm or solution for a minimization optimization problem are defined as follows.

\begin{definition}
An online solution with objective value $\apx$ is $C(\lambda)$-consistent if $\apx \le C(\lambda) \adv$. An online algorithm is $C(\lambda)$-consistent if it generates a $C(\lambda)$-consistent solution. Here, $C: [0,1] \to \R_{\ge 1}$.
\end{definition}

\begin{definition}
An online solution with objective value $\apx$ is $R(\lambda)$-robust if $\apx \le R(\lambda) \opt$. An online algorithm is $R(\lambda)$-robust if it generates a $R(\lambda)$-robust solution. Here, $R: [0,1] \to \R_{\ge 1}$.
\end{definition}
Learning-augmented algorithms for minimization problems consider the advice while approximately minimizing the objective value when the input arrives online. Intuitively, if an advice is accurate and we trust the advice, then we would like the solution to be close to the optimal, so ideally $C(\lambda)$ should approach $1$ as $\lambda$ approaches $0$. On the other hand, having $\lambda$ being close to $1$ denotes no trust in the advice, so $R(1)$ should be close to the optimal competitive ratio of the best pure online algorithm.

\subsection{Our Contributions}
We give a general paradigm for designing learning-augmented algorithms for online covering linear programming \cite{buchbinder2009online}, which generalizes the set cover problem \cite{BamasMS20}, as well as online covering semidefinite programming \cite{EladKN16}, with possibly non-integral constraints and advice. 
Specifically, we present {\em primal-dual learning-augmented (PDLA)} algorithms for these problems, whose performance is close to the optimal offline solution when the advice is accurate, and also whose performance is  asymptotically close to the optimal oblivious algorithm, if the advice is inaccurate.

Our PDLA algorithms consider the advice while approximately minimizing the objective value when the input arrives online.
Our unifying paradigm applies to both online covering linear programs (LP) and online covering semidefinite programs (SDP) described below.



\paragraph{Online covering linear programs.}

A \emph{covering} LP is defined as follows:
\begin{align}
  \begin{aligned}
    \text{minimize } & c^T x 
    \text{ over } x \in \nnreals^n 
    \text{ subject to } A x \geq \mathbf{1}.
  \end{aligned}
                        \label{opt:covering}
\end{align}
Here, $A \in \R_{\geq 0}^{m \times n}$ consists of $m$ {\em covering}
constraints, $\mathbf{1}$ is a vector of all ones, and $c \in \R_{> 0}^n$ denotes the positive coefficients of the linear cost function. We use $a_{ij}$ to denote the $i$-th row $j$-th column entry of $A$ and $c_j$ to denote the $j$-th coordinate of $c$. 

In the online covering problem, the cost vector $c$ is given offline, and each of these covering constraints (rows) is presented one by one in an online fashion, that is, $m$ can be unknown. The goal is to update $x$ in a non-decreasing manner such that all the covering constraints are satisfied and the objective value $c^T x$ is approximately minimized. 

A classical example captured by the covering LP \eqref{opt:covering} is the set cover problem. In this problem, we have a universe $[m]$ and $n$ sets $S_1, S_2, ..., S_n$ which are subsets of $[m]$. Each set $S_j$ where $j \in [n]$ is associated with a cost $c_j$. The goal is to find a subset of set indices $C \subseteq [n]$, such that (1) the universe is covered by the union of the sets whose corresponding indices are in $C$, i.e., $\cup_{j \in C} S_j = [m]$ and (2) the cost $\sum_{j \in C} c_j$ is minimized. An LP relaxation can be formulated by having $x$ encoding the indicator vector for the set selection, the columns representing the sets, and the rows representing the universe. In the constraint matrix $A$, $a_{ij}=1$ if element $i$ is contained in set $S_j$, otherwise $a_{ij}=0$.

The online set cover problem can be naturally captured by the online covering LP \eqref{opt:covering}. We have the sets given offline and the elements arriving online. Upon the arrival of an element, the information that which sets contain the arriving element is revealed (as a row of $A$), and we have to irrevocably pick the sets to cover the universe. The $O(\log m \log n)$-competitive online algorithm introduced in \cite{alon2009online,BuchbinderN09} solves the online set cover problem by incrementing the indicator vector $x$ in the covering LP and rounding the fractional solution online by a threshold-based approach. The $O(\log m)$ factor is from the integrality gap of the covering LP while the $O(\log n)$ factor is from the competitiveness of the online algorithm.

An $O(\log n)$-competitive algorithm for online covering LPs was presented in \cite{buchbinder2009online}, which simultaneously solves both the primal covering LP \eqref{opt:covering} and the dual \emph{packing} LP \eqref{opt:packing}, defined as follows:
\begin{align}
  \begin{aligned}
    \text{maximize } & \mathbf{1}^T y 
    \text{ over } y \in \nnreals^m 
    \text{ subject to } A^T y \leq c.
  \end{aligned}
                        \label{opt:packing}
\end{align}
The analysis in \cite{buchbinder2009online} crucially uses LP-duality and strong connections between the two solutions to argue that they are both nearly optimal. 
The covering solution $x$ is an exponential function of the packing solution $y$ and
both $x$ and $y$ are monotonically increasing. 
The problem naturally extends to the setting that relies on a \emph{separation oracle} to retrieve an unsatisfied covering constraint where the number of constraints can be unbounded. However, as the framework in \cite{buchbinder2009online} fixes all violating constraints, each arriving constraint might be slightly violated so that each individual fix may require a diminishingly small adjustment. 
Consequently the algorithm may have to address exponentially many constraints. The framework was later modified in \cite{grigorescu2021online} which guarantees that addressing polynomially many constraints suffices.

In the learning-augmented problem, we are given a confidence parameter $\lambda \in [0,1]$ and $x' \in \nnreals^n$ served as a fractional advice for LP \eqref{opt:covering}. 
However, we do not have the guarantees about the advice $x'$. 
More specifically, the objective value of the advice $c^T x'$ could be a horrendous approximation to the optimal objective value of LP \eqref{opt:covering} or $x'$ might not even satisfy the constraints. 

We first show an efficient, consistent, and robust PDLA algorithm for the online covering LP \eqref{opt:covering}. 
We use the condition number $\kappa$ to denote the upper bound for the ratio between the maximum positive entry and the minimum positive entry for each fixed column of $A$.  
For ease of presentation, we assume that $x'$ is \emph{feasible}, i.e., there are no violating constraints caused by the advice $x'$.

\begin{theorem}[Informal] \label{thm:covering-inf}
Given a feasible advice $x' \in \nnreals^n$ for LP \eqref{opt:covering} with confidence parameter $\lambda \in [0,1]$, there exists an $O\left(\frac{1}{1-\lambda}\right)$-consistent and $O\left(\log\frac{\kappa n}{\lambda}\right)$-robust online algorithm for the online covering LP problem that encounters polynomially many violating constraints.
\end{theorem}

The formal version of Theorem~\ref{thm:covering-inf} (Theorem \ref{thm:covering}), which addresses the case when $x'$ is infeasible for LP \eqref{opt:covering}, is presented in Section \ref{sec:covering}. We remark that Theorem~\ref{thm:covering-inf} implies that when $\kappa = \poly(n)$, the algorithm is $\log(n/\lambda)$-robust.  

\paragraph{Online covering semidefinite programs.}

We generalize our approach for  learning-augmented covering LPs to handle a more expressive family of optimization problems, namely, covering semidefinite programs.

First, we introduce some standard notation. A matrix $A \in \mathbb{R}^{d\times d}$ is said to be {\em positive semidefinite} (PSD), i.e., $A \succeq 0$, if $v^T A v \ge 0$ for every vector $v \in \R^d$, or equivalently, all the eigenvalues of $A$ are non-negative. If $A$ is PSD and symmetric, then it is called  {\em symmetric positive semidefinite (SPSD)}. A partial order over SPSD matrices in $\mathbb{R}^{d\times d}$ can be induced such that $A \succeq B$ if and only if $A - B \succeq 0$.

The setting of a covering SDP problem is as follows.
\begin{align}
  \begin{aligned}
    \text{minimize } & c^T x
    \text{ over } x \in \nnreals^n 
    \text{ subject to } \sum_{j=1}^n A_jx_j\succeq B
  \end{aligned}
\end{align}
where $A_1,\ldots,A_n\in\mathbb{R}^{d\times d}$ and $B \in \mathbb{R}^{d\times d}$ are SPSD matrices and $c \in \R_{> 0}^n$.

In the online covering SDP problem introduced in \cite{EladKN16}, we have the matrices $A_1, ..., A_n$ and the cost vector $c$ given offline. In each \emph{round} $i \in [m]$ where $m$ can be unknown, we are given a new SPSD matrix $B^{(i)}$ satisfying $B^{(i)} \succeq B^{(i-1)}$. The goal is to cover $B^{(i)}$ using a linear combination $x_1,\ldots,x_n$ of the matrices $A_1,\ldots,A_n$, so that $\sum_{j=1}^n x_jA_j\succeq B^{(i)}$, while minimizing the cost $c^T x$. 
Moreover, 
we must update $x$ in a non-decreasing manner, so that once some amount of the matrix $A_j$ is used in the covering at a round $i$, then it must be used in all subsequent coverings in later rounds. The online covering SDP problem and its dual in round $i$ are as follows:

\begin{align}
\label{opt:sdp-covering}
\text{minimize } c^T x \text{ over } x \in \nnreals^n \text{ subject to } \sum_{j=1}^n A_jx_j\succeq B^{(i)}\\
\text{maximize } B^{(i)} \otimes Y \text{ over $Y \succeq 0$ subject to } A_j\otimes Y\le c_j \forall j \in[n]\label{opt:sdp-packing}
\end{align}
where we use $A\otimes B$ to denote the Frobenius product of $A$ and $B$, i.e., $A\otimes B=\sum_{i,j}A_{i,j}B_{i,j}=\trace(A^T B)$.

We remark that the formulation of online covering SDP \eqref{opt:sdp-covering} generalizes online covering LP \eqref{opt:covering} when the constraint matrix is known offline but there is no guarantee which covering constraint (row) will arrive. 
In particular, the SDP formulation for online set cover with $n$ sets and $d$ elements all given offline (but without the knowledge of which elements arrive and their order) is the following: we define matrices $A_1,\ldots,A_n \in \{0,1\}^{d \times d}$ where $A_j$ is a diagonal matrix whose diagonal is simply the indicator vector for the $j$-th set across the $d$ elements, i.e., entry $(k,k)$ of $A_j$ is $1$ if and only if set $j$ contains element $k$.
The matrices $B^{(i)}$ encode the variables that must be covered in round $i$, so that $B^{(0)}$ is the all zeros matrix and $B^{(i)}-B^{(i-1)}$ is the all-zeros matrix except with a single one in entry $(k,k)$ for the variable $k$ that must be newly covered in round $i$. 
No online SDP algorithm can achieve competitive ratio $o(\log n)$ because even if fractional sets are allowed,  
no online algorithm can achieve competitive ratio better than $O(\log n)$ for the online set cover problem~\cite{BuchbinderN09}.

An optimal $O(\log n)$-competitive online algorithm was presented in \cite{EladKN16}. Similar to online covering LPs, an important idea in this line of work is to use weak duality and the strong connections between the primal and the dual solutions. Observe that if $x$ and $Y$ are feasible solutions for the primal and the dual, then
\begin{equation} \label{eq:sdp-pd}
c^T x\ge\sum_{j=1}^n(A_j\otimes Y)x_j=\left(\sum_{j=1}^n(A_jx_j)\right)\otimes Y\ge B^{(i)}\otimes Y,
\end{equation}
and hence the primal and the dual satisfy weak duality. 

In the learning-augmented problem, we are given a confidence parameter $\lambda \in [0,1]$ and a vector $x'\in\nnreals^n$ that serves as advice for the linear combination $x_1,\ldots,x_n$ that the algorithm should use for the optimal solution. We have no guarantees about the advice. 
More specifically, the objective value of the advice $c^T x'$ could be a terrible approximation to the optimal objective value of SDP \eqref{opt:sdp-covering} or $x'$ might not even be feasible. 

We use $\kappa$ to denote the ratio of the largest positive eigenvalue to the smallest positive eigenvalue of the matrices $A_1,\ldots,A_n,B^{(1)},\ldots,B^{(m)}$ and achieve the following.
\begin{theorem} [Informal]
\label{thm:sdp-inf}
Given a feasible advice $x' \in \nnreals^n$ for SDP \eqref{opt:sdp-covering} with confidence parameter $\lambda\in[0,1]$, there exists a polynomial time, $O\left(\frac{1}{1-\lambda}\right)$-consistent, and $O\left(\log\frac{\kappa n}{\lambda}\right)$-robust online algorithm for the online covering SDP problem.
\end{theorem}
The formal version of Theorem~\ref{thm:sdp-inf} (namely, Theorem \ref{thm:sdp}), which addresses the case when $x'$ is infeasible for SDP \eqref{opt:sdp-covering}, is presented in Section \ref{sec:sdp}. We remark that Theorem~\ref{thm:sdp-inf} implies that we can achieve a constant factor approximation to the optimal solution when the advice is accurate ($O(1)$-competitive), which breaks the known $\Omega(\log n)$ competitive ratio obtained by the oblivious online algorithm for covering SDP in \cite{EladKN16}. 
Moreover, for $\kappa=\poly(n)$, we match the optimal approximation ratio of $O(\log n)$ up to constants when the advice is arbitrarily bad. 

\paragraph{Adding box constraints.} 
In both the LP and SDP case, it is natural to have the requirement that the variables must lie in a bounded region. Defining the bounded region can be done using ``box constraints".
The setting for the covering LP with box constraints is in the following form.
\begin{align}
  \begin{aligned}
    \text{minimize } & c^T x 
    \text{ over } x \in [0,1]^n 
    \text{ subject to } A x \geq \mathbf{1}.
  \end{aligned}
                        \label{opt:covering-box}
\end{align}
We note that this problem might not have any feasible solution. The upper bound is set to one without loss of generality. Suppose $x_j \in [0, u_j]$, then we can scale $c_j$ and the entries in column $j$ by dividing $u_j$. Again, in the online problem, the covering constraints arrive one at a time. The goal is to update $x$ in a non-decreasing manner subject to each coordinate of $x$ being capped at 1, and approximately minimize the objective $c^T x$. An $O(\log n)$-competitive algorithm for online covering with box constraints was obtained in \cite{buchbinder2009online}.

For the learning-augmented variant, we assume  
that the advice $x' \in [0,1]^n$. Our bound is in terms of a notion of {\em sparsity} $s$ of matrix $A$, which we define formally in Theorem \ref{thm:covering-box}.


\begin{theorem}[Informal] \label{thm:covering-box-inf}
Given a feasible advice $x' \in [0,1]^n$ for LP \eqref{opt:covering-box} with confidence parameter $\lambda \in [0,1]$, there exists an $O\left(\frac{1}{1-\lambda}\right)$-consistent and $O\left(\log\frac{s }{\lambda}\right)$-robust online algorithm for the online covering LP problem with box constraints that encounters polynomially many violating constraints.
\end{theorem}

The formal version of Theorem~\ref{thm:covering-box-inf} (namely, Theorem \ref{thm:covering-box}), which addresses the case when $x'$ is infeasible for LP \eqref{opt:covering-box}, is presented in Section \ref{subsec:covering-box}. This result recovers the bound of the learning-augmented algorithm for the more restricted online set cover problem in \cite{BamasMS20}, where $A \in \{0,1\}^{m \times n}$, and the value of $s$ there is the same as row sparsity (i.e., the maximum number of non-zero entries of any row). We emphasize that our algorithm also considers \emph{fractional} advice, even for the online set cover problem. 


Similarly, the setting for online covering SDP with box constraints is in the following form.
\begin{align}
  \begin{aligned}
    \text{minimize } & c^T x
    \text{ over } x \in [0,1]^n 
    \text{ subject to } \sum_{j=1}^n A_jx_j\succeq B^{(i)}.
  \end{aligned}
                        \label{opt:sdp-covering-box}
\end{align}
Again, we assume without loss of generality that the upper bound of $x_j$ is one and our goal is to update $x$ in a non-decreasing manner such that $c^T x$ is approximately minimized. An $O(s)$-competitive algorithm for online covering SDP with box constraints was presented in \cite{EladKN16} where $s$ is the sparsity of the SDP problem depending on the $A_j$'s and $B^{(i)}$'s. The sparsity notion coincides with the maximum row sparsity when the SDP is used for capturing covering LPs. We use the same sparsity notion as \cite{EladKN16} and show the following theorem for the learning-augmented problem when the advice $x' \in [0,1]^n$ is given.

\begin{theorem}[Informal] \label{thm:sdp-box-inf}
Given a feasible advice $x' \in [0,1]^n$ for SDP \eqref{opt:sdp-covering-box} with confidence parameter $\lambda \in [0,1]$, there exists a polynomial time, $O\left(\frac{1}{1-\lambda}\right)$-consistent, and $O\left(\log\frac{s }{\lambda}\right)$-robust online algorithm for the online covering SDP problem with box constraints.
\end{theorem}

The formal version of Theorem~\ref{thm:sdp-box-inf} (namely, Theorem \ref{thm:sdp-box}), which addresses the case when $x'$ is infeasible for SDP \eqref{opt:sdp-covering-box}, is presented in Section \ref{subsec:sdp-box}. 
In Table \ref{table:pdla}, we summarize the state-of-the-art by comparing the most related results from the literature to our framework. We further refer the reader to the high-level technical approach in  Section~\ref{sec:tech}. 
\begin{table}[!htb]
\begin{center}
\def\arraystretch{1.2}
\footnotesize
\begin{tabular}{|*4{l|}}
\hline
\textbf{Paper} & Problem & Approximation Guarantee & Approach\\
\hline
\cite{BuchbinderN09} & \begin{tabular}{@{}l@{}} online covering LP \end{tabular} & \begin{tabular}{@{}l@{}} with and without \\ box constraints: \\ $O(\log n)$-competitive \end{tabular} & \begin{tabular}{@{}l@{}} continuous \\ guess-and-double \end{tabular} \\
\hline
\cite{EladKN16} & \begin{tabular}{@{}l@{}} online covering SDP \end{tabular} & \begin{tabular}{@{}l@{}} without box constraints: \\ $O(\log n)$-competitive \\ with box constraints: \\ $O(\log s)$-competitive \end{tabular} & \begin{tabular}{@{}l@{}} continuous \\ guess-and-double \\ efficient updating \end{tabular} \\
\hline
\cite{BamasMS20} & \begin{tabular}{@{}l@{}} learning-augmented \\ online set cover \end{tabular} & \begin{tabular}{@{}l@{}} without box constraints: \\ $O(1/(1-\lambda))$-consistent \\ $O(\log (d/\lambda))$-robust \end{tabular} & \begin{tabular}{@{}l@{}} discretized \end{tabular} \\
\hline
\cite{grigorescu2021online} & online covering LP & \begin{tabular}{@{}l@{}} without box constraints: \\ $O(\log n)$-competitive \end{tabular} & \begin{tabular}{@{}l@{}} continuous \\ guess-and-double \\ efficient updating \end{tabular} \\
\hline
\textbf{This Work} & \begin{tabular}{@{}l@{}} learning-augmented online \\ covering LP and SDP \\ with fractional advice \end{tabular} & \begin{tabular}{@{}l@{}} without box constraints: \\ $O(1/(1-\lambda))$-consistent \\ $O(\log (\kappa n/\lambda))$-robust \\ with box constraints: \\ $O(1/(1-\lambda))$-consistent \\ $O(\log (s/\lambda))$-robust \end{tabular} & \begin{tabular}{@{}l@{}} continuous \\ guess-and-double \\ efficient updating \end{tabular} \\
\hline
\end{tabular}
\caption{Summary of the competitive, consistency, and robustness ratios. We assume that the advice is feasible for the learning-augmented problems. Here, $n$ refers to the number of sets or variables, $\lambda \in [0,1]$ refers to the confidence parameter, $\kappa$ refers to the condition number, $d$ refers to row sparsity, and $s$ refers to sparsity. We note that online covering LP with box constraints generalizes online set cover with $s = d$. In \cite{EladKN16}, the guess-and-double scheme is not used for online SDP covering with box constraints.} 
\label{table:pdla}
\end{center}
\end{table}

\paragraph{Applications.} We emphasize that our framework uses a continuous approach that is amenable to other learning-augmented optimization problems, and supports fractional advice, which may  be interpreted as probabilities. 
For example, as in \cite{EladKN16,ahlswede2002strong,wigderson2006derandomizing}, our framework for covering SDPs may be applied to the quantum hypergraph covering problem. 
We apply our PDLA algorithm for covering LPs with box constraints in order to obtain online  algorithms for: (1) the fractional online set cover problem with fractional advice, and for (2) the online group Steiner tree problem on trees, where a min-cut algorithm is used as a separation oracle to retrieve violating constraints. Our learning-augmented solver for the group Steiner problem on trees can be employed as a black-box for other related problems, including group Steiner tree on general graphs, multicast problem on trees, and the non-metric facility location problem \cite{alon2006general}.

\subsection{Overview of Our Techniques}\label{sec:tech}
We now give a technical overview of our algorithms and describe how both our algorithms for covering LPs and SDPs are guided by several common underlying principles.

\paragraph{Previous approaches.} 
A natural starting point would be the PDLA algorithm for online set cover by \cite{BamasMS20}, who adopted the primal-dual approach in \cite{BuchbinderN09} to incorporate external advice. 
We recall that in the covering LP formulation of the online set cover problem, each row denotes an element and each column denotes a set. The constraint matrix has entries that are either zero or one. An entry is one if and only if the element (row) belongs to the set (column). 
Additionally, for the online set cover problem considered in \cite{BamasMS20}, each set is either included in the advice or not, i.e., each coordinate of the suggested indicator vector for the set selection is 
either one or zero. 
While it seems plausible that one could extend the \emph{discretized} approach of \cite{BamasMS20} to handle general coefficients in the constraint matrix, i.e., the online covering LP problem, it is unclear how the growth rates of the variables can be adjusted to guarantee dual feasibility. This is because the positive coefficients in every covering constraint (all with value one) are \emph{balanced} in the online set cover problem, which turns out to be a crucial ingredient to argue dual feasibility by the discretized approach, but we do not have this guarantee for general covering LPs with arbitrary positive coefficients. 
Instead, we use a different framework inspired by the classical online algorithm literature, e.g., \cite{buchbinder2009online,EladKN16}.
We present a short summary of our framework in Figure~\ref{fig:framework} and describe it in more details below. 

\begin{figure*}
\begin{mdframed}
For each update, while there exists a violating constraint:
\begin{enumerate}
\item Determine a violating constraint.
\item Acquire a ``growth rate'' for each variable depending on its coefficient in the violating constraint, the corresponding cost, and the advice.
\item Use a guess-and-double approach to determine how fast each of the variables are increased by their growth rates.
\item Increase the variables continuously until the constraint is satisfied.
\end{enumerate}
\end{mdframed}
\caption{Summary of our framework}
\label{fig:framework}
\end{figure*}

\paragraph{Continuous updates.} 
Each time a new constraint arrives, we \emph{continuously} increase the variables 
until the constraint is satisfied. 
We adjust this growth rate of each variable based on its cost in the objective linear function, its coefficient in the arriving constraint, and the advice: a variable is increased at a slower rate if its cost is more expensive, its coefficient in the constraint has a smaller value, or it is less recommended by the predictor. 
The introduction of fractional values  
in the advice is the main technical obstacle of our setting. 
In particular, our algorithm must behave differently in the case where a primal variable has not reached the fractional value recommended by the advice compared to the case where it has reached the recommended value, but the solution does not satisfy all constraints. 
By contrast, in the integral advice setting of \cite{BamasMS20}, the recommendation value always coincides with the limit at one. 
To this end, once the variable reaches the recommended value, our algorithms judiciously decelerate the growth of the variable.

\paragraph{Guess-and-double.} 
However, by allowing the coefficients of the constraint matrix to be arbitrary, the optimal objective value $\opt$ can be arbitrary and we need a nice estimate for this. 
Thus, we adopt the \emph{guess-and-double} technique, e.g., \cite{buchbinder2009online,EladKN16,grigorescu2021online}, where the algorithm is executed in \emph{phases}, so that in each phase we propose a lower-bound estimate of $\opt$, and the algorithm enters the next phase when the value exceeds our estimate. 
Note that such techniques are not necessary for \cite{BamasMS20}, as their assumption of  coefficients in $\{0,1\}$ implicitly provided bounds on $\opt$.

\paragraph{Efficient updating.} In more general applications, each arriving update may induce a large or even infinite number of constraints, such as an infinite number of directions induced by an SDP constraint. 
But now if we sequentially choose a violating constraint and satisfy the constraint exactly as in \cite{BamasMS20}, then there is no guarantee that we will satisfy all the constraints in a small number of iterations. 
Thus another technique we adopt to ensure efficiency in conjunction with the guess-and-double technique is to satisfy each arriving constraint by a factor of $2$. 
That is, we instead continue to increment the primal variables until the violating constraint is satisfied by a factor of $2$, which ensures that at least one primal variable is doubled, which also implies an upper-bound on the number of violating constraints that must be considered.

\paragraph{Showing robustness and consistency.}
With the introduction of general coefficients within many components of our LP formulation, the robustness analysis in \cite{BamasMS20} is no longer applicable, so instead we adapt the primal-dual analysis in \cite{buchbinder2009online} for general covering LP problems.
In particular, we deal with the general  
coefficients via a delicate telescoping argument for dual feasibility, since we tune and change the growth rates multiple times even within the same phase. 

Towards obtaining the consistency bound, we partition the growth rate based on whether the variable has exceeded the value in the advice, and argue that the growth rate not credited to the advice is at most a certain factor of the growth rate credited to the advice, similar to the line of the argument presented in \cite{BamasMS20}.

\paragraph{Extending to online covering SDPs with advice.} 
In the SDP case, we have arriving matrices rather than arriving elements so that at each time, we need to cover a new PSD matrix $B^{(i)}$ that can be \emph{larger} than the previous PSD matrix $B^{(i-1)}$ in an infinite number of directions. 
We repeatedly look at the direction with the largest mass that needs to be covered, i.e., the largest eigenvector $v$ of $B^{(i)}-\sum_{j=1}^n A_jx_j$. 
Then to cover the direction $v$, we set the growth rate of the coefficient of each matrix $A_j$ proportional to the amount that the matrix \emph{aligns} with $v$, i.e., proportional to $v^TA_jv=A_j\otimes V$, where $V=vv^T$ and $\otimes$ is the Frobenius product. 
Unfortunately, it does not suffice to cover $v$ alone -- there may be many other directions for which $B^{(i)}-\sum_{j=1}^n A_jx_j$ is not covered. 
However, as we satisfy the \emph{implicit linear violating constraint} by a factor of $2$, the amount of vectors we have to cover is similarly upper-bounded as in the aforementioned approach for the online covering LP problem.\\

Lastly, we remark that our unifying framework can be naturally adopted to any online problem that has a covering LP or SDP formulation, equipped with a fractional advice and a confidence parameter.

\subsection{Additional Background and Related Work}
\paragraph{Learning-augmented algorithms.} There has been extensive work in incorporating machine-learned predictions in algorithmic design. Machine-learned predictions to enhance the performance of online algorithms were studied in learning-augmented set cover \cite{BamasMS20}, ski rental \cite{PurohitSK18}, and caching \cite{LykourisV18,AGKP2022multiple}. \cite{LykourisV18,Rohatgi20} showed that an accurate predictor could be leveraged to provide competitive ratios better than the limits of classical online algorithms for online caching, while subsequent learning-augmented algorithms studied scheduling~\cite{LattanziLMV20}, ski rental~\cite{PurohitSK18,GollapudiP19}, nearest neighbor search~\cite{DongIRW20}, clustering~\cite{ErgunFSWZ22}, triangle counting~\cite{ChenEILNRSWWZ22}, frequency estimation~\cite{HsuIKV19}, and other algorithmic and data structural problems~\cite{Mitzenmacher18,BamasMS20,BamasMRS20,WeiZ20,JiangLLRW20,DiakonikolasKTV21,EdenINRSW21,antoniadis2020online,antoniadis2020secretary,AGKP2022multiple,thang2021packing}.

Another direction on this line of research is the stochastic setting, where the input is drawn from a known distribution. This includes online stochastic matching \cite{feldman2009online}, online graph optimization \cite{azar2022online,khalil2017learning}, and other online problems \cite{Mitzenmacher18,mahdian2012online}. These models differ from ours since we solely consider a given fractional advice (that might have a distribution interpretation) as a solution of the optimization problem instead of making assumptions on the input distribution. 

More recently, \cite{AGKP2022multiple} presented a model for solving online covering LPs with multiple predictions, but their model assumes that, unlike the model used in this work and~\cite{BamasMS20}, the prediction(s) is not given upfront, and instead upon the arrival of each constraint, a feasible way of satisfying that constraint is presented. Additionally, their analysis compares the solution presented by the learning-augmented algorithm with benchmark solutions that are consistent with the predictions, and guarantees robustness independently, in contrast to our definition of a confidence parameter to parameterize both confidence and robustness.

\cite{thang2021packing} followed the framework presented in~\cite{BamasMS20} and presented a solution to online packing LPs, in complement to our contributions on online covering LPs. Their methods closely resemble that of~\cite{BamasMS20} and this paper, but requires the advice to be integral, while we extend the framework and generalize to allow for fractional advice. Their work also generalizes online packing LPs to allow non-linear objective functions, which leaves room for potentially more optimal algorithms tailored to linear objectives, avoiding loss in generality.

\paragraph{Classical online algorithms with covering LP formulations.}
One of the most classical online problems is the online set cover problem, solved in the seminal work of \cite{alon2009online} by implicitly using the primal-dual technique of Goemans \cite{GoemansW95}.  
The approach was extended to network optimization problems in undirected graphs in \cite{alon2006general}, ski rental \cite{kmmo}, and paging \cite{bansal2012primal}, then abstracted and generalized to a broad LP-based primal-dual framework for online covering and packing LPs in \cite{buchbinder2009online}. 
We refer the reader to the excellent survey by Buchbinder and Naor \cite{BuchbinderN09}.

\paragraph{Other variants of online covering and packing LP problems.}
The main focus of our framework is on solving fundamental optimization problems in the online setting with advice. There are other variants of online covering and packing LP problems without advice, including optimizing convex objectives \cite{azar2016online}, optimizing $\ell_q$-norm objectives \cite{shen2020online}, mixed covering and packing LPs \cite{azar2013online}, and sparse integer programs \cite{gupta2014approximating}. All these frameworks employ the powerful primal-dual technique to ensure the competitiveness of the online algorithm.

\paragraph{Alternative learning augmented algorithms in the online model.} 
Subsequent to our work, a significantly simpler algorithm with tighter qualitative guarantees was brought to our attention by Roie Levin.
For completeness, we describe the algorithm in Appendix~\ref{app:simple}, but we emphasize that the algorithm is due to Roie Levin and is included here with his permission. 
Nevertheless, we expect that the techniques and analysis that we introduce in this paper may be of independent interest for other related problems or settings, such as the advice being adaptive, or in settings of multiple experts. We believe that understanding the full power of the techniques developed in this paper is an intriguing direction for further research in the still emerging area of learning-augmented algorithms.

\subsection{Organization}
In Section~\ref{sec:covering}, we present the PDLA algorithms for online covering LPs, prove Theorems \ref{thm:covering-inf} and \ref{thm:covering-box-inf}, and show the applications on fractional online set cover with fractional advice and group Steiner tree on trees.
In Section~\ref{sec:sdp}, we present the PDLA algorithms for online covering SDP and prove Theorems \ref{thm:sdp-inf} and \ref{thm:sdp-box-inf}.
We show our experimental evaluations in Section \ref{sec:ex}. 
In Appendix \ref{sec:appendix}, we show the learnability of the covering SDP problem when the input is drawn from a particular distribution, which might be of independent interest.

\section{PDLA Algorithms for Online Covering Linear Programs} \label{sec:covering}

In this section, we prove Theorems~\ref{thm:covering-inf} and \ref{thm:covering-box-inf}. Namely, we present efficient, consistent, and robust PDLA algorithms for online covering LPs.
We recall that the covering LP \eqref{opt:covering} is the following.
\begin{align*}
  \begin{aligned}
    \text{minimize } & c^T x 
    \text{ over } x \in \nnreals^n 
    \text{ subject to } A x \geq \mathbf{1}.
  \end{aligned}
\end{align*}
Here, $A \in \R_{\geq 0}^{m \times n}$ consists of $m$ covering
constraints, $\mathbf{1}$ is a vector of all ones, and $c \in \R_{> 0}^n$ denotes the positive coefficients of the linear cost function.

In the online problem, the covering constraints (rows) are presented one at a time, so $m$ can be unknown. The cost $c$ is given offline. We say that a new \emph{round} starts when a new covering constraint arrives. The goal is to update $x$ in a non-decreasing manner such that $x$ is feasible and the objective $c^T x$ is approximately minimized. In the learning-augmented problem, we are also given an advice $x' \in \nnreals^n$. The goal is to further consider the advice $x'$ to obtain a consistent (compared to the advice) and robust (compared to the optimal) solution $x$.

Recall that an important idea is to simultaneously consider the dual {\em packing} LP \eqref{opt:packing}:
\begin{align*}
  \begin{aligned}
    \text{maximize } & \mathbf{1}^T y 
    \text{ over } y \in \nnreals^m 
    \text{ subject to } A^T y \leq c
  \end{aligned}
\end{align*}
where $A^T$ consists of $n$ packing constraints with an upper bound $c$ given offline.

We use a guess-and-double approach. Let $a_{ij}$ denote the $i$-th row $j$-th column entry of $A$. The algorithm works in \emph{phases}. 
We estimate a lower bound $\alpha(r)$ for $\opt$ in phase $r$. 
In phase $1$, let $\alpha(1) \gets \min_{j \in [n]}\{c_j / a_{1j}\}$ be a proper lower bound for $\opt$. Once the online objective exceeds $\alpha(r)$, we start the new phase $r+1$ from the current violating constraint (let us call it constraint $i_{r+1}$, in particular, $i_1 = 1$ since in the first phase the first constraint arrives first), and double the estimated lower bound, i.e. $\alpha(r+1) \gets 2\alpha(r)$.

In the beginning of phase $r$, $x^{(r)}_j \gets \min\{x'_j,\alpha(r)/(2nc_j)\}$. If $x'_j \le \alpha(r)/(2nc_j)$, then it is possible that $\alpha(r)/(2nc_j)$ is large, so we have to set $x^{(r)}_j = x'_j$ to ensure consistency. On the other hand, if $x'_j \ge \alpha(r)/(2nc_j)$, then it is possible that $x'_j$ is large and the advice is bad, so we have to set $x^{(r)}_j = \alpha(r)/(2nc_j)$ to ensure robustness.\footnote{The same initialization is used for the same reason for online covering LPs with box constraints and online covering SDPs with and without box constraints.}

As $x$ must be updated in a non-decreasing manner, the algorithm maintains $\{x^{(\ell)}_j\}_{\ell \in [r]}$, which denotes the value of each variable $x_j$ from phase 1 to phase $r$, and the value of each variable $x_j$ is actually set to $\max_{\ell \in [r]}\{x^{(\ell)}_j\}$.

Let $A_i$ denote the $i$-th row of $A$. In phase $r$, upon the arrival of constraint $i$, if the advice adequately covers constraint $i$, i.e., $A_i x' \geq 1$, then we increase the variables $x^{(r)}_j$ with growth rate
\[\frac{a_{ij}}{c_j}\left(x^{(r)}_j + \frac{\lambda}{A_i \mathbf{1}} + \frac{(1-\lambda)x'_j \mathbf{1}_{x^{(r)}_j < x'_j}}{A_i x'_c}\right)\]
where $x'_c$ is the advice restricted to entries in which the corresponding variable has not reached the advice yet, or equivalently, $x^{(r)}_j < x'_j$. More formally, the $j$-th coordinate of $x'_c$ is equal to the $j$-th coordinate of $x'$ if $x_j^{(r)} < x_j'$ and 0 otherwise. $\mathbf{1}_{x^{(r)}_j < x'_j}$ is an indicator variable with value 1 if $x^{(r)}_j < x'_j$ and 0 otherwise. Intuitively, the additive term $\lambda/(A_i \mathbf{1})$ is the contribution credited to the online algorithm while the additive term $(1-\lambda)x'_j \mathbf{1}_{x^{(r)}_j < x'_j}/(A_i x'_c)$ is the contribution credited to the advice $x'$. We note that 
\begin{equation} \label{eq:covering-sum-to-1}
\sum_{j=1}^n{\frac{a_{ij}}{A_i \mathbf{1}}} = 1 \text { and } \sum_{j=1}^n\frac{a_{ij}x'_j \mathbf{1}_{x^{(r)}_j < x'_j}}{A_i x'_c} = 1
\end{equation}
so the contributions credited to the online algorithm and the advice are normalized and scaled by a factor of $\lambda$ and $(1 - \lambda)$, respectively.

Alternatively if the advice does not cover the constraint enough, 
i.e., $A_i x' < 1$, we increase $x^{(r)}_j$ with growth rate
\[\frac{a_{ij}}{c_j}\left(x^{(r)}_j + \frac{1}{A_i \mathbf{1}}\right).\]
Namely, since the advice $x'$ is not feasible, we only consider the contribution from the online algorithm.

In order to implement with the proper growth rate, the dual variable $y_i$ initialized as 0 is used as a proxy in round $i$. We increment $x^{(r)}$ until the arriving violating constraint $i$ is satisfied \emph{by a factor of 2}. This guarantees that at least one variable $x^{(r)}_j$ is doubled thus ensures that the algorithm only encounters polynomially many violating constraints.

In each phase $r$, each \emph{iteration} ends in one of the following three cases: (1) the arriving constraint $i$ is satisfied by a factor of 2, (2) there exists a variable $x^{(r)}_j$ that reaches the advice value $x'_j$, or (3) the objective $c^T x^{(r)}$ reaches $\alpha(r)$. The coefficients $D_j$ and $B_j$ for fitting boundary conditions are defined based on the value of $x^{(r)}_j$ in the end of the previous iteration, stored as $\bar{x}_j$. The indicator $\mathbf{1}_{x^{(r)}_j < x'_j}$ used for the growth rate, has the same value as $\mathbf{1}_{\bar{x}_j < x'_j}$ during an iteration.

The main algorithm that uses the phase scheme is presented in Algorithm~\ref{alg:covering:phase}. The continuous primal-dual approach in phase $r$ and round $i$, used as a subroutine, is presented in Algorithm~\ref{alg:covering}.

\begin{algorithm}[!htb]
\caption{Phase Scheme for Algorithm~\ref{alg:covering}} \label{alg:covering:phase}
\begin{algorithmic}[1]
\State{$r \gets 1$, $\alpha(1)\gets\min_{j \in [n]}\{c_j / a_{1j}\}$, $i_1 \gets 1$.} \Comment{initialization for round $1$}
\For{each $j \in [n]$}
    \State $x^{(r)}_j \gets \min\{x'_j,\alpha(r)/(2nc_j)\}$.
\EndFor
\For{arriving covering constraint $A_i x \ge 1$} \Comment{$i=1, 2, ..., m$ for an unknown $m$}
\State{Run Algorithm~\ref{alg:covering}.} \label{line:covering-alg}
\If {$c^T x^{(r)} \ge \alpha(r)$} \Comment{start a new phase}
    \State $r \gets r+1$, $\alpha(r) \gets 2\alpha(r-1)$, $i_r \gets i$.
    \For{each $j \in [n]$}
    \State $x^{(r)}_j \gets \min\{x'_j,\alpha(r)/(2nc_j)\}$.
    \EndFor
    \State Go to line \ref{line:covering-alg}.
\EndIf
\For{each $j \in [n]$}
    \State $x_j \gets \max_{\ell \in [r]} \{x_j^{(\ell)}\}$. \Comment{this is the solution returned in each round $i$}
\EndFor
\EndFor
\end{algorithmic}
\end{algorithm}

\begin{algorithm}[!htb]
\caption{PDLA Online Covering in Phase $r$ and Round $i$} \label{alg:covering}

\renewcommand{\algorithmicrequire}{\textbf{Input:}}
\renewcommand{\algorithmicensure}{\textbf{Output:}}

\algorithmicrequire{ $x^{(r)}$: current solution, $\alpha(r)$: estimate for $\opt$, $A_i$: current row, $i_r$: starting round of phase $r$, $y_k$ for $k=i_r, ..., i-1$, $\lambda$: the confidence parameter, and $x'$: the advice for $x$.}

\algorithmicensure{ Updated $x^{(r)}$ and $y_i$.}

\begin{algorithmic}[1]

    \State $y_i \gets 0$. \Comment{the packing variable $y_i$ is used for the analysis}
    \State $\bar{x} \gets x^{(r)}.$ \label{line:covering-restore} \Comment{$\bar{x}$ is the value of $x^{(r)}$ at the end of the previous iteration}
    \For{$j \in [n]$ }
        \If{$A_i x' \geq 1$}
            \State $D_j \gets \frac{\lambda}{A_i \mathbf{1}} + \frac{(1-\lambda)x'_j \mathbf{1}_{\bar{x}_j < x'_j}}{A_i x'_c},$ 
        \Else
            \State $D_j \gets \frac{1}{A_i \mathbf{1}}.$
        \EndIf
        \State $Y^{(i-1)}_j \gets \sum_{k = i_r}^{i-1} a_{kj} y_k.$ \Comment{if phase $r$ just started, then $i=i_r$, so $Y^{(i-1)}_j \gets 0$}
        
        \State $B_j \gets \frac{\bar{x}_j + D_j}{\exp\left(\left(Y^{(i-1)}_j+ a_{ij}y_i\right) / c_j \right)}.$ 
    \EndFor
    \If {$A_i x^{(r)} < 1$}
        \While {$A_i x^{(r)} < 2$} \label{line:covering-while}
            \State Increase $y_i$ continuously.
            \For{each $j \in [n]$}
                \State Increase $x_j^{(r)}$ simultaneously by
                \[x^{(r)}_j \gets B_j \exp\left(\frac{Y^{(i-1)}_j + a_{ij} y_i }{c_j}\right) - D_j.\] \label{line:covering-inc}
            \EndFor
            \If{any $x^{(r)}_j$ reaches $x'_j$}
                    \State Break and go to line \ref{line:covering-restore}.
            \EndIf
            \If{$c^T x^{(r)} \ge \alpha(r)$}
                \State Break and return.
            \EndIf
        \EndWhile
    \EndIf
\end{algorithmic}
\end{algorithm}

We note that when we start a new phase $r$ from the violating constraint $i_r$, $y_{i_r}$ is set to zero. In phase $r$ and round $i$, we are actually considering the following covering and packing LPs:
\begin{align}
  \begin{aligned}
    \text{minimize } & c^T x^{(r)} 
    \text{ over } x^{(r)} \in \nnreals^n 
    \text{ subject to } A_k x^{(r)} \geq 1 \text{ for $k=i_r, ..., i$}
  \end{aligned} \label{opt:covering:phase}
\end{align}
and
\begin{align}
  \begin{aligned}
    \text{maximize } & \mathbf{1}^T y 
    \text{ over } y_k \ge 0 \text{ for $k=i_r, ..., i$ }
    \text{ subject to } \sum_{k=i_r}^i a_{kj} y_k \leq c_j \text{ for $j \in [n]$}.
  \end{aligned} \label{opt:packing:phase}
\end{align}

Although the augmentation is in a continuous fashion, it is not hard to implement it in a discrete way for any desired precision by binary search. The approach of satisfying each arriving violating constraint by a factor of 2 guarantees that the number of iterations is polynomially upper-bounded. This implies efficient applications on problems that generate covering LPs with
exponentially many or unbounded number of constraints, where violating constraints are retrieved by a separation oracle. The performance of Algorithm \ref{alg:covering:phase} is stated in Theorem \ref{thm:covering}, the formal version of Theorem \ref{thm:covering-inf}.

\begin{theorem} \label{thm:covering}
For the learning-augmented online covering LP problem, there exists an online algorithm that generates $x$ such that
\[c^T x \le \min\left\{O\left(\frac{1}{1-\lambda}\right) c^T x' + O(\log(\kappa n))\opt, O\left(\log\frac{\kappa n}{\lambda}\right) \opt\right\}\]
and encounters $O(n (\log \frac{c^T x}{\alpha(1)})(\log n + \log c^T x + \log \beta))$ violating constraints. If $x'$ is feasible for LP \eqref{opt:covering}, then $c^T x \le O\left(\frac{1}{1-\lambda}\right) c^T x'$.
Here, $\kappa:=\max_{j \in [n]}\{a^{\max}_j / a^{\min}_j\}$, $\beta=\max_{j \in [n]}\{a^{\max}_j/c_j\}$, $a^{\max}_j:=\max_{i \in [m]}\{a_{ij} \mid a_{ij} > 0\}$, and $a^{\min}_j:=\min_{i \in [m]}\{a_{ij} \mid a_{ij} > 0\}$.
\end{theorem}

\begin{proof}
Let $P(r)=c^T x^{(r)}$ and $D(r)=\mathbf{1}^T y$ be the objective value of the primal and the dual in phase $r$, respectively. We use Algorithm~\ref{alg:covering:phase} and first show robustness, i.e., $c^T x \le O\left(\log\frac{\kappa n}{\lambda}\right) \opt$, then show consistency, i.e., $c^T x \le O\left(\frac{1}{1-\lambda}\right) c^T x' + O(\log(\kappa n))\opt$. Finally, we show that only $O(n (\log \frac{c^T x}{\alpha(1)})(\log n + \log c^T x + \log \beta))$ violating constraints are encountered.

\paragraph{Robustness.} To show that $c^T x \le O\left(\log\frac{\kappa n}{\lambda}\right) \opt$, we prove the following five claims:
\begin{enumerate} [label=(\roman*)]
\item \label{gpd1} $x$ is feasible for LP \eqref{opt:covering}.
\item \label{gpd2} For each \emph{finished} phase $r$, $\alpha(r) \leq 6 D(r)$.
\item \label{gpd3} $y/\Theta\left(\log\frac{\kappa n}{\lambda}\right)$ is feasible for LP \eqref{opt:packing:phase} in each phase.
\item \label{gpd4} The sum of the covering objective generated from phase 1 to $r$ is at most $2 \alpha(r)$.
\item \label{gpd5} Let $r'$ be the last phase, then the covering objective $c^T x \leq 2 \alpha(r')$.
\end{enumerate}

Equipped with these claims and weak duality, we have that
\[
c^T x \le \Theta(1) \alpha(r') \le \Theta(1) D(r') = O\left(\log\frac{\kappa n}{\lambda}\right) \opt.
\]

\begin{proof}[Proof of \ref{gpd1}.]
We prove that $x^{(r)}$ is feasible in phase $r$ by showing that the growing function in Algorithm~\ref{alg:covering} line \ref{line:covering-inc} increments $x^{(r)}$ in a continuous manner.
In the beginning of an iteration in round $i$, we have that
\[x^{(r)}_j = B_j \exp\left(\frac{Y^{(i-1)}_j+a_{ij}y_i}{c_j}\right) - D_j = \frac{\bar{x}_j + D_j}{\exp\left(\left(Y^{(i-1)}_j+a_{ij}y_i\right) / c_j\right)} \exp\left(\frac{Y^{(i-1)}_j+a_{ij}y_i}{c_j}\right) - D_j = \bar{x}_j.\]
By following Algorithm~\ref{alg:covering}, we have that $x^{(r)}$ is feasible for LP \eqref{opt:covering:phase} since we terminate the while loop at line \ref{line:covering-while} when $A_i x^{(r)} = 2 > 1$. Then, as $x$ is the coordinate-wise maximum of $\{x^{(r)}\}$, $x$ must be feasible for LP \eqref{opt:covering}.
\end{proof}

\begin{proof}[Proof of \ref{gpd2}.]
In the beginning of phase $r$, $x^{(r)}_j = \min\{x'_j,\alpha(r)/(2nc_j)\}$, so $P(r)$ is initially at most $\alpha(r) / 2$. The total increase of $P(r)$ is at least $\alpha(r) / 2$ as $P(r) \geq \alpha(r)$ when phase $r$ ends. Therefore, it suffices to show that in round $i$,
$$\frac{\partial P(r)}{\partial y_i} \leq 3 \frac{\partial D(r)}{\partial y_i}.$$
By considering the partial derivative of $P(r)$ with respect to $y_i$, when $A_i x' \ge 1$, we have that
\begin{align*}
\frac{\partial P(r)}{\partial y_i} &= \sum_{j=1}^n c_j \frac{\partial x^{(r)}_j}{\partial y_i}\\
&= \sum_{j=1}^n a_{ij} c_j \left(\frac{B_j}{c_j} \exp\left(\frac{Y^{(i-1)}_j + a_{ij} y_i}{c_j}\right)\right)\\
&= \sum_{j=1}^n a_{ij} \left(B_j \exp\left(\frac{Y^{(i-1)}_j + a_{ij} y_i}{c_j}\right) - D_j + D_j\right)\\
&= \sum_{j=1}^n a_{ij} \left(x^{(r)}_j + \frac{\lambda}{A_i \mathbf{1}} + \frac{(1-\lambda)x'_j \mathbf{1}_{x^{(r)}_j < x'_j}}{A_i x'_c}\right)\\
&\le 2 + \lambda + (1-\lambda) = 3 = 3 \frac{\partial D(r)}{\partial y_i}
\end{align*}
where the last inequality is due to the fact that $\sum_{j=1}^n a_{ij} x^{(r)}_j < 2$ and \eqref{eq:covering-sum-to-1}. The same result can be obtained when $A_i x' < 1$ by regarding $\lambda$ as 1.
\end{proof}

\begin{proof}[Proof of \ref{gpd3}.]
We show that $y/\Theta\left(\log\frac{\kappa n}{\lambda}\right)$ is feasible for LP \eqref{opt:packing:phase} in the last round $\ell_r$ of phase $r$. The argument applies to any round prior to $\ell_r$. We recall that within a phase $r$, each iteration ends in one of the following two cases: (1) the arriving constraint $i$ is satisfied by a factor of 2, or (2) there exists a variable $x^{(r)}_j$ that reaches the advice value $x'_j$. Suppose there are $t \le \ell_r - i_r + 1 + n$ iterations in phase $r$ since there are $\ell_r - i_r + 1$ covering constraints and $n$ variables advised. In iteration $p \le t$, let $x^{(r,p)}_j$, $B^{(p)}_j$, and $D^{(p)}_j$ be the value of $x^{(r)}_j$, $B_j$, and $D_j$ in the end of iteration $p$, respectively. Additionally, let $i^{(p)}$ denote the corresponding round in which iteration $p$ occurs and $\hat{Y}^{(p)}_j=\sum_{k =i_r}^{i^{(p)}} a_{kj}y_k$ be the accumulative weighted dual variable sum in the end of iteration $p$ with respect to coordinate $j$ in the primal. Note that we are only incrementing $y_{i^{(p)}}$ in iteration $p$. 
In the beginning of phase $r$, $Y^{(i_r)} = \hat{Y}^{(0)}$ is a zero vector; in the end of the last round $\ell_r$, $\hat{Y}^{(t)} = Y^{(\ell_r)}$. We have that
\[
x^{(r,t)}_j = B^{(t)}_j \exp\left(\frac{\hat{Y}^{(t)}_j}{c_j}\right) - D^{(t)}_j \implies \frac{\hat{Y}^{(t)}_j}{c_j} = \ln \frac{x^{(r,t)}_j + D^{(t)}_j}{B^{(t)}_j}.
\]
Then,
\begin{align*}
\frac{\hat{Y}^{(t)}_j}{c_j} &= \ln \frac{x^{(r,t)}_j + D^{(t)}_j}{B^{(t)}_j} = \ln \left(\frac{x^{(r,t)}_j + D^{(t)}_j}{x^{(r,t-1)}_j + D^{(t)}_j}\exp\left(\frac{\hat{Y}^{(t-1)}_j}{c_j}\right)\right) \\
&= \frac{\hat{Y}^{(t-1)}_j}{c_j} + \ln \frac{x^{(r,t)}_j + D^{(t)}_j}{x^{(r,t-1)}_j + D^{(t)}_j}\\
&= ...\\
&= \frac{\hat{Y}^{(0)}_j}{c_j} + \sum_{p=1}^t \ln \frac{x^{(r,p)}_j + D^{(p)}_j}{x^{(r,p-1)}_j + D^{(p)}_j}\\
&= \ln \prod_{p=1}^t \frac{x^{(r,p)}_j + D^{(p)}_j}{x^{(r,p-1)}_j + D^{(p)}_j}.
\end{align*}
From here, we use the following claims.
\begin{claim} \label{cl:a-b-c1-c2}
For all scalars $a \geq b > 0$ and $c_1 \ge c_2 \geq 0$,
$$\frac{a + c_1}{b + c_1} \leq \frac{a + c_2}{b + c_2}.$$
\end{claim}
\begin{claim} \label{cl:sat-by-2}
$x^{(r)}_j \le 2/a^{\min}_j$.
\end{claim}
\begin{proof}
Suppose $x^{(r)}_j > 2/a^{\min}_j$, then by definition of $a^{\min}_j$, if $x^{(r)}_j > 2/a^{\min}_j$ in the current round, then it contradicts to the terminating condition in Algorithm \ref{alg:covering} line \ref{line:covering-while} since $A_i x^{(r)} \ge a^{\min}_j x^{(r)}_j > 2$; if $x^{(r)}_j > 2/a^{\min}_j$ in an earlier round of phase $r$, then it contradicts to the condition that the current round is violated since $A_i x^{(r)} \ge a^{\min}_j x^{(r)}_j > 2$.
\end{proof}
Recall that $\kappa:=\max_{j \in [n]}\{a^{\max}_j / a^{\min}_j\}$, $a^{\max}_j:=\max_{i \in [m]}\{a_{ij} \mid a_{ij} > 0\}$, and $a^{\min}_j:=\min_{i \in [m]}\{a_{ij} \mid a_{ij} > 0\}$.
Notice that $D^{(p)}_j \geq \frac{\lambda}{a^{\max}_j n}$ for all $p \le t$. Let $D^{min}_j := \min_{p \in [t]} \{D^{(p)}_j\} \geq \frac{\lambda}{a^{\max}_j n}$. By Claim~\ref{cl:a-b-c1-c2}, we have that
\begin{align*}
\frac{\hat{Y}^{(t)}_j}{c_j} &= \ln \prod_{p=1}^t \frac{x^{(r,p)}_j + D^{(p)}_j}{x^{(r,p-1)}_j + D^{(p)}_j}
\leq \ln \prod_{p=1}^t \frac{x^{(r,p)}_j + D^{min}_j}{x^{(r,p-1)}_j + D^{min}_j} \\
&= \ln \frac{x^{(r,t)}_j + D^{min}_j}{x^{(r,0)}_j + D^{min}_j}
\leq \ln \left(1 + \frac{2}{a^{\min}_j D^{min}_j}\right)
= O\left(\log \frac{\kappa n}{\lambda}\right)
\end{align*}
where the last inequality is by Claim \ref{cl:sat-by-2} which implies that $x^{(r,t)}_j \le 2/a^{\min}_j$.
\end{proof}

\begin{proof}[Proof of \ref{gpd4}.]
The sum of the covering objective generated from phase 1 to $r$ is at most
$$\sum_{k=1}^r \alpha(k) = \sum_{k=1}^r \frac{\alpha(r)}{2^{k-r}} \leq 2 \alpha(r).$$
\end{proof}

\begin{proof}[Proof of \ref{gpd5}.]
In the last phase $r'$, $x$ is feasible because it is the coordinate-wise maximum of $\{x^{(r)}\}_{r \in [r']}$. We have
$$c^T x = \sum_{j=1}^n c_j x_j \leq \sum_{j=1}^n \sum_{r=1}^{r'} c_j x^{(r)}_j = \sum_{r=1}^{r'} \left(\sum_{j=1}^n c_j x^{(r)}_j\right) \leq \sum_{r=1}^{r'} \alpha(r) \leq 2 \alpha(r')$$
where the first inequality holds because $x_j = \max_r \{x^{(r)}_j\} \leq \sum_r x^{(r)}_j$, the second inequality is by the fact that the covering objective $\sum_j c_j x^{(r)}_j$ cannot exceed the estimated lower bound $\alpha(r)$, while the last inequality is by \ref{gpd4}.
\end{proof}

\paragraph{Consistency.} We then show that $c^T x \leq O\left(\frac{1}{1-\lambda}\right) c^T x' + O(\log(\kappa n))\opt$. Suppose Algorithm \ref{alg:covering} is in phase $r$ with constraint $i$ arriving. If $A_i x' < 1$, then the change of $c^T x$ simply follows \ref{gpd2} and \ref{gpd3} by regarding $\lambda$ as 1, so $c^T x \le O(\log(\kappa n))\opt$. For the more interesting case when $A_i x' \ge 1$, we decompose $\frac{\partial P(r)}{\partial y_i}=\frac{\partial P_c}{\partial y_i}+\frac{\partial P_u}{\partial y_i}$, where $P_c$ is the component of the primal objective due to the advice and $P_u$ is the component of the primal objective due to the online algorithm. 
We have that the rate of change is credited to $\partial P_c$ if $x^{(r)}_j < x'_j$ and the rate of change is credited to $\partial P_u$ otherwise, if $x^{(r)}_j \ge x'_j$. It suffices to check the rate of change since $x^{(r)}_j$ is initialized to $\min\{x'_j, \alpha(r)/(2n c_j)\}$, i.e., $P_c$ is non-negative and $P_u=0$ in the beginning of phase $r$.\footnote{The same argument is used for proving the consistency in Theorems \ref{thm:covering-box}, \ref{thm:sdp}, and \ref{thm:sdp-box}.} 
In particular, in round $i$,
\begin{align*}
\frac{\partial P_c}{\partial y_i} &= \sum_{j \in [n]: x^{(r)}_j < x'_j} a_{ij} \left(x^{(r)}_j + \frac{\lambda}{A_i \mathbf{1}} + \frac{(1-\lambda)x'_j \mathbf{1}_{x^{(r)}_j < x'_j}}{A_i x'_c}\right) \geq 0 + \frac{a^{\min}_j\lambda}{a^{\max}_jn} + (1 - \lambda),\\
\frac{\partial P_u}{\partial y_i} &= \sum_{j \in [n]: x^{(r)}_j \ge x'_j} a_{ij} \left(x^{(r)}_j + \frac{\lambda}{A_i \mathbf{1}} + \frac{(1-\lambda)x'_j \mathbf{1}_{x^{(r)}_j < x'_j}}{A_i x'_c}\right) \leq 2 + \lambda + 0.
\end{align*}
Thus we have $\frac{\partial P_u}{\partial y_i} \leq \frac{2 + \lambda}{a^{\min}_j\lambda / \left(a^{\max}_jn\right) + 1 - \lambda} \cdot \frac{\partial P_c}{\partial y_i}$, so that
\[\frac{\partial P(r)}{\partial y_i} \leq \left(1 + \frac{2 + \lambda}{\frac{a^{\min}_j\lambda}{a^{\max}_jn} + 1 - \lambda}\right) \frac{\partial P_c}{\partial y_i} = O\left(\frac{1}{1-\lambda}\right) \frac{\partial P_c}{\partial y_i}.\]
We note that if $x'$ is feasible, then $A_i x' \ge 1$ for all $i \in [m]$, so $c^T x \leq O\left(\frac{1}{1-\lambda}\right) c^T x'$.

\paragraph{Bounding the number of violating constraints.} Finally, we show that Algorithm~\ref{alg:covering:phase} encounters

\noindent$O(n (\log \frac{c^T x}{\alpha(1)})(\log n + \log c^T x + \log \beta))$ violating constraints. 
We first show that there are $O(\log (c^T x/\alpha(1)))$ phases. The estimated lower bound $\alpha$ doubles when we start a new phase. Suppose there are $r'$ phases, then $\alpha(1) \cdot 2^{r'-1} = O(c^T x)$. This implies that $r' = O(\log (c^T x/\alpha(1)))$.

In each phase $r$, when a violating  constraint $i$ just arrived, we increment $x^{(r)}$ until the constraint is satisfied by a factor of 2. One of the following two cases must hold after updating $x^{(r)}$: (1) there exists a \emph{large} variable $x^{(r)}_j \geq 1/(2na^{\max}_j)$ that is updated to at least $3x^{(r)}_j/2$, or (2) there exists a \emph{small} variable $x^{(r)}_j < 1/(2na^{\max}_j)$ that becomes large, i.e., $x^{(r)}_j$ is updated to at least $1/(2na^{\max}_j)$. Let $L$ and $S$ be the set of large and small variable subscript labels before the violating constraint $i$ arrives, respectively, and $\hat{x}^{(r)}_j$ be the value of $x^{(r)}_j$ after the update. If none of these two cases holds, then
\[\sum_{j=1}^n a_{ij} \hat{x}^{(r)}_j < \frac{3}{2}\sum_{j \in L}a_{ij}x^{(r)}_j + \sum_{j \in S}\frac{a_{ij}}{2na^{\max}_j} < \frac{3}{2} + \frac{1}{2} = 2\]
where the second inequality is by the fact that constraint $i$ is violated and $a_{ij} \leq a^{\max}_j$.
This implies that constraint $i$ is not satisfied by a factor of $2$ after the update, a contradiction.

Suppose $x^{(r)}_j$ has been updated $t$ times since it was large in phase $r$, then
\[\frac{c_j}{2na^{\max}_j}\left(\frac{3}{2}\right)^{t} =  O(c^T x)\]
which implies that $t=O(\log n + \log c^T x + \log (a^{\max}_j/c_j)) = O(\log n + \log c^T x + \log \beta)$.

There are $n$ variables, each variable can be updated from small to large once and updated $t$ times by a factor of $3/2$ since it was large in each phase. Hence, Algorithm~\ref{alg:covering:phase} encounters
$O(n (\log \frac{c^T x}{\alpha(1)})(\log n + \log c^T x + \log \beta))$
violating covering constraints.
\end{proof}

\subsection{Adding Box Constraints} \label{subsec:covering-box}
We recall that the covering LP \eqref{opt:covering-box} with box constraints is the following. 
\begin{align*}
  \begin{aligned}
    \text{minimize } & c^T x 
    \text{ over } x \in [0,1]^n 
    \text{ subject to } A x \geq \mathbf{1}.
  \end{aligned}
\end{align*}
Our PDLA algorithm simultaneously considers the dual packing LP:
\begin{align}
  \begin{aligned}
    \text{maximize } & \mathbf{1}^T y - \mathbf{1}^T z
    \text{ over } y \in \nnreals^m \text{ and } z \in \nnreals^n
    \text{ subject to } A^T y - z \leq c.
  \end{aligned}
                        \label{opt:packing-box}
\end{align}
We recall that we assume that the advice $x' \in [0,1]^n$.

We use the guess-and-double approach similar to Algorithm~\ref{alg:covering:phase} and \ref{alg:covering}, with a tweak of maintaining the set $T$, which denotes the subscript indices of the $x$ variables that are \emph{tight}. 
In phase $r$, upon the arrival of a violating constraint $i$, we increment $x^{(r)}_j$ in terms of an exponential function of $y_i$ and $z_j$ subject to $x^{(r)}_j \le 1$. 
Once $x^{(r)}_j = 1$, we add $j$ to the tight set $T$ and stop incrementing $x^{(r)}_j$, but we still increment $y_i$ continuously and $z_j$ with rate $a_{ij} y_i$ in order to maintain dual feasibility. 
In the beginning of each phase, $z$ is a zero vector and $T$ is an empty set. 
$z_j = 0$ whenever $j \in [n] \setminus T$.

Subject to $x^{(r)}_j \le 1$, $x^{(r)}_j$ is increased until the cost \emph{outside of the tight set} exceeds the \emph{remaining capacity} by a factor of 2. 
More specifically, when we have a violating constraint $i$, we have that $\sum_{j \in [n] \setminus T} a_{ij}x^{(r)}_j < 1 - \sum_{j \in T} a_{ij}$, and we increment $x^{(r)}_j$ until $\sum_{j \in [n] \setminus T} a_{ij}x^{(r)}_j \ge 2\left( 1 - \sum_{j \in T} a_{ij}\right)$. 
Let $A'_i$ denote the vector with entries of $A_i$ by considering only the coordinates that are not tight. The entry is zero if the coordinate is tight. More formally,
\[
a'_{ij} = 
\begin{cases}
a_{ij} & \text{if } x_j < 1, \text{ i.e., } j \in [n] \setminus T,\\
0 & \text{if } x_j = 1, \text{ i.e., } j \in T,\\
\end{cases}
\]
where $a'_{ij}$ denotes the $j$-th coordinate entry of $A'_i$. 

If the advice adequately covers constraint $i$, i.e., $A'_i x' \geq 1$, we increase the variables $x^{(r)}_j$ with growth rate
\[\frac{a_{ij}}{c_j} \left(x^{(r)}_j + \left(\frac{\lambda}{A'_i \mathbf{1}} + \frac{(1-\lambda)x'_j \mathbf{1}_{x^{(r)}_j < x'_j}}{A'_i x'_c}\right) \left(1 - \sum_{j \in T} a_{ij}\right)\right).\]
Alternatively if the advice does not cover the constraint enough, 
i.e., $A'_i x' < 1$, we increase $x^{(r)}_j$ with growth rate
\[\frac{a_{ij}}{c_j} \left(x^{(r)}_j + \frac{1 - \sum_{j \in S} a_{ij}}{A'_i \mathbf{1}}\right).\]
We note that we increment $x^{(r)}$ only when $1 - \sum_{j \in T} a_{ij} > 0$, since otherwise the constraint $A_i x^{(r)} \ge \sum_{j \in T} a_{ij} \ge 1$ is already satisfied and there is nothing to be updated.
Similar to \eqref{eq:covering-sum-to-1}, we have that 
\begin{equation} \label{eq:covering-box-sum-to-1}
\sum_{j \in [n] \setminus T}{\frac{a_{ij}}{A'_i \mathbf{1}}} = 1 \text { and } \sum_{j \in [n] \setminus T}\frac{a_{ij} x'_j \mathbf{1}_{x^{(r)}_j < x'_j}}{A'_i x'_c} = 1.
\end{equation}

Each iteration now ends in one of the following four cases: (1) the arriving constraint $i$ induces a cost outside of the tight set that exceeds the remaining capacity by a factor of 2, (2) there exists a variable $x^{(r)}_j$ that reaches the advice value $x'_j$, (3) there exists a variable $x^{(r)}_j$ that reaches 1, or (4) the objective $c^T x^{(r)}$ reaches $\alpha(r)$. 
The coefficients $D_j$ and $B_j$ for fitting boundary conditions are defined based on the value of $x^{(r)}_j$ in the end of the previous iteration, stored as $\bar{x}_j$. 
Again, the indicator value used for the growth rate, $\mathbf{1}_{x^{(r)}_j < x'_j}$, has the same value as $\mathbf{1}_{\bar{x}_j < x'_j}$ during an iteration.

The main algorithm that uses the phase scheme is presented in Algorithm~\ref{alg:covering-box:phase}. 
The continuous primal-dual approach in phase $r$ and round $i$, used as a subroutine, is presented in Algorithm~\ref{alg:covering-box}.

\begin{algorithm}[!htb]
\caption{Phase Scheme for Algorithm~\ref{alg:covering-box}} \label{alg:covering-box:phase}
\begin{algorithmic}[1]
\State{$r \gets 1$, $\alpha(1)\gets\min_{j \in [n]}\{c_j / a_{1j}\}$, $T \gets \emptyset$, $i_1 \gets 1$.} \Comment{initialization for round $1$}
\For{each $j \in [n]$}
    \State $x^{(r)}_j \gets \min\{x'_j, \alpha(r)/(2nc_j)\}$, $z_j \gets 0$.
\EndFor
\For{arriving covering constraint $A_i x \ge 1$} \Comment{$i=1, 2, ..., m$ for an unknown $m$}
\State{Run Algorithm~\ref{alg:covering-box}.} \label{line:covering-box-alg}
\If {$c^T x^{(r)} \ge \alpha(r)$} \Comment{start a new phase}
    \State $r \gets r+1$, $\alpha(r) \gets 2\alpha(r-1)$, $T \gets \emptyset$, $i_r \gets i$.
    \For{each $j \in [n]$}
    \State $x^{(r)}_j \gets \min\{x'_j, \alpha(r)/(2nc_j)\}$, $z_j \gets 0$.
    \EndFor
    \State Go to line \ref{line:covering-box-alg}.
\EndIf
\For{each $j \in [n]$}
    \State $x_j \gets \max_{\ell \in [r]} \{x_j^{(\ell)}\}$. \Comment{this is the solution returned in each round $i$}
\EndFor
\EndFor
\end{algorithmic}
\end{algorithm}

\begin{algorithm}[!htb]
\caption{PDLA Online Covering with Box Constraints in Phase $r$ and Round $i$} \label{alg:covering-box}

\renewcommand{\algorithmicrequire}{\textbf{Input:}}
\renewcommand{\algorithmicensure}{\textbf{Output:}}

\algorithmicrequire{ $x^{(r)}$: current solution, $\alpha(r)$: estimate for $\opt$, $A_i$: current row, $i_r$: starting round of phase $r$, $y_k$ for $k=i_r, ..., i-1$, $z$: dual variable vector, $T$: tight variable set, $\lambda$: the confidence parameter, and $x'$: the advice for $x$.}

\algorithmicensure{ Updated $x^{(r)}$, $y_i$, and $z$.}

\begin{algorithmic}[1]

    \State $y_i \gets 0$. \Comment{the packing variable $y_i$ is used for the analysis}
    \State $\bar{x} \gets x^{(r)}.$ \label{line:covering-box-restore} \Comment{$\bar{x}$ is the value of $x^{(r)}$ at the end of the previous iteration}
    \For{$j \in [n]$}
        \If{$A_i x' \geq 1$}
            \State $D_j \gets \left(\frac{\lambda}{A'_i \mathbf{1}} + \frac{(1-\lambda)x'_j \mathbf{1}_{\bar{x}_j < x'_j}}{A'_i x'_c}\right) \left(1 - \sum_{j \in T} a_{ij}\right),$ 
        \Else
            \State $D_j \gets \frac{1 - \sum_{j \in T} a_{ij}}{A'_i \mathbf{1}}.$
        \EndIf
        \State $Y^{(i-1)}_j \gets \sum_{k=i_r}^{i-1} a_{kj} y_k.$ \Comment{if phase $r$ just started, then $i=i_r$, so $Y^{(i-1)}_j \gets 0$}
        \State $B_j \gets \frac{\bar{x}_j + D_j}{\exp\left(\left(Y^{(i-1)}_j+ a_{ij}y_i-z_j\right) / c_j\right)}.$ 
    \EndFor
        \If {$A_i x^{(r)} < 1$}
            \While {$T \neq [n]$ and $\sum_{j \in [n] \setminus T} a_{ij}x^{(r)}_j < 2\left(1 - \sum_{j \in T} a_{ij}\right)$} \label{line:covering-box-while}
            \State Increase $y_i$ continuously.
            \For{each $j \in [n]$}
                \If{$j \in T$}
                    \State Increase $z_j$ with rate $a_{ij} y_i$.
            \EndIf
                \State Update $x_j^{(r)}$ simultaneously by
                \[x^{(r)}_j \gets \min\left\{1,B_j \exp\left(\frac{Y^{(i-1)}_j + a_{ij} y_i -z_j}{c_j}\right) - D_j\right\}.\] \label{line:covering-box-inc}
            \EndFor
            \If{any $x^{(r)}_j=x'_j < 1$}
                    \State Break and go to line \ref{line:covering-box-restore}.
                \EndIf
                \If{any $x^{(r)}_j=1$ for $j \notin T$}
                    \State Add $j$ to $T$ and go to line \ref{line:covering-box-restore}.
                    \EndIf
            \If{$c^T x^{(r)} \ge \alpha(r)$}
                \State Break and return.
            \EndIf
        \EndWhile
        \If {$T = [n]$ and $A_i x^{(r)} < 1$}
        \State \Return no feasible solution. \label{line:covering-box-infeasible}
    \EndIf
        \EndIf
\end{algorithmic}
\end{algorithm}

We note that when we start a new phase $r$ from the violating constraint $i_r$, $y_{i_r}$ is set to zero and $z$ is set to a zero vector. In phase $r$ and round $i$, we are actually considering the following covering and packing LPs with box constraints:
\begin{align}
  \begin{aligned}
    \text{minimize } & c^T x^{(r)} 
    \text{ over } x^{(r)} \in [0,1]^n 
    \text{ subject to } A_k x^{(r)} \geq 1 \text{ for $k=i_r, ..., i$}.
  \end{aligned} \label{opt:covering-box:phase}
\end{align}
and
\begin{align}
  \begin{aligned}
    \text{maximize } & \mathbf{1}^T y - \mathbf{1}^T z
    \text{ over } y_k \ge 0 \text{ for $k=i_r, ..., i$ and } z \in \nnreals^n
    \text{ subject to } \sum_{k=i_r}^i a_{kj} y_k - z_j \leq c_j \text{ for $j \in [n]$}.
  \end{aligned} \label{opt:packing-box:phase}
\end{align}

In Algorithm~\ref{alg:covering-box} line \ref{line:covering-box-inc}, when $j \in T$, $x^{(r)}_j$ remains unchanged, since $z_j$ is increased with rate $a_{ij}y_i$. This ensures that approximate dual feasibility is maintained when $x^{(r)}_j$ is tight. The performance of Algorithm \ref{alg:covering-box:phase} is stated in Theorem \ref{thm:covering-box}, the formal version of Theorem \ref{thm:covering-box-inf}.

\begin{theorem} \label{thm:covering-box}
For the learning-augmented online covering LP problem with box constraints, there exists an online algorithm that generates $x$ such that
\[c^T x \le \min\left\{O\left(\frac{1}{1-\lambda}\right) c^T x' + O(\log s) \opt, O\left(\log \frac{s}{\lambda}\right) \opt\right\}\]
and encounters $O(n \log s \log \frac{c^T x}{\alpha(1)})$ violating constraints. If $x'$ is feasible for LP \eqref{opt:covering-box}, then $c^T x \le O\left(\frac{1}{1-\lambda}\right) c^T x'$.
Here,
\[s:= \max_{i \in [m], T \subseteq [n]} \left\{\frac{\sum_{j \in [n] \setminus T} a_{ij}}{1-\sum_{j \in T} a_{ij}} \mid 1-\sum_{j \in T} a_{ij} > 0\right\}.\]
\end{theorem}

\begin{proof}

Let $P(r)=c^T x^{(r)}$ and $D(r)=\mathbf{1}^T y - \mathbf{1}^T z$ be the objective value of the primal and the dual in phase $r$, respectively. We use Algorithm~\ref{alg:covering-box:phase} and first show robustness, i.e., $c^T x \le O\left(\log\frac{s}{\lambda}\right) \opt$, then show consistency, i.e., $c^T x \le O\left(\frac{1}{1-\lambda}\right) c^T x' + O(\log s)\opt$. Finally, we show that only $O(n \log s \log \frac{c^T x}{\alpha(1)})$ violating constraints are encountered.

\paragraph{Robustness.}
To show that $c^T x \le O\left(\log\frac{s}{\lambda}\right) \opt$, we prove the following five claims:
\begin{enumerate} [label=(\roman*)]
\item \label{gpd1-box} If line \ref{line:covering-box-infeasible} in Algorithm \ref{alg:covering-box} is not encountered, then $x$ is feasible for LP \eqref{opt:covering-box}.
\item \label{gpd2-box} For each \emph{finished} phase $r$, $\alpha(r) \leq 6 D(r)$.
\item \label{gpd3-box} $(y,z)/\Theta\left(\log\frac{s}{\lambda}\right)$  
is feasible for LP \eqref{opt:packing-box:phase} in each phase.
\item \label{gpd4-box} The sum of the covering objective generated from phase 1 to $r$ is at most $2 \alpha(r)$.
\item \label{gpd5-box} Let $r'$ be the last phase, then the covering objective $c^T x \leq 2 \alpha(r')$.
\end{enumerate}
The proofs of \ref{gpd4-box} and \ref{gpd5-box} are the same as the ones in Theorem~\ref{thm:covering}, hence omitted. Equipped with these claims and weak duality, we have that
\[
c^T x \le \Theta(1) \alpha(r') \le \Theta(1) D(r') = O\left(\log\frac{s}{\lambda}\right) \opt.
\]

\begin{proof}[Proof of \ref{gpd1}.]
We prove that $x^{(r)}$ is feasible in phase $r$ by showing that the growing functions in Algorithm \ref{alg:covering-box} line \ref{line:covering-box-inc} increments $x^{(r)}$ in a continuous manner. In the beginning of an iteration in round $i$, we have that
\[x^{(r)}_j = B_j \exp\left(\frac{Y^{(i-1)}_j+a_{ij}y_i}{c_j}\right) - D_j = \frac{\bar{x}_j + D_j}{\exp\left(\left(Y^{(i-1)}_j+a_{ij}y_i\right) / c_j\right)} \exp\left(\frac{Y^{(i-1)}_j+a_{ij}y_i}{c_j}\right) - D_j = \bar{x}_j.\]
By following Algorithm \ref{alg:covering-box}, we have that $x^{(r)}$ is feasible for LP \eqref{opt:covering-box:phase} since we terminate the while loop at line \ref{line:covering-box-while} when the cost outside of the tight set exceeds the remaining capacity by a factor of 2. Then, as $x$ is the coordinate-wise maximum of $\{x^{(r)}\}$, capped at 1, $x$ must be feasible for LP \eqref{opt:covering-box}.
\end{proof}

\begin{proof}[Proof of \ref{gpd2-box}.]
In the beginning of phase $r$, $x^{(r)}_j = \min\{x'_j,\alpha(r) / (2nc_j)\}$, so $P(r)$ is initially at most $\alpha(r) / 2$. The total increase of $P(r)$ is at least $\alpha(r) / 2$ as $P(r) \geq \alpha(r)$ when phase $r$ ends. Therefore, it suffices to show that
$$\frac{\partial P(r)}{\partial y_i} \leq 3 \frac{\partial D(r)}{\partial y_i}.$$
Recall that $T$ is the set of the tight $x$ indices and the $z_j$ variables in $T$ are increasing with rate $a_{ij} y_i$, we have that
\[\frac{\partial D(r)}{\partial y_i} = 1 - \sum_{j \in T} a_{ij}.\]
When $y_i$ is increasing, this partial derivative is always non-negative due to the condition in Algorithm \ref{alg:covering-box} line \ref{line:covering-box-while}. Namely, at the point when $x^{(r)}_j$ is tight, we add $j$ to $T$, and this may make the partial derivative change from non-negative to negative. In this case, we have that $2(1 - \sum_{j \in T} a_{ij}) < 0$ so the condition automatically fails. From here, when $A_i x' \geq 1$, we have that
\begin{align*}
\frac{\partial P(r)}{\partial y_i} &= \sum_{j \in [n] \setminus T} c_j \frac{\partial x^{(r)}_j}{\partial y_i}\\
&= \sum_{j \in [n] \setminus T} a_{ij} c_j \left(\frac{B_j}{c_j} \exp\left(\frac{Y^{(i-1)}_j + a_{ij} y_i - z_j}{c_j}\right)\right)\\
&= \sum_{j \in [n] \setminus T} a_{ij} \left(B_j \exp\left(\frac{Y^{(i-1)}_j + a_{ij} y_i - z_j}{c_j}\right)-D_j+D_j\right)\\
&= \sum_{j \in [n] \setminus T} a_{ij} \left(x^{(r)}_j + \left(\frac{\lambda}{A'_i \mathbf{1}} + \frac{(1-\lambda)x'_j \mathbf{1}_{\bar{x}_j < x'_j}}{A'_i x'_c}\right) \left(1 - \sum_{j \in T} a_{ij}\right)\right)\\
&\le 2\left(1 - \sum_{j \in T} a_{ij}\right) + (\lambda + (1-\lambda))\left(1 - \sum_{j \in T} a_{ij}\right)\\
&= 3\left(1 - \sum_{j \in T} a_{ij}\right) = 3 \frac{\partial D(r)}{\partial y_i}
\end{align*}
where the last inequality is due to the fact that $\sum_{j \in [n] \setminus T} a_{ij}x^{(r)}_j < 2\left(1 - \sum_{j \in T} a_{ij}\right)$ and \eqref{eq:covering-box-sum-to-1}. The same result can be obtained when $A_i x' < 1$ by regarding $\lambda$ as 1.
\end{proof}

\begin{proof}[Proof of \ref{gpd3-box}.]
We show that $(y,z)/\Theta\left(\log\frac{\kappa n}{\lambda}\right)$ is feasible for LP \eqref{opt:packing-box:phase} in the last round $\ell_r$ of phase $r$. The argument applies to any round prior to $\ell_r$. We recall that within a phase $r$, an iteration ends in one of the following three cases: (1) the arriving constraint $i$ induces a cost outside of the tight set that exceeds the remaining capacity by a factor of 2, (2) there exists a variable $x^{(r)}_j$ that reaches the advice value $x'_j$, or (3) there exists a variable $x^{(r)}_j$ that reaches 1. Suppose there are $t \le \ell_r - i_r + 1 + 2n$ iterations in phase $r$ since there are $\ell_r - i_r + 1$ covering constraints, $n$ variables advised, and each variable is capped at 1 at most once. In iteration $p \le t$, let $x^{(r,p)}_j$, $B^{(p)}_j$, and $D^{(p)}_j$ be the value of $x^{(r)}_j$, $B_j$, and $D_j$ in the end of iteration $p$, respectively. Additionally, let $i^{(p)}$ denote the corresponding round in which iteration $p$ occurs and $\hat{Y}^{(p)}_j=\sum_{k =i_r}^{i^{(p)}} a_{kj}y_k - z_j$ be the accumulative weighted dual variable sum in the end of iteration $p$ with respect to coordinate $j$ in the primal. Note that we are only incrementing $y_{i^{(p)}}$ in iteration $p$ and possibly incrementing $z_j$ with rate $a_{i^{(p)}j}y_{i^{(p)}}$ if $j \in T$. If $j \in T$ in the beginning of iteration $p$, then $x_j^{(r,p-1)} = x_j^{(r,p)} = 1$ and $\hat{Y}^{(p-1)}_j = \hat{Y}^{(p)}_j$ because $a_{i^{(p)}j} y_j$ and $z_j$ cancel out each other. The goal is to show that $\hat{Y}^{(t)}_j / c_j \leq \Theta\left(\log\frac{
s}{\lambda}\right)$.

By using a similar argument in Theorem \ref{thm:covering} \ref{gpd3}, we can arrive at
$$\frac{\hat{Y}^{(t)}_j}{c_j} = \ln \prod_{p=1}^t \frac{x^{(r,p)}_j + D^{(p)}_j}{x^{(r,p-1)}_j + D^{(p)}_j} \leq \ln \frac{x^{(r,t)}_j + D^{min}_j}{x^{(r,0)}_j + D^{min}_j}$$
where $D^{min}_j := \min_{p \in [t]} \{D^{(p)}_j\}$.

From here, notice that $x^{(r,t)}_j \leq 1$, $x^{(r,0)}_j \geq 0$, and \[D^{min}_j \geq \min_{\substack{
i \in \{i_r, ..., \ell_r\}, T \subseteq [n] \\ 1 - \sum_{j \in T} a_{ij}>0}}\left\{\frac{\lambda\left(1 - \sum_{j \in T} a_{ij}\right)}{A'_i \mathbf{1}}\right\} \ge \min_{\substack{
i \in \{i_r, ..., \ell_r\}, T \subseteq [n] \\ 1 - \sum_{j \in T} a_{ij}>0}}\left\{\frac{\lambda\left(1 - \sum_{j \in T} a_{ij}\right)}{\sum_{j \in [n] \setminus T} a_{ij}}\right\} \ge \frac{\lambda}{s}.\]
We thus have 
$$\ln \frac{x^{(r,t)}_j + D^{min}_j}{x^{(r,0)}_j + D^{min}_j} \leq \ln \left(1 + \frac{1}{D^{min}_j}\right) = O\left(\log \frac{s}{\lambda}\right).$$
\end{proof}

\paragraph{Consistency.} We then show that $c^T x \leq O\left(\frac{1}{1-\lambda}\right) c^T x' + O(\log s)\opt$. Suppose Algorithm \ref{alg:covering-box} is in phase $r$ with constraint $i$ arriving. If $A_i x' < 1$, then the change of $c^T x$ simply follows \ref{gpd2-box} and \ref{gpd3-box} by regarding $\lambda$ as 1, so $c^T x \le O(\log s)\opt$. For the more interesting case when $A_i x' \ge 1$, we decompose $\frac{\partial P(r)}{\partial y_i}=\frac{\partial P_c}{\partial y_i}+\frac{\partial P_u}{\partial y_i}$, where $P_c$ is the component of the primal objective due to the advice and $P_u$ is the component of the primal objective due to the online algorithm. 
We have that the rate of change is credited to $\partial P_c$ if $x^{(r)}_j < x'_j$ and the rate of change is credited to $\partial P_u$ otherwise, if $x^{(r)}_j \ge x'_j$. 
In particular,
\begin{align*}
\frac{\partial P_c}{\partial y_i} &= \sum_{j \in [n] \setminus T: x^{(r)}_j < x'_j} a_{ij} \left(x^{(r)}_j + \left(\frac{\lambda}{A'_i \mathbf{1}} + \frac{(1-\lambda)x'_j \mathbf{1}_{x^{(r)}_j < x'_j}}{A'_i x'_c}\right) \left(1 - \sum_{j \in T} a_{ij}\right)\right) \geq 0 + \frac{\lambda}{s} + (1 - \lambda),\\
\frac{\partial P_u}{\partial y_i} &= \sum_{j \in [n] \setminus T: x^{(r)}_j \ge x'_j} a_{ij} \left(x^{(r)}_j + \left(\frac{\lambda}{A'_i \mathbf{1}} + \frac{(1-\lambda)x'_j \mathbf{1}_{x^{(r)}_j < x'_j}}{A'_i x'_c}\right) \left(1 - \sum_{j \in T} a_{ij}\right)\right) \leq 2 + \lambda + 0.
\end{align*}
Thus we have $\frac{\partial P_u}{\partial y_i} \leq \frac{2 + \lambda}{\lambda / s + 1 - \lambda} \cdot \frac{\partial P_c}{\partial y_i}$, so that
\[\frac{\partial P(r)}{\partial y_i} \leq \left(1 + \frac{2 + \lambda}{\frac{\lambda}{s} + 1 - \lambda}\right) \frac{\partial P_c}{\partial y_i} = O\left(\frac{1}{1-\lambda}\right) \frac{\partial P_c}{\partial y_i}.\]
We note that if $x'$ is feasible, then $A_i x' \ge 1$ for all $i \in [m]$, so $c^T x \leq O\left(\frac{1}{1-\lambda}\right) c^T x'$.

\paragraph{Bounding the number of violating constraints.} 
We show that Algorithm~\ref{alg:covering-box:phase} encounters $O(n \log s \log \frac{c^T x}{\alpha(1)})$ violating constraints. 
Using the same argument as in the proof of Theorem \ref{thm:covering-box}, we have that there are $O(\log (c^T x/\alpha(1)))$ phases.

In each phase $r$, when a violating  constraint $i$ just arrived, we have that $\sum_{j \in [n] \setminus T} a_{ij}x^{(r)}_j < 1 - \sum_{j \in T} a_{ij}$. We increment $x^{(r)}_j$'s where $j \in [n] \setminus T$ until $\sum_{j \in [n] \setminus T} a_{ij}x^{(r)}_j \ge 2\left(1 - \sum_{j \in T} a_{ij}\right)$. One of the following two cases must hold after updating $x^{(r)}$: (1) there exist a \emph{large} variable $x^{(r)}_j \geq 1/(2s)$ that is updated to at least $3x^{(r)}_j/2$, or (2) there exists a \emph{small} variable $x^{(r)}_j < 1/(2s)$ that becomes large, i.e., $x^{(r)}_j$ is updated to at least $1/(2s)$. Let $L$ and $S$ be the set of large and small variable subscript labels in $[n] \setminus T$ before the violating constraint $i$ arrives, respectively, and $\hat{x}^{(r)}_j$ be the value of $x^{(r)}_j$ after the update. If none of these two cases holds, then
\[\sum_{j \in [n] \setminus T} a_{ij} \hat{x}^{(r)}_j < \frac{3}{2}\sum_{j \in L}a_{ij}x^{(r)}_j + \frac{1}{2s} \sum_{j \in S}a_{ij} < \frac{3}{2}\left(1 - \sum_{j \in T} a_{ij}\right) + \frac{\left(1 - \sum_{k \in T} a_{ij}\right)}{2\sum_{k \in [n] \setminus T} a_{ij}} \sum_{j \in S}a_{ij} \le 2 \left(1 - \sum_{j \in T} a_{ij}\right)\]
where the second inequality is by the fact that constraint $i$ is violated and the definition of the sparsity $s$, while the last inequality is by $S \subseteq [n] \setminus T$.
This implies that the cost outside of the tight set does not exceed the remaining capacity by a factor of 2 after the update, a contradiction.

Suppose $x^{(r)}_j$ has been updated $t$ times by a multiplicative factor of 3/2 since it was large in phase $r$, then $t=O(\log s)$ since $\frac{1}{2s}(\frac{3}{2})^{t} \le 1$.

There are $n$ variables, each variable can be updated from small to large once and updated at most $t$ times by a factor of $3/2$ since it was large in each phase. Hence, Algorithm~\ref{alg:covering-box:phase} encounters
$O(n \log s \log \frac{c^T x}{\alpha(1)})$
violating covering constraints.
\end{proof}

\subsection{Applications} \label{sec:covering-app}

Our general framework on learning-augmented online covering LPs can be directly applied to online problems with predictions, including the online set cover problem and the online group Steiner tree problem on trees.

\subsubsection{Online fractional set cover with fractional advice} In the online set cover problem introduced in \cite{alon2009online}, we are given offline $n$ sets $S_j$ where $j \in [n]$, each associated with a positive cost $c_j$. The elements arrive online one at a time. Upon the arrival of an element $i \in [m]$ where $m$ can be unknown, the information of whether each set $S_j$ contains element $i$ is revealed. The goal is to irrevocably pick sets to cover all the elements that have arrived while minimizing the total cost. Let $F_i:=\{j \in [n] \mid i \in S_j\}$ denote the subset of subscript indices with the corresponding sets containing element $i$. Let $d:= \max_{i \in [m]}\{|F_i|\}$ denote the maximum number of sets that contain one specific element $i$, over all possible $i \in [m]$. An $O(\log m \log d)$-competitive algorithm was introduced in \cite{alon2009online}. The online algorithm first finds an $O(\log d)$-competitive fractional solution for the following LP formulation:
\begin{align}
  \begin{aligned}
    \text{minimize } & c^T x 
    \text{ over } x \in [0,1]^n 
    \text{ subject to } \sum_{j \in F_i} x_j \geq 1 & \forall i \in [m]
  \end{aligned}
                        \label{opt:set-cover}
\end{align}
and then round by paying a factor of $O(\log m)$. The online fractional set cover problem considers LP \eqref{opt:set-cover}, where each row is associated with an element that arrives and each column is associated with a set. The goal is to update $x$ in a non-decreasing manner while minimizing the objective $c^T x$.

The learning-augmented problem was later introduced in \cite{BamasMS20}, where the input also includes a confidence parameter $\lambda \in [0,1]$ and an advice $A \subseteq [n]$ presented as a subset of subscript indices of the sets given offline. The advice $A$ is feasible if $\cup_{j \in A} S_j = [m]$. Another way to present the advice is to describe it as a vector $x' \in \{0,1\}^n$, namely, set $S_j$ is suggested by the advice if and only if $x'_j = 1$. Our framework therefore naturally extends to a more general setting when the advice can be fractional, i.e., $x' \in [0,1]^n$. Let $\opt$ denote the optimal objective value for LP \eqref{opt:set-cover}. We show that the online fractional set cover problem with fractional advice can be solved efficiently (independent of the number of elements) with the same approximation guarantee as \cite{BamasMS20}.

\begin{corollary} \label{cor:set-cover}
For the learning-augmented online fractional set cover problem with fractional advice, there exists an online algorithm that generates $x$ and encounters polynomially many uncovered elements such that
\[c^T x \le \min\left\{O\left(\frac{1}{1-\lambda}\right) c^T x' + O(\log d) \opt, O\left(\log \frac{d}{\lambda}\right) \opt\right\}.\]
If $x'$ is feasible for LP \eqref{opt:set-cover}, then $c^T x \le O\left(\frac{1}{1-\lambda}\right) c^T x'$.
\end{corollary}

\begin{proof}
We use Theorem~\ref{thm:covering-box} and observe that the sparsity parameter $s = d$ since the entries of the constraint matrix is either 0 or 1. We note that the number of violating constraints encountered does not depend on $m$ but depends on the cost $c_j$ since $\alpha(1)$ depends on $c_j$.
\end{proof}

\subsubsection{Online group Steiner tree on trees} We consider the learning-augmented online group Steiner tree problem on trees. We note that this problem generalizes other online problems including the online group Steiner tree problem on general graphs (by paying another logarithmic factor in the competitive ratio), the online multicast problem on trees, and the online non-metric facility location problem \cite{alon2006general}.

In the group Steiner tree problem on trees, we are given a weighted rooted tree $T=(V,E,r)$ and groups of vertices $V_1, V_2, ..., V_k \subseteq V$. Let $n=|V|$ denote the number of vertices. Each edge $e \in E$ is associated with a positive cost $c_e$. The goal is to find a minimum weighted rooted subtree $T'=(V',E',r)$ that contains at least one vertex from each group $V_i$. An $O(\log n \log k)$-approximation algorithm was presented in \cite{garg2000polylogarithmic} by considering the following LP formulation:
\begin{align}
  \begin{aligned}
    \text{minimize } & \sum_{e \in E}c_e x_e 
    \text{ over } x \in [0,1]^{|E|} 
    \text{ subject to } x \text{ supports an $r$ to $V_i$ flow of value 1 } \forall i \in [k]
  \end{aligned}
                        \label{opt:steiner}
\end{align}
and round. We note that although LP \eqref{opt:steiner} might have exponentially many constraints, it can be solved in polynomial time by using a minimum cut procedure to retrieve violating constraints.

In the online problem, the groups $V_i$ arrive one at a time, and the goal is to find $T'$ by irrevocably adding edges from $E$. The $O(\log^2 n \log k)$-competitive algorithm in \cite{alon2006general} implicitly considers LP \eqref{opt:steiner} in an online fashion by  updating $x$ in a non-decreasing manner and rounding $x$ online.

In the learning-augmented problem, we are given a confidence parameter $\lambda \in [0,1]$ and a fractional advice $x' \in [0,1]^{|E|}$ indicating how likely each edge should be selected. Let $\opt$ denote the weight of the minimum weighted subtree rooted at $r$ that contains at least one vertex from each group $V_i$. We show the following.

\begin{corollary} \label{cor:steiner}
For the learning-augmented online group Steiner tree problem on trees, there exists a polynomial time randomized online algorithm that generates $T'=(V',E',r)$ such that
\[\sum_{e \in E'} c_e \le \log n \log k \cdot \min\left\{O\left(\frac{1}{1-\lambda}\right) c^T x' + O(\log n) \opt, O\left(\log \frac{n}{\lambda}\right) \opt\right\}\]
and $V' \cap V_i \neq \emptyset$ for each $i \in [k]$. If $x'$ is feasible, then $\sum_{e \in E'} c_e \le O\left(\frac{\log n \log k}{1-\lambda}\right) c^T x'$.
\end{corollary}

\begin{proof}
We use the following lemma due to \cite{alon2006general}.

\begin{lemma} \label{lem:steiner-rounding}
Given an $\alpha$-competitive feasible solution for LP \eqref{opt:steiner}, there exists an $O(\alpha \log n \log k)$-competitive randomized rounding scheme for the online group Steiner tree problem on trees.
\end{lemma}

We use Theorem~\ref{thm:covering-box} to obtain an online solution $x$ and observe that the sparsity parameter $s = O(|E|)=O(n^2)$. Let $\lp^*$ denote the optimal objective value of LP \eqref{opt:steiner}. We have that
\begin{align*}
c^T x &\le \min\left\{O\left(\frac{1}{1-\lambda}\right) c^T x' + O(\log n^2) \lp^*, O\left(\log \frac{n^2}{\lambda}\right) \lp^*\right\} \\
&\le \min\left\{O\left(\frac{1}{1-\lambda}\right) c^T x' + O(\log n) \lp^*, O\left(\log \frac{n}{\lambda}\right) \lp^*\right\}
\end{align*}
and $c^T x \le O\left(\frac{1}{1-\lambda}\right) c^T x'$ if $x'$ is feasible for LP \eqref{opt:steiner}.

Combining the inequality above with Lemma \ref{lem:steiner-rounding}, we have the desired guarantee for $T'$. \end{proof}

\section{PDLA Algorithms for Covering Semidefinite Programming} \label{sec:sdp}

In this section, we prove Theorems~\ref{thm:sdp-inf} and \ref{thm:sdp-box-inf}. Namely, we present efficient, consistent, and robust PDLA algorithms for online covering SDPs.
We recall that the online covering SDP problem \eqref{opt:sdp-covering} is the following.
\begin{align*}
  \begin{aligned}
    \text{minimize } & c^T x
    \text{ over } x \in \nnreals^n 
    \text{ subject to } \sum_{j=1}^n A_jx_j\succeq B^{(i)}.
  \end{aligned}
\end{align*}
Here, $A_j \in \R_{\geq 0}^{d \times d}$ is an SPSD matrix for each $j \in [n]$ and $c \in \R_{> 0}^n$ denotes the positive coefficients of the linear cost function. We assume that in the beginning, $B^{(0)}$ is a zero matrix, and there are $m$ arriving SPSD matrices $B^{(i)}$'s where $m$ might be unknown. For each $i \in [m]$, it is required that $B^{(i)} \succeq B^{(i-1)}$. In round $i$, $B^{(i)}$ arrives and our goal is to approximately minimize the objective $c^T x$ by updating $x$ in a non-decreasing manner such that $x$ is feasible. An important idea is to simultaneously consider the dual packing SDP problem \eqref{opt:sdp-packing}:
\begin{align*}
  \begin{aligned}
    \text{maximize } & B^{(i)} \otimes Y
    \text{ over $Y \succeq 0$ subject to } A_j\otimes Y\le c_j & \forall j \in[n].
  \end{aligned}
\end{align*}

The approach closely follows the one for online covering LPs and the one in \cite{EladKN16}. In round $i$, either $\sum_{j=1}^n A_jx_j\succeq B^{(i)}$ and there is nothing to be done, or we use a subroutine to retrieve an implicit violating \emph{linear} constraint and update $x$. More specifically, observe that $\sum_{j=1}^n A_jx_j \succeq B^{(i)}$ if and only if $v^T\left(\sum_{j=1}^n A_jx_j\right) v \ge v^T B^{(i)} v$ for all $v \in \R^n$, if the constraint in round $i$ is violated, we can find an SPSD matrix $V = v v^T$ such that $v^T\left(\sum_{j=1}^n A_jx_j\right) v = \sum_{j=1}^n A_j x_j \otimes V  < B^{(i)} \otimes V = v^T B^{(i)} v$, and we update $x_j$ proportionally to $A_j \otimes V$. There are many possible options for $V$ and without loss of generality, we can scale $V$ by a factor of $\trace(V)$ such that $\trace(V) = 1$. One option is to use $V=v v^T$ where $v$ is a unit vector corresponding to the smallest eigenvalue of $\sum_{j=1}^n A_jx_j - B^{(i)}$. Since $\sum_{j=1}^n A_jx_j - B^{(i)} \nsucceq 0$, its smallest eigenvalue $\lambda_d$ is negative, so $(\sum_{j=1}^n A_jx_j - B^{(i)}) \otimes V = v (\sum_{j=1}^n A_jx_j - B^{(i)}) v^T = \lambda_d < 0$ as required.

Again, we use the guess-and-double approach by increasing the primal and dual variables in each phase $r$. We estimate a lower bound $\alpha(r)$ for $\opt$ in phase $r$. In phase $1$, let $\alpha(1) \gets \min_{j=1}^n\{\frac{c_j\trace(B^{(1)})}{\trace(A_j)}\}$ be a proper lower bound for $\opt$. Once the online objective exceeds $\alpha(r)$, we start the new phase $r+1$ from the current $B^{(i)}$ matrix (say phase $r+1$ starts with $i = i_{r+1}$, in particular, $i_1 = 1$), and double the estimated lower bound, i.e. $\alpha(r+1) \gets 2\alpha(r)$. In the beginning of phase $r$, $x^{(r)}_j \gets \min\{x'_j, \alpha(r)/(2n c_j)\}$.

We recall that $x$ must be updated in a non-decreasing manner, so the algorithm maintains $\{x^{(r)}_j\}$, which denotes the value of each variable $x_j$ in each phase $r$, and the value of each variable $x_j$ is actually set to $\max_r\{x^{(r)}_j\}$.

In phase $r$ and round $i$, we find $V$ and increase $Y$ continuously along $V$. If the advice is feasible, i.e., $\sum_{j=1}^n A_jx'_j \succeq B^{(i)}$, we increase each variable $x^{(r)}_j$ with growth rate
\[\frac{A_j \otimes V}{c_j} \left(x^{(r)}_j + \frac{\lambda B^{(i)} \otimes V}{\sum_{k=1}^n A_k\otimes V}+\frac{(1-\lambda)x'_j\mathbf{1}_{x^{(r)}_j < x'_j}B^{(i)} \otimes V}{\sum_{k=1}^n \mathbf{1}_{x^{(r)}_k < x'_k}A_kx'_k\otimes V}\right)\]
where $\mathbf{1}_{x^{(r)}_k < x'_k}$ is an indicator variable with value 1 if $x^{(r)}_k < x'_k$ and 0 otherwise. Alternatively if the advice is not feasible, we increase $x^{(r)}_j$ with growth rate
\[\frac{A_j \otimes V}{c_j} \left(x^{(r)}_j + \frac{B^{(i)} \otimes V}{\sum_{k=1}^n A_k\otimes V}\right).\]
We note that 
\begin{equation} \label{eq:sdp-sum-to-1}
\sum_{j=1}^n{\frac{A_j \otimes V}{\sum_{k=1}^n A_k\otimes V}} = 1 \text { and } \sum_{j=1}^n\frac{A_j \otimes V \left(x'_j\mathbf{1}_{x^{(r)}_j < x'_j}\right)}{\sum_{k=1}^n \mathbf{1}_{x^{(r)}_k < x'_k}A_kx'_k\otimes V} = 1.
\end{equation}

To implement with the proper growth rate, the dual variable $Y$ initialized as a zero matrix is used as a proxy in each phase $r$. We increment $x^{(r)}$ until the implicit violating constraint $i$ is satisfied \emph{by a factor of 2}. This guarantees that at least one variable $x^{(r)}_j$ is doubled thus ensures that the algorithm only updates $x^{(r)}$ polynomially many times.

Similar to Algorithm~\ref{alg:covering}, in phase $r$, each iteration ends in one of the following three cases: (1) the implicit violating linear constraint is satisfied by a factor of 2, (2) there exists a variable $x^{(r)}_j$ that reaches the advice value $x'_j$, or (3) the objective $c^T x^{(r)}$ reaches $\alpha(r)$. The coefficients $D_j$ and $B_j$ for fitting boundary conditions are defined based on the value of $x^{(r)}_j$ in the end of the previous iteration, stored as $\bar{x}_j$. The indicator value used for the growth rate, $\mathbf{1}_{x^{(r)}_j < x'_j}$, has the same value as $\mathbf{1}_{\bar{x}_j < x'_j}$ during an iteration.

The main algorithm that uses the phase scheme is presented in Algorithm~\ref{alg:sdp:phase}. The continuous primal-dual approach in phase $r$ and round $i$, used as a subroutine, is presented in Algorithm~\ref{alg:sdp}.

\begin{algorithm}[!htb]
\caption{Phase Scheme for Algorithm~\ref{alg:sdp}} \label{alg:sdp:phase}
\begin{algorithmic}[1]
\State{$r \gets 1$, $\alpha(1)\gets\min_{j \in [n]}\{c_j\trace(B^{(1)}) / \trace(A_j)\}$, $Y \gets 0$, $i_1 \gets 1$.}
\For{each $j \in [n]$}
    \State $x^{(r)}_j \gets \min\{x'_j, \alpha(r)/(2n c_j)\}$.
\EndFor
\For{arriving covering constraint $\sum_{j=1}^n A_jx_j\succeq B^{(i)}$} \Comment{$i=1, 2, ..., m$ for an unknown $m$}
\State{Run Algorithm~\ref{alg:sdp}.} \label{line:sdp-alg}
\If {$c^T x^{(r)} \ge \alpha(r)$} \Comment{start a new phase}
    \State $r \gets r+1$, $\alpha(r) \gets 2\alpha(r-1)$, $Y \gets 0$, $i_r \gets i$.
    \For{each $j \in [n]$}
    \State $x^{(r)}_j \gets \min\{x'_j, \alpha(r)/(2n c_j)\}$.
    \EndFor
    \State Go to line \ref{line:sdp-alg}.
\EndIf
\For{each $j \in [n]$}
    \State $x_j \gets \max_{\ell \in [r]} \{x_j^{(\ell)}\}$. \Comment{this is the solution returned in each round $i$}
\EndFor
\EndFor
\end{algorithmic}
\end{algorithm}

\begin{algorithm}[!htb]
\caption{PDLA Online SDP Covering in Phase $r$ and Round $i$} \label{alg:sdp}

\renewcommand{\algorithmicrequire}{\textbf{Input:}}
\renewcommand{\algorithmicensure}{\textbf{Output:}}

\algorithmicrequire{ $x^{(r)}$: current solution, $\alpha(r)$: estimate for $\opt$, $B^{(i)}$: current lower bound matrix, $i_r$: starting round of phase $r$, $Y$: dual variable matrix, $\lambda$: the confidence parameter, and $x'$: the advice for $x$.}

\algorithmicensure{ Updated $x^{(r)}$ and $Y$.}

\begin{algorithmic}[1]

    \State $\bar{x} \gets x^{(r)}.$ \label{line:sdp-restore} \Comment{$\bar{x}$ is the value of $x^{(r)}$ at the end of the previous iteration}
    \While{$\sum_{j=1}^n A_jx_j\not\succeq B^{(i)}$} \label{line:sdp-while-outer}
        \State{Find an SPSD matrix $V$ with $\trace(V)=1$ such that $\sum_{j=1}^n (A_jx^{(r)}_j)\otimes V<B^{(i)}\otimes V$.}
        \For {each $j \in [n]$}
        \If{$\sum_{k=1}^n A_kx'_k \succeq B^{(i)}$}
            \State{$D_j\gets\frac{\lambda B^{(i)} \otimes V}{\sum_{k=1}^n A_k\otimes V}+\frac{(1-\lambda)x'_j\mathbf{1}_{\bar{x}_j < x'_j}B^{(i)} \otimes V}{\sum_{k=1}^n \mathbf{1}_{\bar{x}_k < x'_k}A_kx'_k\otimes V}$,}
        \Else
            \State{$D_j\gets\frac{B^{(i)} \otimes V}{\sum_{k=1}^n A_k \otimes V}$.}
        \EndIf

        \State{$B_j\gets\frac{\bar{x}_j+D_j}{\exp\left(\frac{A_j\otimes Y}{c_j}\right)}$.}
        \EndFor
        \While{$\sum_{j=1}^n A_j x^{(r)}_j \otimes V<2B^{(i)}\otimes V$} \label{line:sdp-while-inner}
       \State{Set $\delta=0$ and increase it continuously.}
       \State{Increase $Y$ by continuously adding $V\delta$ to $Y$.}
       \State{Increase $x^{(r)}$ continuously by simultaneously setting \[x_j^{(r)}\gets B_j\,\exp\left(\frac{A_j\otimes Y}{c_j}\right)-D_j.\]} \label{line:sdp-inc}
       \If{any $x^{(r)}_j = x'_j$}
                    \State Break and go to line \ref{line:sdp-restore}.
        \EndIf
        \If{$c^T x^{(r)} \ge \alpha(r)$}
                \State Break and return.
            \EndIf
       \EndWhile
       \EndWhile
\end{algorithmic}
\end{algorithm}

The augmentation is in a continuous fashion, and it is not hard to implement it in a discrete way for any desired precision by binary search. In each iteration, we use binary search to find the proper $\delta$, and $x^{(r)}$ is updated accordingly once. Therefore, to show that the algorithm is efficient, it suffices to bound the number of iterations. The performance of Algorithm \ref{alg:sdp:phase} is stated in Theorem \ref{thm:sdp}, the formal version of Theorem \ref{thm:sdp-inf}.

\begin{theorem}
\label{thm:sdp}
For the learning-augmented online covering SDP problem, there exists an online algorithm that generates $x$ such that
\[c^Tx \le \min\left\{O\left(\frac{1}{1-\lambda}\right)c^Tx' +  O\left(\log\left(\kappa n\right)\right)\opt,O\left(\log \frac{\kappa n}{\lambda}\right)\opt\right\}\]
and $x$ is updated $O(n (\log \frac{c^T x}{\alpha(1)})(\log n + \log c^T x + \log \beta))$ times. 
If $x'$ is feasible for SDP \eqref{opt:sdp-covering}, then $c^Tx \le O(\frac{1}{1-\lambda})c^Tx'$.
Here, $\beta=\max_{j \in [n]}\{a^{\max}_j/c_j\}$, $a^{\max}_j:=\max_{i \in [m]}\{a_{ij} \mid a_{ij} > 0\}$, and
\[\kappa:=\frac{\max_{j \in [n], i \in [m]}\{\lambda_{\max}(A_j), \lambda_{\max}(B^{(i)})\}}{\min_{j \in [n], i \in [m]}\{\lambda_{\min}(A_j), \lambda_{\min}(B^{(i)})\}}\]
where $\lambda_{\max}(M)$ and $\lambda_{\min}(M)$ are the positive maximum and minimum eigenvalues of a PSD matrix $M$, respectively.
\end{theorem}

\begin{proof}
Let $P(r)=c^T x^{(r)}$ and $D(r)=B^{(i)} \otimes Y$ be the objective value of the primal and the dual in phase $r$, respectively. We use Algorithm~\ref{alg:sdp:phase} and first show robustness, i.e., $c^T x \le O\left(\log\frac{\kappa n}{\lambda}\right) \opt$, then show consistency, i.e., $c^T x \le O\left(\frac{1}{1-\lambda}\right) c^T x' + O(\log(\kappa n))\opt$. Finally, we show that only $O(n (\log \frac{c^T x}{\alpha(1)})(\log n + \log c^T x + \log \beta))$ iterations are needed.

\paragraph{Robustness.} To show that $c^T x \le O\left(\log\frac{\kappa n}{\lambda}\right) \opt$, we prove the following five claims:
\begin{enumerate} [label=(\roman*)]
\item \label{sdp1} $x$ is feasible for SDP \eqref{opt:sdp-covering}.
\item \label{sdp2} For each \emph{finished} phase $r$, $\alpha(r) \leq 6 D(r)$.
\item \label{sdp3} $Y/\Theta\left(\log\frac{\kappa n}{\lambda}\right)$ is feasible for SDP \eqref{opt:sdp-packing} in each phase.
\item \label{sdp4} The sum of the covering objective generated from phase 1 to $r$ is at most $2 \alpha(r)$.
\item \label{sdp5} Let $r'$ be the last phase, then the covering objective $c^T x \leq 2 \alpha(r')$.
\end{enumerate}

The proofs of \ref{sdp4} and \ref{sdp5} are the same as the ones in Theorem~\ref{thm:covering}, hence omitted. Equipped with these claims and weak duality \eqref{eq:sdp-pd}, we have that
\[
c^T x \le \Theta(1) \alpha(r') \le \Theta(1) D(r') = O\left(\log\frac{\kappa n}{\lambda}\right) \opt.
\]

\begin{proof}[Proof of \ref{sdp1}.]
We prove that $x^{(r)}$ is feasible in phase $r$ by showing that the growing functions in Algorithm \ref{alg:sdp} line \ref{line:sdp-inc} increments $x^{(r)}$ in a continuous manner.

In the beginning of an iteration in round $i$, we have that
\[x^{(r)}_j = B_j\exp\left(\frac{A_j\otimes Y}{c_j}\right)-D_j = \frac{\bar{x}_j + D_j}{\exp\left(\frac{A_j\otimes Y}{c_j}\right)} \exp\left(\frac{A_j\otimes Y}{c_j}\right) - D_j = \bar{x}_j.\]
By following Algorithm \ref{alg:sdp}, we have that $x^{(r)}$ is feasible since we terminate the while loop at line \ref{line:sdp-while-outer} when there does not exists an SPSD matrix $V$ such that $\sum_{j=1}^n A_j x_j \otimes V < B^{(i)} \otimes V$, indicating that $\sum_{j=1}^n A_jx_j \succeq B^{(i)}$. Whenever we find such a $V$, it induces a linear constraint and we increment $x^{(r)}$ until it is satisfied by a factor of 2 as in Algorithm \ref{alg:sdp} line \ref{line:sdp-while-inner}. As $x$ is the coordinate-wise maximum of $\{x^{(r)}\}$ and the sum of matrices is still PSD, $x$ must be feasible for SDP \eqref{opt:sdp-covering}.
\end{proof}

\begin{proof}[Proof of \ref{sdp2}.]
In the beginning of phase $r$, $x^{(r)}_j = \min \{x'_j, \alpha(r) / (2nc_j)\}$, so $P(r)$ is initially at most $\alpha(r) / 2$. The total increase of $P(r)$ is at least $\alpha(r) / 2$ as $P(r) \geq \alpha(r)$ when phase $r$ ends. Therefore, it suffices to show that
$$\frac{\partial P(r)}{\partial \delta} \leq 3 \frac{\partial D(r)}{\partial \delta}.$$
By considering the partial derivative of $P(r)$ with respect to $\delta$, when $\sum_{j=1}^n A_jx'_j \succeq B^{(i)}$, we have that
\begin{align*}
\frac{\partial P(r)}{\partial \delta} &= \sum_{j=1}^n c_j \frac{\partial x^{(r)}_j}{\partial \delta}\\
&= \sum_{j=1}^n c_j\frac{\partial}{\partial\delta}\left( B_j\,\exp\left(\frac{A_j\otimes (Y+V\delta)}{c_j}\right)-D_j\right)\\
&= \sum_{j=1}^n c_j\left(\frac{A_j\otimes V}{c_j}\right)B_j\,\exp\left(\frac{A_j\otimes (Y+V\delta)}{c_j} \right) \\
&= \sum_{j=1}^n (A_j\otimes V)\left(x_j^{(r)}+D_j\right)\\
&= \sum_{j=1}^n (A_j\otimes V)\left(x_j^{(r)}+\frac{\lambda B^{(i)} \otimes V}{\sum_{k=1}^n A_k\otimes V}+\frac{(1-\lambda)x'_j\mathbf{1}_{x^{(r)}_j < x'_j} B^{(i)} \otimes V}{\sum_{k=1}^n \mathbf{1}_{x^{(r)}_k < x'_k}A_kx'_k\otimes V}\right)\\
&\le \left(2 + \lambda + (1-\lambda)\right) B^{(i)} \otimes V = 3 \frac{\partial D(r)}{\partial \delta}
\end{align*}
where the last inequality is due to the fact that $\sum_{j=1}^n A_j x_j \otimes V<2B^{(i)}\otimes V$ and \eqref{eq:sdp-sum-to-1}. The same result can be obtained when $\sum_{j=1}^n A_jx'_j \nsucceq B^{(i)}$ by regarding $\lambda$ as 1.
\end{proof}

\begin{proof}[Proof of \ref{sdp3}.]
$Y$ is the sum of matrices $\delta V$, where $\delta>0$ and $V$ is SPSD since $V=vv^T$. 
Hence, $\delta V$ is SPSD. $Y$ is the sum of SPSD matrices and is therefore itself SPSD. 

It remains to show that $Y/\Theta(\log\frac{\kappa n}{\lambda})$ is feasible for SDP \eqref{opt:sdp-packing} after satisfying the last implicit linear constraint in phase $r$. We recall that within a phase, each iteration ends in one of the following two cases: (1) the implicit violating linear constraint is satisfied by a factor of 2, or (2) there exists a variable $x^{(r)}_j$ that reaches the advice value $x'_j$. Suppose there are $t$ iterations in phase $r$. In iteration $p \le t$, let $x^{(r,p)}_j$, $B^{(p)}_j$, and $D^{(p)}_j$ be the value of $x^{(r)}_j$, $B_j$, and $D_j$ in the end of iteration $p$, respectively. Additionally, let $Y^{(p)}$ be the value of $Y$ in the end of iteration $p$. We have $x_j^{(r,p)}=B^{(p)}_j\,\exp\left(\frac{A_j\otimes Y^{(p)}}{c_j}\right)-D^{(p)}_j$. 
By using a similar argument in Theorem \ref{thm:covering} \ref{gpd3}, we have that
\begin{align*}
\frac{A_j\otimes Y^{(t)}}{c_j}&=\ln\frac{x_j^{(r,t)}+D^{(t)}_j}{B^{(t)}_j} =\ln\left(\frac{x_j^{(r,t)}+D^{(t)}_j}{x_j^{(r,t-1)}+D^{(t)}_j}\exp\left(\frac{A_j\otimes Y^{(t-1)}}{c_j}\right)\right) \\
&= \frac{A_j\otimes Y^{(t-1)}}{c_j} +  \ln\frac{x_j^{(r,t)}+D^{(t)}_j}{x_j^{(r,t-1)}+D^{(t)}_j} = ... \\
&= \ln\prod_{p=1}^t\frac{x_j^{(r,p)}+D^{(p)}_j}{x_j^{(r,p-1)}+D^{(p)}_j} \le \ln\frac{x_j^{(r,p)}+D^{min}_j}{x_j^{(r,0)}+D^{min}_j}
\end{align*}
where $D^{min}_j := \min_{p \in [t]} \{D^{(p)}_j\}$. 

Now we upper-bound $x^{(r,p)}_j$ and lower-bound $D^{min}_j$ to obtain the robustness ratio. Notice that $x^{(r,p)}_j \le 2\kappa$ since whenever $\sum_{i=1}^n A_i x^{(r)}_i \otimes V < 2 B^{(p)} \otimes V$, if we only consider $x^{(r,p)}_j$, the worst case is that the unit vector $v$ such that $V=v v^T$, is in the same direction of the vector corresponding to the largest eigenvalue of $B^{(p)}$ and that of the smallest eigenvalue of $A_i$, and $x^{(r,p)}$ is incremented until the implicit linear constraint is satisfied by a factor of 2. We also have that
\[D^{min}_j \ge \frac{\lambda \min_{i \in [m]}\{\lambda_{\min}(B^{(i)})\}}{\sum_{k=1}^n \lambda_{\max}(A_i)} \ge  \frac{\lambda}{\kappa n}.\]
With $x^{(r,0)}_j = \alpha(r)/(2nc_j) \geq 0$, we thus have that
$$\ln \frac{x^{(r,p)}_j + D^{min}_j}{x^{(r,0)}_j + D^{min}_j} \leq \ln \left(1 + \frac{2\kappa}{D^{min}_j}\right) = O\left(\log \frac{\kappa^2n}{ \lambda}\right) = O\left(\log \frac{\kappa n}{ \lambda} + \log \kappa \right) = O\left(\log \frac{\kappa n}{ \lambda}\right).$$
\end{proof}

\paragraph{Consistency.} We then show that $c^T x \leq O\left(\frac{1}{1-\lambda}\right) c^T x' + O\left(\log \left(\kappa n\right)\right)\opt$. Suppose the algorithm is in phase $r$ with constraint $i$ arriving. If $\sum_{j=1}^n A_jx'_j \nsucceq B^{(i)}$, then the change of $c^T x$ simply follows \ref{sdp2} and \ref{sdp3} by regarding $\lambda$ as 1, so $c^T x \le O\left(\log \left(\kappa n\right)\right)\opt$. For the more interesting case when $\sum_{j=1}^n A_jx'_j \succeq B^{(i)}$, we decompose $\frac{\partial P(r)}{\partial y_i}=\frac{\partial P_c}{\partial y_i}+\frac{\partial P_u}{\partial y_i}$, where $P_c$ is the component of the primal objective due to the advice and $P_u$ is the component of the primal objective due to the online algorithm. 
We have that the rate of change is credited to $\partial P_c$ if $x^{(r)}_j < x'_j$ and the rate of change is credited to $\partial P_u$ otherwise, if $x^{(r)}_j \ge x'_j$. 
In particular,
\begin{align*}
\frac{\partial P_c}{\partial\delta} &= \sum_{j\in[n]: x_j < x'_j} (A_j\otimes V)\left(x_j^{(r)}+\frac{\lambda B^{(i)} \otimes V}{\sum_{k=1}^n A_k\otimes V}+\frac{(1-\lambda)x'_j\mathbf{1}_{x^{(r)}_j < x'_j} B^{(i)} \otimes V}{\sum_{k=1}^n \mathbf{1}_{x^{(r)}_k < x'_k}A_kx'_k\otimes V}\right) \\
&\geq 0 + \left(\frac{\lambda}{\kappa n} + 1 - \lambda\right) B^{(i)} \otimes V,\\
\frac{\partial P_u}{\partial\delta} &= \sum_{j\in[n]: x_j \ge x'_j} (A_j\otimes V)\left(x_j^{(r)}+\frac{\lambda B^{(i)} \otimes V}{\sum_{k=1}^n A_k\otimes V}+\frac{(1-\lambda)x'_j\mathbf{1}_{x^{(r)}_j < x'_j} B^{(i)} \otimes V}{\sum_{k=1}^n \mathbf{1}_{x^{(r)}_k < x'_k}A_kx'_k\otimes V}\right) \leq (2 + \lambda) B^{(i)} \otimes V + 0.
\end{align*}
Thus we have $\frac{\partial P_u}{\partial \delta} \leq \frac{(2 + \lambda) B^{(i)} \otimes V}{(\lambda/(\kappa n) + 1 - \lambda) B^{(i)} \otimes V} \cdot \frac{\partial P_c}{\partial \delta}$, so that
\[\frac{\partial P(r)}{\partial \delta} \leq \left(1+\frac{2 + \lambda}{1-\lambda}\right) \frac{\partial P_c}{\partial \delta} = O\left(\frac{1}{1-\lambda}\right) \frac{\partial P_c}{\partial \delta}.\]
We note that if $x'$ is feasible, then $\sum_{j=1}^n A_jx'_j \succeq B^{(i)}$, so $c^T x \leq O\left(\frac{1}{1-\lambda}\right) c^T x'$.

\paragraph{Bounding the number of iterations.} The analysis simply follows the one for Theorem \ref{thm:covering}. The main difference is that (1) instead of violating linear covering constraints, we use implicit violating linear constraints, and (2) we bound the number of iterations instead of implicit violating constraints, where the number of iterations is $O(n)$ more than the number of implicit violating constraints since each variable can reach the advice value once.
\end{proof}

\subsection{Adding Box Constraints} \label{subsec:sdp-box}

We recall that the covering SDP problem \eqref{opt:sdp-covering-box} with box constraints is the following.
\begin{align*}
  \begin{aligned}
    \text{minimize } & c^T x
    \text{ over } x \in [0,1]^n 
    \text{ subject to } \sum_{j=1}^n A_jx_j\succeq B^{(i)}.
  \end{aligned}
\end{align*}
Our PDLA algorithm simultaneously considers the dual packing SDP problem:
\begin{align}
  \begin{aligned}
    \text{maximize } & B^{(i)} \otimes Y - \mathbf{1}^T z
    \text{ over } Y\succeq 0 \text{ and } z \in \nnreals^n
    \text{ subject to } A_j \otimes Y - z_j\le c_j & \forall j \in[n].
  \end{aligned}
                        \label{opt:sdp-packing-box}
\end{align}
We recall that we assume that the advice $x' \in [0,1]^n$.

We use the guess-and-double approach similar to Algorithms~\ref{alg:covering-box:phase} and \ref{alg:sdp:phase}. We maintain the set $T$ which denotes the subscript indices of the $x$ variables that are tight. In phase $r$, upon the arrival of an implicit violating linear constraint, we increment $x^{(r)}_j$ in terms of an exponential function of $\delta$ and $z_j$ subject to $x^{(r)}_j \le 1$. Once $x^{(r)}_j = 1$, we add $j$ to the tight set $T$ and stop incrementing $x^{(r)}_j$, but we still increment $Y$ continuously and $z_j$ with rate $A_j \otimes V \delta$ in order to maintain dual feasibility. In the beginning of each phase, $z$ is a zero vector and $T$ is an empty set. $z_j = 0$ whenever $j \in [n] \setminus T$. Subject to $x^{(r)}_j \le 1$, $x^{(r)}_j$ is increasing until the cost \emph{outside of the tight set} exceeds the \emph{remaining capacity} by a factor of 2. More specifically, when we have an implicit violating linear constraint induced by $V$, we have that $\sum_{j \in [n] \setminus T} A_j x^{(r)}_j \otimes V < B^{(i)} \otimes V - \sum_{j \in T} A_j \otimes V$, and we increment until $\sum_{j \in [n] \setminus T} A_j x^{(r)}_j \otimes V \ge 2(B^{(i)} \otimes V - \sum_{j \in T} A_j \otimes V)$.

If the advice is feasible, i.e., $\sum_{j=1}^n A_jx'_j \succeq B^{(i)}$, we increase each variable $x^{(r)}_j$ with growth rate
\[\frac{A_j \otimes V}{c_j} \left(x^{(r)}_j + \left(\frac{\lambda}{\sum_{k \in [n] \setminus T} A_k\otimes V}+\frac{(1-\lambda)x'_j\mathbf{1}_{x^{(r)}_j < x'_j}}{\sum_{k \in [n] \setminus T} \mathbf{1}_{x^{(r)}_k < x'_k}A_kx'_k\otimes V}\right)\left(B^{(i)} \otimes V -  \sum_{k\in T} A_k\otimes V \right)\right).\]
Alternatively if the advice is not feasible, we increase $x^{(r)}_j$ with growth rate
\[\frac{A_j \otimes V}{c_j} \left(x^{(r)}_j + \frac{B^{(i)} \otimes V -  \sum_{k\in T} A_k\otimes V}{\sum_{k \in [n] \setminus T} A_k\otimes V}\right).\]
Similar to \eqref{eq:sdp-sum-to-1}, we have that 
\begin{equation} \label{eq:sdp-box-sum-to-1}
\sum_{j \in [n] \setminus T}{\frac{A_j \otimes V}{\sum_{k \in [n] \setminus T} A_k \otimes V}} = 1 \text { and } \sum_{j \in [n] \setminus T}\frac{x'_j\mathbf{1}_{x^{(r)}_j < x'_j} A_j \otimes V}{\sum_{k \in [n] \setminus T} \mathbf{1}_{x^{(r)}_k < x'_k}A_kx'_k\otimes V} = 1.
\end{equation}

Each iteration now ends in one of the following four cases: (1) the implicit linear constraint induces a cost outside of the tight set that exceeds the remaining capacity by a factor of 2, (2) there exists a variable $x^{(r)}_j$ that reaches the advice value $x'_j$, (3) there exists a variable $x^{(r)}_j$ that reaches 1, or (4) the objective $c^T x^{(r)}$ reaches $\alpha(r)$. The coefficients $D_j$ and $B_j$ for fitting boundary conditions are defined based on the value of $x^{(r)}_j$ in the end of the previous iteration, stored as $\bar{x}_j$. Again, the indicator value used for the growth rate, $\mathbf{1}_{x^{(r)}_j < x'_j}$, has the same value as $\mathbf{1}_{\bar{x}_j < x'_j}$ during an iteration.

The main algorithm that uses the phase scheme is presented in Algorithm~\ref{alg:sdp-box:phase}. The continuous primal-dual approach in phase $r$ and round $i$, used as a subroutine, is presented in Algorithm~\ref{alg:sdp-box}. The performance of Algorithm \ref{alg:sdp-box:phase} is stated in Theorem \ref{thm:sdp-box}, the formal version of Theorem \ref{thm:sdp-box-inf}. We note that in line \ref{line:sdp-box:find-V} of Algorithm \ref{alg:sdp-box}, given the round $i$ and the tight set $T$, we find $V$ such that $\frac{\sum_{j\in[n] \setminus T} A_j\otimes V}{B^{(i)} \otimes V - \sum_{j\in T} A_j\otimes V}$ is minimized.\footnote{Following \cite{EladKN16}, for fixed $i$ and $T$, $V$ can be calculated by estimating $\eta$ to any desired precision via solving the SDP:
\[\min_V\left\{\left(\sum_{j\in[n] \setminus T} A_j - \eta \left(B^{(i)} - \sum_{j\in T} A_j\right)\right)\otimes V \mid V \succeq 0, \left(B^{(i)} - \sum_{j\in T} A_j\right) \otimes V \ge 0 \right\}.\]}

\begin{algorithm}[!htb]
\caption{Phase Scheme for Algorithm~\ref{alg:sdp-box}} \label{alg:sdp-box:phase}
\begin{algorithmic}[1]
\State{$r \gets 1$, $\alpha(1)\gets\min_{j \in [n]}\{c_j\trace(B^{(1)}) / \trace(A_j)\}$, $Y \gets 0$, $T \gets \emptyset$, $i_1 \gets 1$.}
\For{each $j \in [n]$}
    \State $x^{(r)}_j \gets \min\{x'_j,\alpha(r)/(2nc_j)\}$, $z_j \gets 0$.
\EndFor
\For{arriving covering constraint $\sum_{j=1}^n A_jx_j\succeq B^{(i)}$} \Comment{$i=1, 2, ..., m$ for an unknown $m$}
\State{Run Algorithm~\ref{alg:sdp-box}.} \label{line:sdp-box-alg}
\If {$c^T x^{(r)} \ge \alpha(r)$} \Comment{start a new phase}
    \State $r \gets r+1$, $\alpha(r) \gets 2\alpha(r)$, $Y \gets 0$, $T \gets \emptyset$, $i_r \gets i$.
    \For{each $j \in [n]$}
    \State $x^{(r)}_j \gets \min\{x'_j,\alpha(r)/(2nc_j)\}$, $z_j \gets 0$.
    \EndFor
    \State Go to line \ref{line:sdp-box-alg}.
\EndIf
\For{each $j \in [n]$}
    \State $x_j \gets \max_{\ell \in [r]} \{x_j^{(\ell)}\}$. \Comment{this is the solution returned in each round $i$}
\EndFor
\EndFor
\end{algorithmic}
\end{algorithm}

\begin{algorithm}[!htb]
\caption{PDLA Online SDP Covering with Box Constraints in Phase $r$ and Round $i$} \label{alg:sdp-box}

\renewcommand{\algorithmicrequire}{\textbf{Input:}}
\renewcommand{\algorithmicensure}{\textbf{Output:}}

\algorithmicrequire{ $x^{(r)}$: current solution, $\alpha(r)$: estimate for $\opt$, $B^{(i)}$: current lower bound matrix, $i_r$: starting round of phase $r$, $Y$: dual variable matrix, $z$: dual variable vector, $T$: tight variable set, $\lambda$: the confidence parameter, and $x'$: the advice for $x$.}

\algorithmicensure{ Updated $x^{(r)}$, $Y$, and $z$.}

\begin{algorithmic}[1]

    \State $\bar{x} \gets x^{(r)}.$ \label{line:sdp-box-restore} \Comment{$\bar{x}$ is the value of $x^{(r)}$ at the end of the previous iteration}
    \While{$\sum_{j=1}^n A_jx_j\not\succeq B^{(i)}$} \label{line:sdp-box-while-outer}
        \State{Find an SPSD matrix $V$ with $\trace(V)=1$ such that $\sum_{j=1}^n (A_jx^{(r)}_j)\otimes V<B^{(i)}\otimes V$.} \label{line:sdp-box:find-V}
        \For{each $j \in [n]$}
        \If{$\sum_{k=1}^n A_kx'_k \succeq B^{(i)}$}
            \State{$D_j\gets\left(\frac{\lambda}{\sum_{k \in [n] \setminus T} A_k\otimes V}+\frac{(1-\lambda)x'_j\mathbf{1}_{\bar{x}_j < x'_j}}{\sum_{k \in [n] \setminus T} \mathbf{1}_{\bar{x}_k < x'_k}A_kx'_k\otimes V}\right)\left(B^{(i)} \otimes V -  \sum_{k\in T} A_k\otimes V \right)$,}
        \Else
            \State{$D_j\gets\frac{B^{(i)} \otimes V - \sum_{k\in T} A_k\otimes V}{\sum_{k\in[n] \setminus T} A_k\otimes V}$.}
        \EndIf

        \State{$B_j\gets\frac{\bar{x}_j+D_j}{\exp\left(\left(A_j\otimes Y - z_j\right) / c_j\right)}$.}
        \EndFor

        \While {$T \neq [n]$ and $\sum_{j \in [n] \setminus T} A_j x^{(r)}_j \otimes V < 2(B^{(i)} \otimes V - \sum_{j \in T} A_j \otimes V)$} \label{line:sdp-box-while-inner}
       \State{Set $\delta=0$ and increase it continuously.}
       \State{Increase $Y$ by continuously adding $V\delta$ to $Y$.}
       \For{each $j \in T$}
            \State Increase $z_j$ with rate $A_j \otimes V \delta$.
        \EndFor
       \State{Update $x^{(r)}$ continuously by simultaneously setting \[x_j^{(r)}\gets \min\left\{1, B_j\,\exp\left(\frac{A_j\otimes Y - z_j}{c_j}\right)-D_j\right\}.\]} \label{line:sdp-box-inc}
       \If{any $x^{(r)}_j=x'_j < 1$}
                    \State Break and go to line \ref{line:sdp-box-restore}.
        \EndIf
                \If{any $x^{(r)}_j=1$ for $j \notin T$}
                    \State Add $j$ to $T$ and go to line \ref{line:sdp-box-restore}.
                    \EndIf
                \If{$c^T x^{(r)} \ge \alpha(r)$}
                \State Break and return.
            \EndIf
       \EndWhile
       \If {$T = [n]$ and $\sum_{j=1}^n A_jx_j\not\succeq B^{(i)}$}
        \State \Return no feasible solution. \label{line:sdp-box-infeasible}
    \EndIf
       \EndWhile
\end{algorithmic}
\end{algorithm}

\begin{theorem}
\label{thm:sdp-box}
For the learning-augmented online covering SDP problem with box constraints, there exists an online algorithm that generates $x$ such that
\[c^Tx \le \min\{O(\frac{1}{1-\lambda})c^Tx' +  O(\log s)\opt,O(\log \frac{s}{\lambda} )\opt\}\]
and $x$ is updated $O(n \log s \log \frac{c^T x}{\alpha(1)})$ times. If $x'$ is feasible for SDP \eqref{opt:sdp-covering-box}, then $c^Tx \le O(\frac{1}{1-\lambda})c^Tx'$.
Here,
\[s:=\max_{
i \in [m], T \subseteq [n]}\min_{
V: B^{(i)} \otimes V - \sum_{j\in T} A_j\otimes V > 0
}\left\{\frac{\sum_{j\in[n] \setminus T} A_j\otimes V}{B^{(i)} \otimes V - \sum_{j\in T} A_j\otimes V}\right\}.\]
\end{theorem}

\begin{proof}

Let $P(r)=c^T x^{(r)}$ and $D(r)=B^{(i)} \otimes Y - \mathbf{1}^T z$ be the objective value of the primal and the dual in phase $r$, respectively. We use Algorithm~\ref{alg:sdp-box:phase} and first show robustness, i.e., $c^T x \le O\left(\log\frac{s}{\lambda}\right) \opt$, then show consistency, i.e., $c^T x \le O\left(\frac{1}{1-\lambda}\right) c^T x' + O(\log s)\opt$. Finally, we show that only $O(n \log s \log \frac{c^T x}{\alpha(1)})$ iterations are needed.

\paragraph{Robustness.}
To show that $c^T x \le O\left(\log\frac{s}{\lambda}\right) \opt$, we prove the following five claims:
\begin{enumerate} [label=(\roman*)]
\item \label{sdp1-box} If line \ref{line:sdp-box-infeasible} in Algorithm \ref{alg:sdp-box} is not encountered, then $x$ is feasible for SDP \eqref{opt:sdp-covering-box}.
\item \label{sdp2-box} For each \emph{finished} phase $r$, $\alpha(r) \leq 6 D(r)$.
\item \label{sdp3-box} $(y,z)/\Theta\left(\log\frac{s}{\lambda}\right)$  
is feasible for SDP \eqref{opt:sdp-packing-box} in each phase.
\item \label{sdp4-box} The sum of the covering objective generated from phase 1 to $r$ is at most $2 \alpha(r)$.
\item \label{sdp5-box} Let $r'$ be the last phase, then the covering objective $c^T x \leq 2 \alpha(r')$.
\end{enumerate}
The proofs of \ref{sdp4-box} and \ref{sdp5-box} are the same as the ones in Theorem~\ref{thm:sdp}, hence omitted. Equipped with these claims and weak duality \eqref{eq:sdp-pd}, we have that
\[
c^T x \le \Theta(1) \alpha(r') \le \Theta(1) D(r') = O\left(\log\frac{s}{\lambda}\right) \opt.
\]

\begin{proof}[Proof of \ref{sdp1-box}.]
We prove that $x^{(r)}$ is feasible in phase $r$ by showing that the growing functions in line \ref{line:covering-box-inc} increments $x^{(r)}$ in a continuous manner. In the beginning of an iteration in round $i$, we have that
\[x^{(r)}_j = B_j\exp\left(\frac{A_j\otimes Y}{c_j}\right)-D_j = \frac{\bar{x}_j + D_j}{\exp\left(\frac{A_j\otimes Y}{c_j}\right)} \exp\left(\frac{A_j\otimes Y}{c_j}\right) - D_j = \bar{x}_j.\]
By following Algorithm \ref{alg:sdp-box}, we have that $x^{(r)}$ is feasible since we terminate the while loop at line \ref{line:sdp-box-while-outer} when there does not exists an SPSD matrix $V$ such that $\sum_{j=1}^n A_j x_j \otimes V < B^{(i)} \otimes V$, indicating that $\sum_{j=1}^n A_jx_j \succeq B^{(i)}$. Whenever we find such a $V$, it induces a linear constraint and we increment $x^{(r)}$ until the cost outside of the tight set exceeds the remaining capacity by a factor of 2 in line \ref{line:sdp-box-while-inner}. As $x$ is the coordinate-wise maximum of $\{x^{(r)}\}$, capped at 1, and the sum of SPSD matrices is still SPSD, $x$ must be feasible for SDP \eqref{opt:sdp-covering-box}.
\end{proof}

\begin{proof}[Proof of \ref{sdp2-box}.]
In the beginning of phase $r$, $x^{(r)}_j = \min\{x'_j,\alpha(r) / (2nc_j)\}$, so $P(r)$ is initially at most $\alpha(r) / 2$. The total increase of $P(r)$ is at least $\alpha(r) / 2$ as $P(r) \geq \alpha(r)$ when phase $r$ ends. Therefore, it suffices to show that
$$\frac{\partial P(r)}{\partial \delta} \leq 3 \frac{\partial D(r)}{\partial \delta}.$$
Recall that $T$ is the set of the tight $x$ indices and the $z_j$ variables with in $j \in T$ are increasing with rate $A_j \otimes V$, we have that
\[\frac{\partial D(r)}{\partial \delta} = B^{(i)} \otimes V - \sum_{j \in T} A_j \otimes V.\]
When $Y$ is increasing, this partial derivative is always non-negative due to the condition in Algorithm \ref{alg:sdp-box} line \ref{line:sdp-box-while-inner}. Namely, at the point when $x^{(r)}_j$ is tight, we add $j$ to $T$, and this may make the partial derivative change from non-negative to negative. In this case, we have that $2(B^{(i)} \otimes V - \sum_{j \in T} A_j \otimes V) < 0$ so the condition automatically fails. From here, when $\sum_{j=1}^n A_jx'_j \succeq B^{(i)}$, we have that
\begin{align*}
\frac{\partial P(r)}{\partial \delta} &= \sum_{j \in [n] \setminus T} c_j \frac{\partial x^{(r)}_j}{\partial \delta} = \sum_{j\in[n] \setminus T} c_j\frac{\partial}{\partial\delta}\left( B_j\,\exp\left(\frac{A_j\otimes (Y+V\delta)- z_j}{c_j}\right)-D_j\right)\\
&= \sum_{j\in[n] \setminus T} c_j\left(\frac{A_j\otimes V}{c_j}\right)B_j\,\exp\left(\frac{A_j\otimes (Y+V\delta) - z_j}{c_j} \right) = \sum_{j\in[n] \setminus T} (A_j\otimes V)(x_j^{(r)}+D_j)\\
&= \sum_{j\in[n] \setminus T} (A_j\otimes V)(x_j^{(r)}+\frac{\lambda (B^{(i)} \otimes V - \sum_{k\in T} A_k\otimes V)}{\sum_{k\in[n] \setminus T} A_k\otimes V}+\frac{(1-\lambda)x'_j\mathbf{1}_{x^{(\ell)}_j < x'_j}(B^{(i)} \otimes V - \sum_{k\in T} A_k\otimes V)}{\sum_{k\in[n] \setminus T} \mathbf{1}_{x^{(\ell)}_k < x'_k}A_kx'_k\otimes V})\\
&\le \left(2 + \lambda + (1-\lambda)\right) (B^{(i)} \otimes V - \sum_{j\in T} A_j\otimes V) = 3 \frac{\partial D(r)}{\partial \delta}
\end{align*}
where the last inequality is due to the fact that $\sum_{j \in [n] \setminus T} A_j x_j \otimes V < 2(B^{(i)} \otimes V - \sum_{j \in T} A_j \otimes V)$ and \eqref{eq:sdp-box-sum-to-1}. The same result can be obtained when $\sum_{j=1}^n A_jx'_j \nsucceq B^{(i)}$ by regarding $\lambda$ as 1.
\end{proof}

\begin{proof}[Proof of \ref{sdp3-box}.]
The constraint $Y\succeq 0$ is satisfied by the same statement in Theorem \ref{thm:sdp}. We show that $y/\Theta(\log\frac{\kappa n}{\lambda})$ is feasible for SDP \eqref{opt:sdp-packing-box} after satisfying the last implicit linear constraint in phase $r$. We recall that within a phase $r$, each iteration ends in one of the following three cases: (1) the implicit linear constraint induces a cost outside of the tight set that exceeds the remaining capacity by a factor of 2, (2) there exists a variable $x^{(r)}_j$ that reaches the advice value $x'_j$, or (3) there exists a variable $x^{(r)}_j$ that reaches 1. Suppose there are $t$ iterations in phase $r$. In iteration $p \le t$, let $x^{(r,p)}_j$, $B^{(p)}_j$, and $D^{(p)}_j$ be the value of $x^{(r)}_j$, $B_j$, and $D_j$ in the end of iteration $p$, respectively. Additionally, let $Y^{(p)}$ and $z^{(p)}_j$ be the value of $Y$ and $z_j$ in the end of iteration $p$, respectively. We have $x_j^{(r,p)}=B^{(p)}_j\,\exp\left(\frac{A_j\otimes Y^{(p)} - z^{(p)}_j}{c_j}\right)-D^{(p)}_j$. 
By using a similar argument in Theorem \ref{thm:sdp} \ref{sdp3}, we have that
\begin{align*}
\frac{A_j\otimes Y^{(t)} - z^{(t)}_j}{c_j}&=\ln\frac{x_j^{(r,t)}+D^{(t)}_j}{B^{(t)}_j} =\ln\left(\frac{x_j^{(r,t)}+D^{(t)}_j}{x_j^{(r,t-1)}+D^{(t)}_j}\exp\left(\frac{A_j\otimes Y^{(t-1)} - z^{(t-1)}_j}{c_j}\right)\right) \\
&= \frac{A_j\otimes Y^{(t-1)} - z^{(t-1)}_j}{c_j} +  \ln\frac{x_j^{(r,t)}+D^{(t)}_j}{x_j^{(r,t-1)}+D^{(t)}_j} = ... \\
&= \ln\prod_{p=1}^p\frac{x_j^{(r,p)}+D^{(p)}_j}{x_j^{(r,p-1)}+D^{(p)}_j} \le \ln\frac{x_j^{(r,p)}+D^{min}_j}{x_j^{(r,0)}+D^{min}_j}
\end{align*}
where $D^{min}_j := \min_{p \in [t]} \{D^{(p)}_j\}$. 

Now we upper-bound $x^{(r,p)}_j$ and lower-bound $D^{min}_j$ to obtain the robustness ratio. Clearly, $x^{(r,p)}_j \le 1$. We also have that
\begin{align*}
D^{min}_j &\ge \frac{\lambda\left(B^{(i)} \otimes V -  \sum_{k\in T} A_k\otimes V \right)}{\sum_{k \in [n] \setminus T} A_k\otimes V}\\
&\ge \min_{
i \in [m], T \subseteq [n]}\max_{
V: B^{(i)} \otimes V - \sum_{k\in T} A_k\otimes V > 0
}\left\{\frac{\lambda \left(B^{(i)} \otimes V - \sum_{k\in T} A_k\otimes V\right)}{\sum_{k\in[n] \setminus T} A_k\otimes V}\right\} \ge \frac{\lambda}{s}
\end{align*}
where the second inequality holds by considering all possible $i$ and $T$ and using our choice of $V$.

With $x^{(r,0)}_j \in [0,1]$, we thus have that
$$\ln \frac{x^{(r,p)}_j + D^{min}_j}{x^{(r,0)}_j + D^{min}_j} \leq \ln \left(1 + \frac{1}{D^{min}_j}\right) = O\left(\log \frac{s}{ \lambda}\right).$$
\end{proof}

\paragraph{Consistency.} We then show that $c^T x \leq O\left(\frac{1}{1-\lambda}\right) c^T x' + O(\log s)\opt$. Suppose the algorithm is in phase $r$ with constraint $i$ arriving. If $\sum_{j=1}^n A_jx'_j \nsucceq B^{(i)}$, then the change of $c^T x$ simply follows \ref{sdp2-box} and \ref{sdp3-box} by regarding $\lambda$ as 1, so $c^T x \le O(\log s)\opt$. For the more interesting case when $\sum_{j=1}^n A_jx'_j \succeq B^{(i)}$, we decompose $\frac{\partial P(r)}{\partial y_i}=\frac{\partial P_c}{\partial y_i}+\frac{\partial P_u}{\partial y_i}$, where $P_c$ is the component of the primal objective due to the advice and $P_u$ is the component of the primal objective due to the online algorithm. 
We have that the rate of change is credited to $\partial P_c$ if $x^{(r)}_j < x'_j$ and the rate of change is credited to $\partial P_u$ otherwise, if $x^{(r)}_j \ge x'_j$. 
In particular,
\begin{align*}
\frac{\partial P_c}{\partial\delta} &= \sum_{j\in[n]: x_j < x'_j} (A_j\otimes V)\left(x_j^{(r)}+\frac{\lambda\left(B^{(i)} \otimes V -  \sum_{k\in T} A_k\otimes V \right)}{\sum_{k \in [n] \setminus T} A_k\otimes V}+\frac{(1-\lambda)x'_j\mathbf{1}_{x^{(r)}_j < x'_j}\left(B^{(i)} \otimes V -  \sum_{k\in T} A_k\otimes V \right)}{\sum_{k \in [n] \setminus T} \mathbf{1}_{x^{(r)}_k < x'_k}A_kx'_k\otimes V}\right) \\
&\geq 0 + \left(\frac{\lambda}{s} + 1 - \lambda\right) \left(B^{(i)} \otimes V -  \sum_{k\in T} A_k\otimes V \right),\\
\frac{\partial P_u}{\partial\delta} &= \sum_{j\in[n]: x_j \ge x'_j} (A_j\otimes V)\left(x_j^{(r)}+\frac{\lambda\left(B^{(i)} \otimes V -  \sum_{k\in T} A_k\otimes V \right)}{\sum_{k \in [n] \setminus T} A_k\otimes V}+\frac{(1-\lambda)x'_j\mathbf{1}_{x^{(r)}_j < x'_j}\left(B^{(i)} \otimes V -  \sum_{k\in T} A_k\otimes V \right)}{\sum_{k \in [n] \setminus T} \mathbf{1}_{x^{(r)}_k < x'_k}A_kx'_k\otimes V}\right)\\
&\leq (2 + \lambda) \left(B^{(i)} \otimes V -  \sum_{k\in T} A_k\otimes V \right) + 0.
\end{align*}
Thus we have $\frac{\partial P_u}{\partial \delta} \leq \frac{(2 + \lambda) \left(B^{(i)} \otimes V -  \sum_{k\in T} A_k\otimes V \right)}{(\lambda/s + 1 - \lambda) \left(B^{(i)} \otimes V -  \sum_{k\in T} A_k\otimes V \right)} \cdot \frac{\partial P_c}{\partial \delta}$, so that
\[\frac{\partial P(r)}{\partial \delta} \leq \left(1+\frac{2 + \lambda}{1-\lambda}\right) \frac{\partial P_c}{\partial \delta} = O\left(\frac{1}{1-\lambda}\right) \frac{\partial P_c}{\partial \delta}.\]
We note that if $x'$ is feasible, then $\sum_{j=1}^n A_jx'_j \succeq B^{(i)}$, so $c^T x \leq O\left(\frac{1}{1-\lambda}\right) c^T x'$.

\paragraph{Bounding the number of iterations.} The analysis follows similarly to the one for Theorem \ref{thm:covering-box}. We first have that there are $O(\log (c^T x/\alpha(1)))$ phases.

We recall that within phase $r$, each iteration ends in one of the following three cases: (1) the implicit linear constraint induces a cost outside of the tight set that exceeds the remaining capacity by a factor of 2, (2) there exists a variable $x^{(r)}_j$ that reaches the advice value $x'_j$, (3) there exists a variable $x^{(r)}_j$ that reaches 1. If an iteration ends in the first case, before $x^{(r)}$ is updated, we have an implicit violating linear constraint induced by $V$ such that $\sum_{j \in [n] \setminus T} A_j x^{(r)}_j \otimes V < B^{(i)} \otimes V - \sum_{j \in T} A_j \otimes V$. We increment $x^{(r)}_j$'s where $j \in [n] \setminus T$ until $\sum_{j \in [n] \setminus T} A_j x^{(r)}_j \otimes V \ge 2\left(B^{(i)} \otimes V - \sum_{j \in T} A_j \otimes V\right)$. One of the following two cases must hold after updating $x^{(r)}$: (1) there exist a \emph{large} variable $x^{(r)}_j \geq 1/(2s)$ that is updated to at least $3x^{(r)}_j/2$, or (2) there exists a \emph{small} variable $x^{(r)}_j < 1/(2s)$ that becomes large, i.e., $x^{(r)}_j$ is updated to at least $1/(2s)$. Let $L$ and $S$ be the set of large and small variable subscript labels in $[n] \setminus T$ before the violating constraint $i$ arrives, respectively, and $\hat{x}^{(r)}_j$ be the value of $x^{(r)}_j$ after the update. If none of these two cases holds, then
\begin{align*}
\sum_{j \in [n] \setminus T} A_j \hat{x}^{(r)}_j \otimes V &< \frac{3}{2} \sum_{j \in L} A_j x^{(r)}_j \otimes V + \frac{1}{2s} \sum_{j \in S} A_j \otimes V \\
&< \frac{3}{2}\left(B^{(i)} \otimes V - \sum_{j \in T} A_j \otimes V\right) + \frac{\left(B^{(i)} \otimes V - \sum_{j \in T} A_j \otimes V\right)}{2\sum_{j \in [n] \setminus T} A_j \otimes V} \sum_{j \in S} A_j \otimes V \\
&\le 2 \left(B^{(i)} \otimes V - \sum_{j \in T} A_j \otimes V\right)
\end{align*}
where the second inequality is by the fact that implicit linear constraint induced by $V$ is violated before the update and the definition of the sparsity $s$, while the last inequality is by $S \subseteq [n] \setminus T$ and $A_j \otimes V \ge 0$ since $A_j$ is SPSD.
This implies that the cost outside of the tight set does not exceed the remaining capacity by a factor of 2 after the update, a contradiction.

Suppose $x^{(r)}_j$ has been updated $t$ times by a multiplicative factor of 3/2 since it was large in phase $r$, then $t=O(\log s)$ since $\frac{1}{2s}(\frac{3}{2})^{t} \le 1$.

There are $n$ variables. In each phase, each variable can be updated from small to large once, updated at most $t$ times by a factor of $3/2$ since it was large, capped at 1 once, and reach the advice value once. Hence, Algorithm~\ref{alg:sdp-box:phase} has
$O(n \log s \log \frac{c^T x}{\alpha(1)})$
iterations.
\end{proof}

\section{Empirical Evaluations} \label{sec:ex}

We demonstrate the applicability of our algorithmic framework on a synthetic dataset as well as a real world case study on internet router graphs. Our focus will be on evaluating our online covering algorithms, Algorithms \ref{alg:covering:phase} and \ref{alg:covering}, with possibly fractional hints and fractional entries. We focus on this setting since it’s the simplest of our algorithms and already captures our overall algorithm framework. Note that prior work in \cite{BamasMS20} has already demonstrated the empirical advantage of learning-based methods for online covering with integral hints and constraints, albeit on synthetic datasets. 

\paragraph{Datasets.}
Our synthetic dataset is constructed as follows. The constraint matrix $A$ represents a $n \times n$ matrix where each entry is uniformly in $\{0, 1\}$. We set $n = 500$. The objective function $c$ is a scaled vector with entries uniform in $[0,1]$. The vector we are covering is the all ones vector, i.e., we are solving $Ax \ge \mathbf{1}$, and the rows of $A$ arrive online. Our graph dataset is constructed as follows. We have a sequence of nine (unweighted) graphs which represents an internet router network sampled across time~\cite{snapnets, oregon1}\footnote{Graphs can be accessed in \url{https://snap.stanford.edu/data/Oregon-1.html}}. The graphs have approximately $n \sim 10^4$ nodes and $m \sim 2.2 \cdot 10^4$ edges. We note that the nodes of the graphs are labeled and the labeling is consistent throughout the different time stamps. Each graph defines an instance of the set cover problem derived from vertex cover as follows. The edges of the graphs are labeled and the $n$ vertex neighborhoods define $n$ sets. The edges of the graph (the universe elements) appear online and we must cover these edges by choosing vertices (the sets) which are incident on them. The objective function $c$ will be the same as the synthetic case so our problem represents an instance of weighted vertex cover.

\paragraph{Predictions.}  We instantiate predictions in a variety of ways, drawing inspiration from many prior works\cite{BamasMS20, DinitzILMV21, chen2022}. First we describe our predictions for the synthetic datasets. We consider two types of predictions. In one case, we first find the optimal offline solution $x$ by solving the full linear program. We then noisy corrupt the entries of $x$ by setting the entries to be $0$ independently with some probability $p$. This is the same type of predictions in \cite{BamasMS20} which they refer to as the ‘replacement rate’ strategy.

The second type of predictions are informed by the following motivating setting of learning-augmented algorithms: we are solving many different, but related, problem instances. To mimic this situation, we have matrices $A_0, A_1, \ldots$ where each index represents a new problem instance. $A_0$ is our synthetic matrix and $A_{i+1}$ updates $A_i$ by flipping $n$ entries at random. We fix $c$ to be the same throughout. The predictions for all instances $i \ge 1$ are given by the optimal offline solution generated from the first instance $A_0$. This ``batch" experimental design naturally models the scenario where the current problem instance is similar to past instances and so one can hope to utilize past learned information to improve current algorithmic performance. A similar style of predictions, although not in an online context, has been employed in \cite{ChenEILNRSWWZ22, ErgunFSWZ22,  DinitzILMV21, chen2022}. Furthermore, this mimics the setting for proving PAC learning bounds.

For our graph dataset, we first solve the set cover instance on the first graph in the family using an offline linear program solver. We then use the solution from the first graph as the hint for all subsequent graphs. We also noisy alter this hint for one of our experiments using the replacement rate strategy. Note that the set of vertices might vary across graphs since new vertices can appear in the network while older vertices may be removed. In this case, we set the corresponding entry in the hint vector to be $0$. 

The type of dataset and predictions used will always be stated explicitly in our experimental results below.

\paragraph{Results.} 
Our results are shown in Figures \ref{fig:synthetic} and \ref{fig:graphs}. Figure \ref{fig:synthetic} shows the results on the synthetic dataset while Figure \ref{fig:graphs} refers to our graph dataset. Description of each figure follows.

In Figure \ref{fig:synthetic_a}, we consider a single online instance of the synthetic dataset. Our prediction is the offline optimal solution. The figure shows a smooth trade-off in the competitive ratio as the parameter $\lambda$ ranges from $0$ (full trust in the predictions) to $1$ (standard online algorithm with no hints), as predicted by our theoretical bounds. Since the instance is random, we plot the average of $20$ trials for each setting of $\lambda$ and also shade in one standard deviation. The plot validates the consistency of our algorithms as the competitive ratio is a factor of two lower with accurate predictions.

In Figure \ref{fig:synthetic_b}, we again consider a single synthetic instance. This time, we consider the ``replacement rate" strategy and randomly zero out the entries of the prediction (which is again the offline optimum) independently. The expected fraction of entries in the hint vector being set to zero is denoted as the corruption rate and is shown in the $x$-axis. We see that for a fixed setting of $\lambda$, such as $\lambda = 0.1$, our algorithm performs much better than not utilizing any hints when the corruption factor is low. This intuitively makes sense and is exactly what Figure \ref{fig:synthetic_a} shows. However, as we increase the corruption factor, the performance of the algorithm degrades. Crucially, the performance \emph{does not} degrade arbitrarily worse compared to $\lambda = 1$ (no hints) performance, which empirically validates the robustness of our algorithm. Indeed, our algorithm with hints is able to outperform the baseline even up to a high corruption factor. Lastly, if we utilize hints in a naive way where we only scale the hint variables to satisfy the constraints, then the competitive ratio is at least four orders of magnitude larger than the values in Figure \ref{fig:synthetic_b}. Thus, post-processing the hints, such as what our algorithm does, is crucial.

In Figure \ref{fig:synthetic_c}, we consider a ``batch" experimental design for our synthetic dataset with $20$ time steps. The green curve shows the competitive ratio without using any hints, i.e., $\lambda$ is set to $1$ (the baseline). The orange curve shows the competitive ratio across the varying instances when we use a batch prediction. Precisely, this means we use the optimal offline solution of the \emph{first} instance and this hint is fixed for all future instances which noisily drift away from the first instance. The blue curve showcases more powerful predictions where the offline optimal of time step $t-1$ is used as the hint for time step $t$. We display the average values across $20$ instances. We see that as the time step increases, the orange curve drifts upwards, which is intuitive as the problem instances are increasingly different. Nevertheless, the hint stays valid for many time steps as the orange curve is still below the green baseline even after many time steps. As expected, the blue curve consistently has the lowest competitive ratio as the hints are also updated. We do not shade in the standard deviation to increase the clarity of the figure but the variance of the curves similar to Figure \ref{fig:synthetic_a}.

\begin{figure}
     \centering
     \begin{subfigure}[b]{0.3\textwidth}
         \centering
         \includegraphics[width=\textwidth]{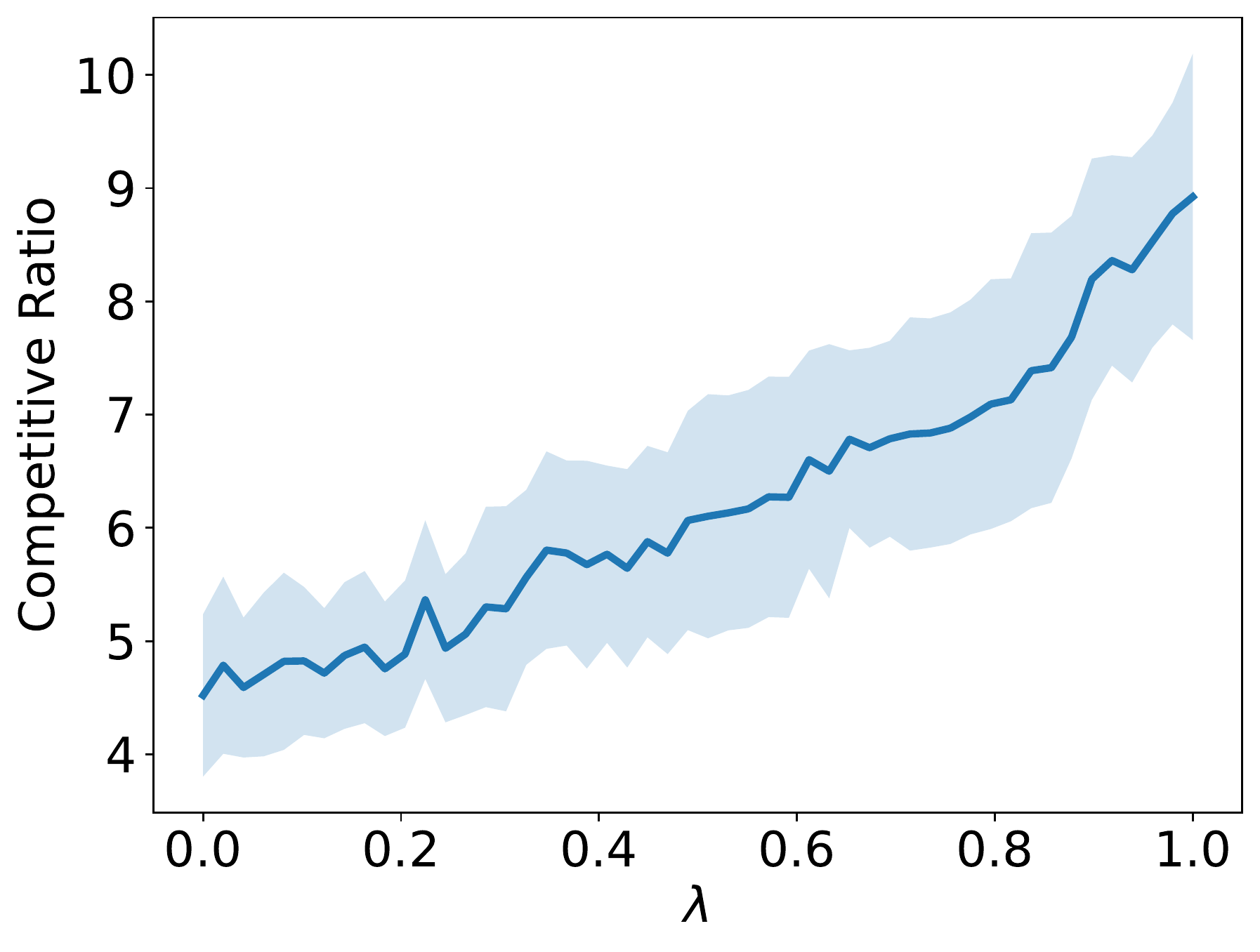}
         \caption{}
         \label{fig:synthetic_a}
     \end{subfigure}
     \hfill
     \begin{subfigure}[b]{0.3\textwidth}
         \centering
         \includegraphics[width=\textwidth]{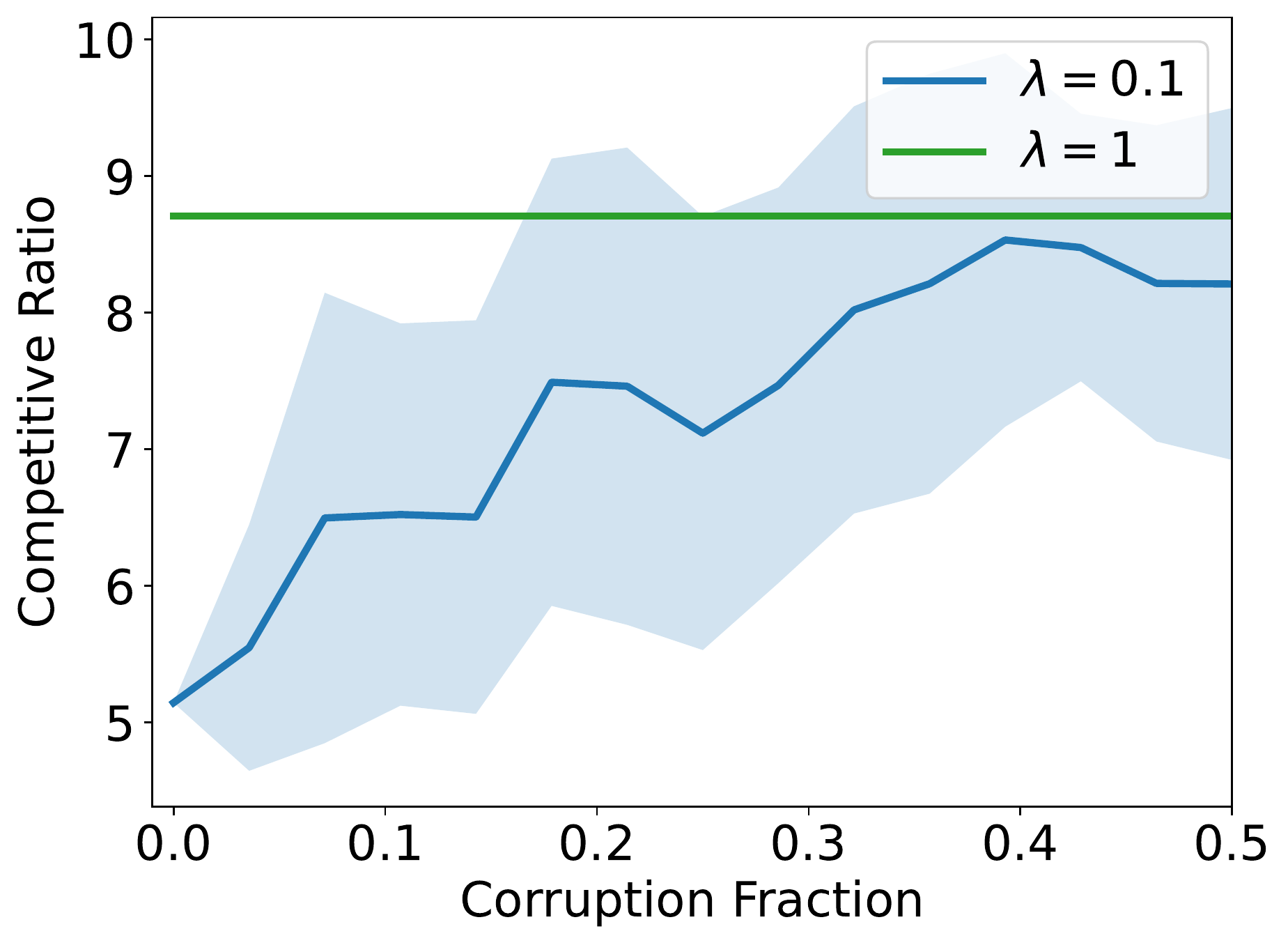}
         \caption{}
         \label{fig:synthetic_b}
     \end{subfigure}
     \hfill
     \begin{subfigure}[b]{0.3\textwidth}
         \centering
         \includegraphics[width=\textwidth]{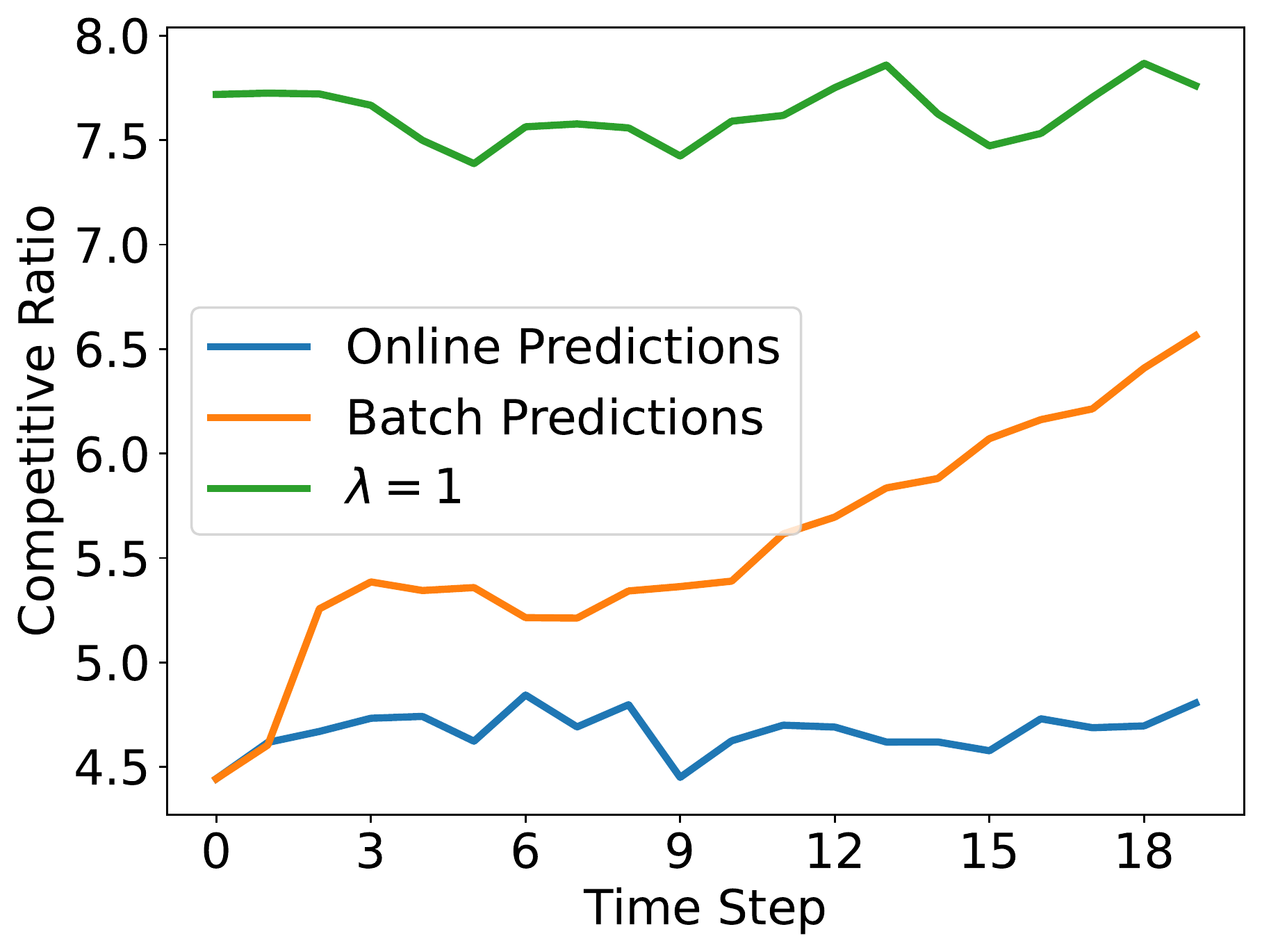}
         \caption{}
         \label{fig:synthetic_c}
     \end{subfigure}
        \caption{Figures for our synthetic dataset.}
        \label{fig:synthetic}
\end{figure}

\begin{figure}
     \centering
     \begin{subfigure}[b]{.45\textwidth}
         \centering
         \includegraphics[width=.7\textwidth]{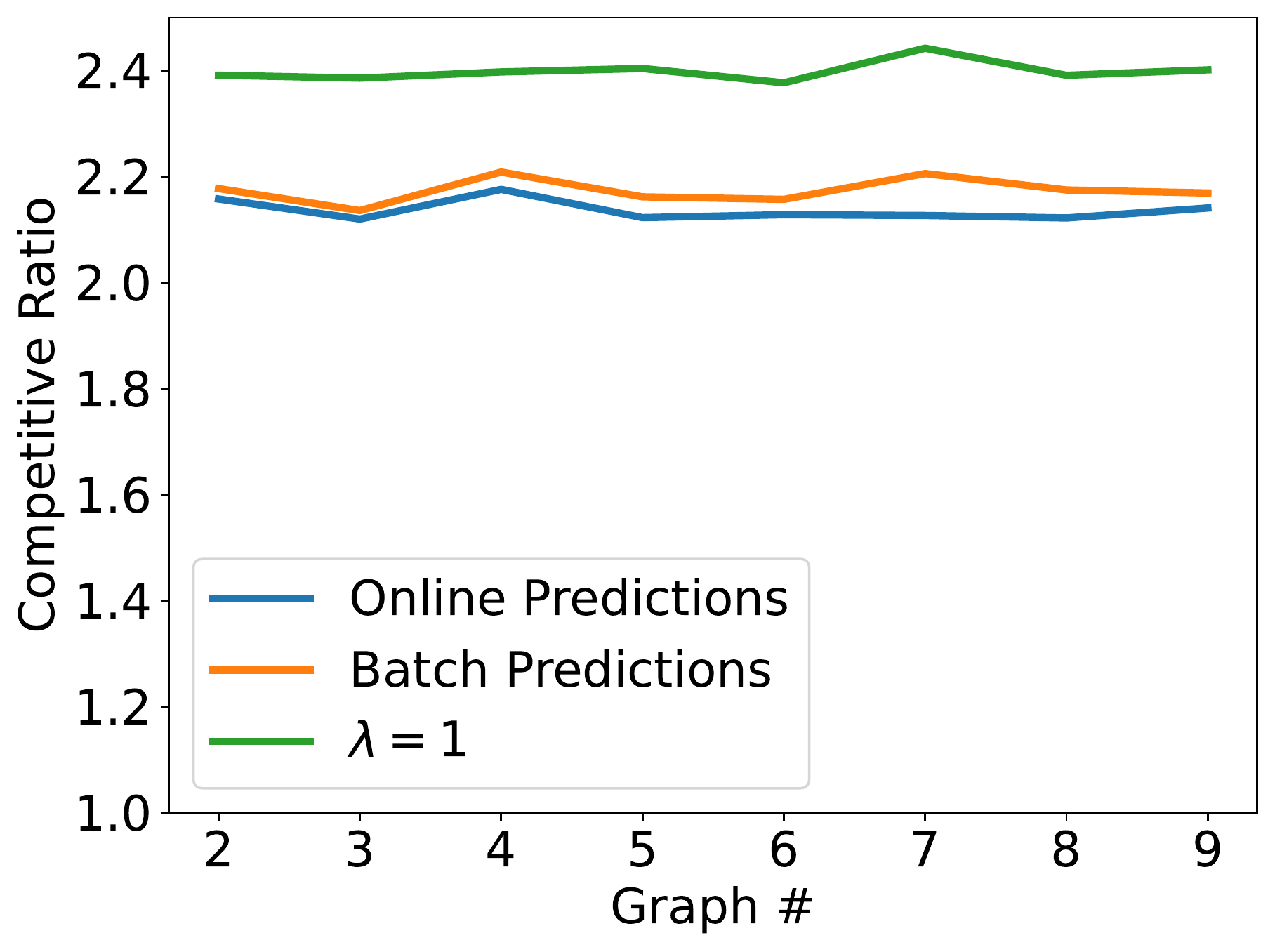}
         \caption{}
         \label{fig:graph_a}
     \end{subfigure}
     \hfill
     \begin{subfigure}[b]{0.45\textwidth}
         \centering
         \includegraphics[width=.7\textwidth]{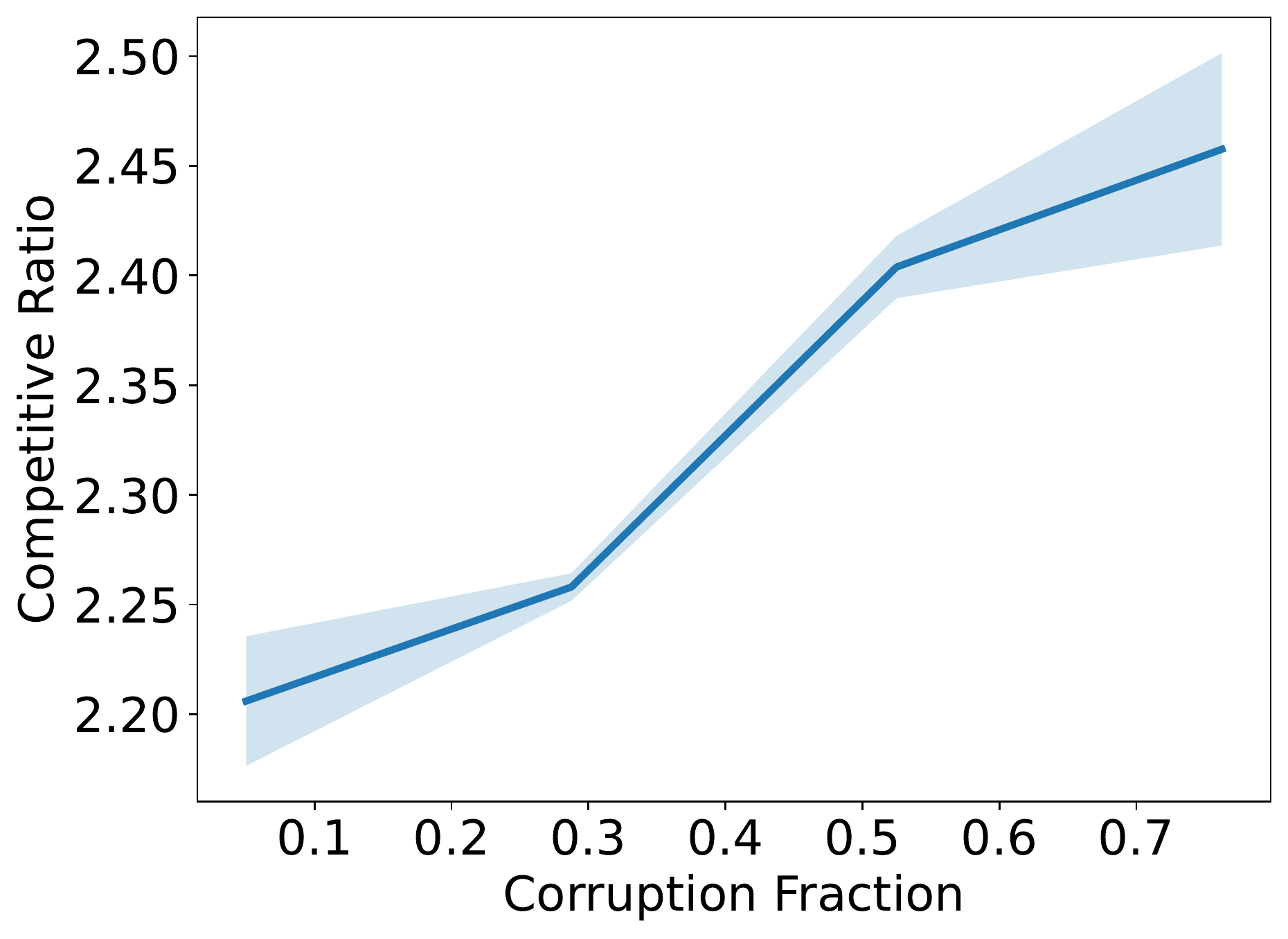}
         \caption{}
         \label{fig:graph_b}
     \end{subfigure}
        \caption{Figures for our time-varying graph dataset.}
        \label{fig:graphs}
\end{figure}

We now describe the experimental results on our graph dataset. In Figure \ref{fig:graph_a}, the green curve represents not using any hints (baseline) while the orange curve shows the competitive ratio as we vary the graph instance while using the hint derived from graph $\#1$ ($\lambda$ is set to $0.1$). It is shown that the hints help outperform the baseline and in addition, the hints stay accurate even if the structure of graph $\#9$ has drifted away from that of graph $\#1$. The online predictions, shown in blue, does marginally better than the batch version. 

Figure \ref{fig:graph_b} shows a similar plot as Figure \ref{fig:synthetic_b}. We consider a replacement rate strategy where we zero out coordinates of the hint vector independently with varying probabilities. The curve shows the average of $5$ trials and $\lambda = 0.1$ again. The same qualitative message as Figure \ref{fig:synthetic_b} holds: while the corruption rate is small, we achieve a similar competitive ratio as in Figure \ref{fig:graph_a} and as the corruption rate increases, there is a smooth increase in the competitive ratio.

In addition to extending and complementing the experimental results of \cite{BamasMS20}, we summarize our experimental results in the following points: 
(a) Our theory is predictive of experimental performance and qualitatively validates our robustness and consistency trade offs.
(b) Our algorithm framework which underlies all of our algorithm contributions is efficient to carry out and execute in practice. 
(c) Learning-augmented online algorithms can be applied to real world datasets varying over time such as in the analysis of graphs derived from a dynamic network.

\section{Conclusion}

We present the first framework for the learning-augmented online covering LP and SDP problem. As shown in table \ref{table:pdla}, for the problems without box constraints, our algorithms are $O(1/(1-\lambda)))$-consistent and $O(\log (\kappa n/\lambda))$-robust; for the problems with box constraints, our algorithms are $O(1/(1-\lambda)))$-consistent and $O(\log (s/\lambda))$-robust. Our algorithms not only support fractional advice but also update the variables efficiently. Our work raises several open questions.

First, for the problem without box constraints, is it possible to remove the dependency on the condition number $\kappa$ in the robustness ratio? The competitive ratio of the optimal online algorithms for covering LPs \cite{BuchbinderN09} and SDPs \cite{EladKN16} is $O(\log n)$. For the learning-augmented set cover problem where the entries in the constraint matrix are zeros and ones, the robustness ratio is $O(\log(d/\lambda))$, where the worst case for $d$ is $d=n$. It seems that the condition number arises due to additive term $D_j$ used for the growth rate. Another direction is to discover the trader-off lower bounds between the robustness and consistency ratio for the online covering and SDP covering problem as in \cite{WeiZ20} for the ski rental and the  non-clairvoyant scheduling problem.

Next, our technique can potentially be used on other online covering problems \cite{azar2016online,shen2020online,gupta2014approximating} for the learning-augmented extension. It would be exciting to see if the primal-dual continuous augmentation approach with well-selected growth rates can be employed for these problems.

As stated in \cite{BamasMS20}, another natural future direction is to design PDLA algorithms for packing LPs. For general online packing LPs, an $O(1/ \log \kappa)$-competitive online solution can be obtained only if a condition number $\kappa$ is known offline \cite{buchbinder2009online}, implying impossibility results without assumptions. The hope here is to study structured packing problems, e.g., load balancing \cite{BuchbinderN09} and ad-auction revenue maximization \cite{buchbinder2007online}.

\section*{Acknowledgment}
We thank Roie Levin for bringing to our attention the significantly simpler algorithm described in Appendix~\ref{app:simple}.
\bibliographystyle{alpha}
\bibliography{reference}

\newcommand{\etalchar}[1]{$^{#1}$}
\begin{thebibliography}{FMMM09}

\bibitem[AAA{\etalchar{+}}06]{alon2006general}
Noga Alon, Baruch Awerbuch, Yossi Azar, Niv Buchbinder, and Joseph Naor.
\newblock A general approach to online network optimization problems.
\newblock {\em ACM Transactions on Algorithms (TALG)}, 2(4):640--660, 2006.

\bibitem[AAA{\etalchar{+}}09]{alon2009online}
Noga Alon, Baruch Awerbuch, Yossi Azar, Niv Buchbinder, and Joseph Naor.
\newblock The online set cover problem.
\newblock {\em SIAM Journal on Computing}, 39(2):361--370, 2009.

\bibitem[AB99]{AnthonyB99}
Martin Anthony and Peter~L. Bartlett.
\newblock {\em Neural Network Learning: Theoretical Foundations}.
\newblock Cambridge University Press, 1999.

\bibitem[ABC{\etalchar{+}}16]{azar2016online}
Yossi Azar, Niv Buchbinder, TH~Hubert Chan, Shahar Chen, Ilan~Reuven Cohen,
  Anupam Gupta, Zhiyi Huang, Ning Kang, Viswanath Nagarajan, Joseph Naor,
  et~al.
\newblock Online algorithms for covering and packing problems with convex
  objectives.
\newblock In {\em 57th Annual Symposium on Foundations of Computer Science
  (FOCS)}, pages 148--157, 2016.

\bibitem[ABFP13]{azar2013online}
Yossi Azar, Umang Bhaskar, Lisa Fleischer, and Debmalya Panigrahi.
\newblock Online mixed packing and covering.
\newblock In {\em Proceedings of the twenty-fourth annual ACM-SIAM symposium on
  Discrete algorithms}, pages 85--100, 2013.

\bibitem[ACE{\etalchar{+}}20]{antoniadis2020online}
Antonios Antoniadis, Christian Coester, Marek Elias, Adam Polak, and Bertrand
  Simon.
\newblock Online metric algorithms with untrusted predictions.
\newblock In {\em International Conference on Machine Learning}, pages
  345--355. PMLR, 2020.

\bibitem[AGKK20]{antoniadis2020secretary}
Antonios Antoniadis, Themis Gouleakis, Pieter Kleer, and Pavel Kolev.
\newblock Secretary and online matching problems with machine learned advice.
\newblock {\em Advances in Neural Information Processing Systems},
  33:7933--7944, 2020.

\bibitem[AGKP22]{AGKP2022multiple}
Keerti Anand, Rong Ge, Amit Kumar, and Debmalya Panigrahi.
\newblock Online algorithms with multiple predictions, 2022.

\bibitem[APT22]{azar2022online}
Yossi Azar, Debmalya Panigrahi, and Noam Touitou.
\newblock Online graph algorithms with predictions.
\newblock In {\em Proceedings of the 2022 Annual ACM-SIAM Symposium on Discrete
  Algorithms (SODA)}, pages 35--66. SIAM, 2022.

\bibitem[AW02]{ahlswede2002strong}
Rudolf Ahlswede and Andreas Winter.
\newblock Strong converse for identification via quantum channels.
\newblock {\em IEEE Transactions on Information Theory}, 48(3):569--579, 2002.

\bibitem[BBN12]{bansal2012primal}
Nikhil Bansal, Niv Buchbinder, and Joseph Naor.
\newblock A primal-dual randomized algorithm for weighted paging.
\newblock {\em Journal of the ACM (JACM)}, 59(4):1--24, 2012.

\bibitem[BJN07]{buchbinder2007online}
Niv Buchbinder, Kamal Jain, and Joseph~Seffi Naor.
\newblock Online primal-dual algorithms for maximizing ad-auctions revenue.
\newblock In {\em European Symposium on Algorithms}, pages 253--264. Springer,
  2007.

\bibitem[BMRS20]{BamasMRS20}
{\'{E}}tienne Bamas, Andreas Maggiori, Lars Rohwedder, and Ola Svensson.
\newblock Learning augmented energy minimization via speed scaling.
\newblock In {\em Advances in Neural Information Processing Systems 33: Annual
  Conference on Neural Information Processing Systems, NeurIPS}, 2020.

\bibitem[BMS20]{BamasMS20}
{\'{E}}tienne Bamas, Andreas Maggiori, and Ola Svensson.
\newblock The primal-dual method for learning augmented algorithms.
\newblock In {\em Advances in Neural Information Processing Systems 33: Annual
  Conference on Neural Information Processing Systems, NeurIPS}, 2020.

\bibitem[BN09a]{BuchbinderN09}
Niv Buchbinder and Joseph Naor.
\newblock The design of competitive online algorithms via a primal-dual
  approach.
\newblock {\em Found. Trends Theor. Comput. Sci.}, 3(2-3):93--263, 2009.

\bibitem[BN09b]{buchbinder2009online}
Niv Buchbinder and Joseph Naor.
\newblock Online primal-dual algorithms for covering and packing.
\newblock {\em Mathematics of Operations Research}, 34(2):270--286, 2009.

\bibitem[CEI{\etalchar{+}}22]{ChenEILNRSWWZ22}
Justin~Y. Chen, Talya Eden, Piotr Indyk, Honghao Lin, Shyam Narayanan, Ronitt
  Rubinfeld, Sandeep Silwal, Tal Wagner, David~P. Woodruff, and Michael Zhang.
\newblock Triangle and four cycle counting with predictions in graph streams.
\newblock In {\em The Tenth International Conference on Learning
  Representations, {ICLR}}, 2022.

\bibitem[CSVZ22]{chen2022}
Justin~Y. Chen, Sandeep Silwal, Ali Vakilian, and Fred Zhang.
\newblock Faster fundamental graph algorithms via learned predictions.
\newblock In {\em International Conference on Machine Learning, {ICML}}, pages
  3583--3602, 2022.

\bibitem[DIL{\etalchar{+}}21]{DinitzILMV21}
Michael Dinitz, Sungjin Im, Thomas Lavastida, Benjamin Moseley, and Sergei
  Vassilvitskii.
\newblock Faster matchings via learned duals.
\newblock In {\em Advances in Neural Information Processing Systems 34: Annual
  Conference on Neural Information Processing Systems, NeurIPS}, pages
  10393--10406, 2021.

\bibitem[DIRW20]{DongIRW20}
Yihe Dong, Piotr Indyk, Ilya~P. Razenshteyn, and Tal Wagner.
\newblock Learning space partitions for nearest neighbor search.
\newblock In {\em 8th International Conference on Learning Representations,
  {ICLR}, 2020}, 2020.

\bibitem[DKT{\etalchar{+}}21]{DiakonikolasKTV21}
Ilias Diakonikolas, Vasilis Kontonis, Christos Tzamos, Ali Vakilian, and Nikos
  Zarifis.
\newblock Learning online algorithms with distributional advice.
\newblock In {\em Proceedings of the 38th International Conference on Machine
  Learning, {ICML}}, pages 2687--2696, 2021.

\bibitem[EFS{\etalchar{+}}22]{ErgunFSWZ22}
Jon~C. Ergun, Zhili Feng, Sandeep Silwal, David~P. Woodruff, and Samson Zhou.
\newblock Learning-augmented $k$-means clustering.
\newblock In {\em The Tenth International Conference on Learning
  Representations, {ICLR}}, 2022.

\bibitem[EIN{\etalchar{+}}21]{EdenINRSW21}
Talya Eden, Piotr Indyk, Shyam Narayanan, Ronitt Rubinfeld, Sandeep Silwal, and
  Tal Wagner.
\newblock Learning-based support estimation in sublinear time.
\newblock In {\em 9th International Conference on Learning Representations,
  {ICLR}}, 2021.

\bibitem[EKN16]{EladKN16}
Noa Elad, Satyen Kale, and Joseph~(Seffi) Naor.
\newblock Online semidefinite programming.
\newblock In {\em 43rd International Colloquium on Automata, Languages, and
  Programming, {ICALP}}, pages 40:1--40:13, 2016.

\bibitem[FMMM09]{feldman2009online}
Jon Feldman, Aranyak Mehta, Vahab Mirrokni, and Shan Muthukrishnan.
\newblock Online stochastic matching: Beating 1-1/e.
\newblock In {\em 2009 50th Annual IEEE Symposium on Foundations of Computer
  Science}, pages 117--126. IEEE, 2009.

\bibitem[GKR00]{garg2000polylogarithmic}
Naveen Garg, Goran Konjevod, and Ramamoorthi Ravi.
\newblock A polylogarithmic approximation algorithm for the group steiner tree
  problem.
\newblock {\em Journal of Algorithms}, 37(1):66--84, 2000.

\bibitem[GL20]{gupta2020online}
Anupam Gupta and Roie Levin.
\newblock The online submodular cover problem.
\newblock In {\em Proceedings of the Fourteenth Annual ACM-SIAM Symposium on
  Discrete Algorithms}, pages 1525--1537. SIAM, 2020.

\bibitem[GLQ21]{grigorescu2021online}
Elena Grigorescu, Young-San Lin, and Kent Quanrud.
\newblock Online directed spanners and steiner forests.
\newblock In {\em Approximation, Randomization, and Combinatorial Optimization.
  Algorithms and Techniques (APPROX/RANDOM 2021)}, 2021.

\bibitem[GN14]{gupta2014approximating}
Anupam Gupta and Viswanath Nagarajan.
\newblock Approximating sparse covering integer programs online.
\newblock {\em Mathematics of Operations Research}, 39(4):998--1011, 2014.

\bibitem[GP19]{GollapudiP19}
Sreenivas Gollapudi and Debmalya Panigrahi.
\newblock Online algorithms for rent-or-buy with expert advice.
\newblock In {\em Proceedings of the 36th International Conference on Machine
  Learning, {ICML}}, pages 2319--2327, 2019.

\bibitem[GW95]{GoemansW95}
Michel~X. Goemans and David~P. Williamson.
\newblock Improved approximation algorithms for maximum cut and satisfiability
  problems using semidefinite programming.
\newblock {\em J. {ACM}}, 42(6):1115--1145, 1995.

\bibitem[HIKV19]{HsuIKV19}
Chen{-}Yu Hsu, Piotr Indyk, Dina Katabi, and Ali Vakilian.
\newblock Learning-based frequency estimation algorithms.
\newblock In {\em 7th International Conference on Learning Representations,
  {ICLR}}, 2019.

\bibitem[ISZ21]{IzzoSZ21}
Zachary Izzo, Sandeep Silwal, and Samson Zhou.
\newblock Dimensionality reduction for wasserstein barycenter.
\newblock In {\em Advances in Neural Information Processing Systems 34: Annual
  Conference on Neural Information Processing Systems 2021, NeurIPS}, pages
  15582--15594, 2021.

\bibitem[JLL{\etalchar{+}}20]{JiangLLRW20}
Tanqiu Jiang, Yi~Li, Honghao Lin, Yisong Ruan, and David~P. Woodruff.
\newblock Learning-augmented data stream algorithms.
\newblock In {\em 8th International Conference on Learning Representations,
  {ICLR}}, 2020.

\bibitem[KDZ{\etalchar{+}}17]{khalil2017learning}
Elias Khalil, Hanjun Dai, Yuyu Zhang, Bistra Dilkina, and Le~Song.
\newblock Learning combinatorial optimization algorithms over graphs.
\newblock {\em Advances in neural information processing systems}, 30, 2017.

\bibitem[KMMO94]{kmmo}
Anna~R. Karlin, Mark~S. Manasse, Lyle~A. McGeoch, and Susan~S. Owicki.
\newblock Competitive randomized algorithms for non-uniform problems.
\newblock {\em Algorithmica}, 11(6):542--571, 1994.

\bibitem[LFKF18]{Feldman18}
Mario Lucic, Matthew Faulkner, Andreas Krause, and Dan Feldman.
\newblock Training gaussian mixture models at scale via coresets.
\newblock {\em Journal of Machine Learning Research}, 18(160):1--25, 2018.

\bibitem[LK14]{snapnets}
Jure Leskovec and Andrej Krevl.
\newblock {SNAP Datasets}: {Stanford} large network dataset collection.
\newblock \url{http://snap.stanford.edu/data}, June 2014.

\bibitem[LKF05]{oregon1}
Jure Leskovec, Jon~M. Kleinberg, and Christos Faloutsos.
\newblock Graphs over time: densification laws, shrinking diameters and
  possible explanations.
\newblock In {\em Proceedings of the Eleventh {ACM} {SIGKDD} International
  Conference on Knowledge Discovery and Data Mining}, pages 177--187, 2005.

\bibitem[LLMV20]{LattanziLMV20}
Silvio Lattanzi, Thomas Lavastida, Benjamin Moseley, and Sergei Vassilvitskii.
\newblock Online scheduling via learned weights.
\newblock In {\em Proceedings of the 2020 {ACM-SIAM} Symposium on Discrete
  Algorithms, {SODA}}, pages 1859--1877, 2020.

\bibitem[LV18]{LykourisV18}
Thodoris Lykouris and Sergei Vassilvitskii.
\newblock Competitive caching with machine learned advice.
\newblock In {\em Proceedings of the 35th International Conference on Machine
  Learning, {ICML}}, pages 3302--3311, 2018.

\bibitem[Mit18]{Mitzenmacher18}
Michael Mitzenmacher.
\newblock A model for learned bloom filters and optimizing by sandwiching.
\newblock In {\em Advances in Neural Information Processing Systems 31: Annual
  Conference on Neural Information Processing Systems, NeurIPS}, pages
  462--471, 2018.

\bibitem[MNS12]{mahdian2012online}
Mohammad Mahdian, Hamid Nazerzadeh, and Amin Saberi.
\newblock Online optimization with uncertain information.
\newblock {\em ACM Transactions on Algorithms (TALG)}, 8(1):1--29, 2012.

\bibitem[ND21]{thang2021packing}
Kim~Thang Nguyen and Christoph D{\"{u}}rr.
\newblock Online primal-dual algorithms with predictions for packing problems.
\newblock {\em CoRR}, abs/2110.00391, 2021.

\bibitem[PSK18]{PurohitSK18}
Manish Purohit, Zoya Svitkina, and Ravi Kumar.
\newblock Improving online algorithms via {ML} predictions.
\newblock In {\em Advances in Neural Information Processing Systems 31: Annual
  Conference on Neural Information Processing Systems, NeurIPS}, pages
  9684--9693, 2018.

\bibitem[Roh20]{Rohatgi20}
Dhruv Rohatgi.
\newblock Near-optimal bounds for online caching with machine learned advice.
\newblock In {\em Proceedings of the 2020 {ACM-SIAM} Symposium on Discrete
  Algorithms, {SODA}}, pages 1834--1845, 2020.

\bibitem[SN20]{shen2020online}
Xiangkun Shen and Viswanath Nagarajan.
\newblock Online covering with $l_q$-norm objectives and applications to
  network design.
\newblock {\em Mathematical Programming}, 184, 2020.

\bibitem[SZS{\etalchar{+}}14]{SzegedyZSBEGF13}
Christian Szegedy, Wojciech Zaremba, Ilya Sutskever, Joan Bruna, Dumitru Erhan,
  Ian~J. Goodfellow, and Rob Fergus.
\newblock Intriguing properties of neural networks.
\newblock In {\em 2nd International Conference on Learning Representations,
  {ICLR}, Conference Track Proceedings}, 2014.

\bibitem[VB96]{VandenbergheB96}
Lieven Vandenberghe and Stephen~P. Boyd.
\newblock Semidefinite programming.
\newblock {\em {SIAM} Rev.}, 38(1):49--95, 1996.

\bibitem[WX06]{wigderson2006derandomizing}
Avi Wigderson and David Xiao.
\newblock Derandomizing the aw matrix-valued chernoff bound using pessimistic
  estimators and applications.
\newblock {\em ECCC TR06-105}, 2006.

\bibitem[WZ20]{WeiZ20}
Alexander Wei and Fred Zhang.
\newblock Optimal robustness-consistency trade-offs for learning-augmented
  online algorithms.
\newblock In {\em Advances in Neural Information Processing Systems 33: Annual
  Conference on Neural Information Processing Systems 2020, NeurIPS}, 2020.

\end{thebibliography}

\appendix
\section{Learnability of Predictor} \label{sec:appendix}
In this section, we present formal learning bounds for learning a good predictor for the covering formulation of online semidefinite programming, which includes learning a good predictor for online fractional linear programming. 
Specifically, we prove that a good predictor can be efficiently learned if the problem instances are drawn from a particular distribution, in the style of the probably approximately correct (PAC) learning framework for data-driven algorithm design. 
Similar ideas were used by \cite{ErgunFSWZ22,IzzoSZ21,ChenEILNRSWWZ22, chen2022} to show efficient learners for predictors for $k$-means clustering and triangle counting. Thus while the results of this section use ideas of prior papers, they are still meaningful as it provides an ``end-to-end" recipe for designing learning-augmented online algorithms. The prior part of this paper designs online algorithms which leverage predictions and thus, the natural follow up questions are how efficiently one can learn these hint from data. Therefore this learnability section precisely addresses these questions and makes our results complete.

Suppose there exists an underlying distribution $\mathcal{D}$ that generates independent covering semidefinite programming or covering linear program instances, capturing the case where similar problem instances are being solved. 
Note that this setting also mirrors some of our experiments, specifically in the case of our 
online covering datasets which are derived from an underlying dataset (an internet router graph), sampled across time. We would like to efficiently learn a good predictor $f$ from a family of functions $\mathcal{F}$, 
where the input to $f$ is an covering SDP $\mathcal{S}$ and the output is a weight vector. 
Each covering SDP instance can be written in terms of $d:=O(nm^2)$ variables, so we assume that each input instance $\mathcal{S}$ is encoded as a vector in $\mathbb{R}^{d}$. 
The output of $f$ is then an $n$-dimensional vector, so that each coordinate $i\in[n]$ represents the associated weight of the matrix $A_i$ selected in the covering SDP. 

To quantify the quality of each function $f$, we consider a loss function $L: f \times\mathcal{S}\rightarrow \mathbb{R}$ that outputs the cost of corresponding SDP cover. 
For example, $f$ can output a set of matrix weights given an instance $\mathcal{S}$. 
If the matrix weights are not feasible, then the loss function may output $\infty$. 
Otherwise if the matrix weights are feasible, then the loss function may output the corresponding cost, defined by the cost vector that is also given in $\mathcal{S}$. 

We would like to learn the best function $f \in \mathcal{F}$ that minimizes the objective:
\begin{equation}\label{eq:best:f}
\mathbb{E}_{\mathcal{S} \sim\mathcal{D}}[L(f,\mathcal{S})].
\end{equation}
Let $f^*$ denote the minimizer of the above objective, i.e., the optimal function in $\mathcal{F}$, so that 
\[f^*=\argmin \mathbb{E}_{\mathcal{S}\sim\mathcal{D}}[L(f,\mathcal{S})].\] 
We assume that for each SDP instance $\mathcal{S}$ and each $f\in\mathcal{F}$ that $f(\mathcal{S})$ and $L(f,\mathcal{S})$ can be computed in time that is a (small) polynomial in $n$ and $d$, which we denote by $T(n,d)$. 

We show that there exists an efficient algorithm that outputs a function $\hat{f}\in\mathcal{F}$ whose loss is within an additive $\eps$ of the loss of the optimal function $f^*$. 
\begin{theorem}\label{thm:learning}
There exists an algorithm that uses $\poly\left(T(n,d),\frac{1}{\eps}\right)$ samples and outputs a function $\hat{f}$ that satisfies with probability at least $\frac{9}{10}$,
\[ \mathbb{E}_{\mathcal{S}\sim\mathcal{D}}[L(\hat{f},\mathcal{S})] \le  \mathbb{E}_{\mathcal{S}\sim\mathcal{D}}[L(f^*,\mathcal{S})] + \eps,\]
where $f^*=\argmin \mathbb{E}_{\mathcal{S}\sim\mathcal{D}}[L(f,\mathcal{S})]$. 
\end{theorem}
Theorem~\ref{thm:learning} shows that only a small number of samples are needed to have a good probability of learning an approximately-optimal function $\hat{f}$, in a typical PAC-style bound. 
Specifically, the algorithm to compute $\hat{f}$ is simply the empirical risk minimizer, i.e., the algorithm minimizes the empirical loss after an appropriate number of samples are drawn. 
To prove correctness of Theorem~\ref{thm:learning}, we first  define the notion of pseudo-dimension for a function class, which generalizes the more familiar VC dimension to real functions, e.g., Definition $9$ in \cite{Feldman18}. 
\begin{definition}[Pseudo-Dimension]
Let $\mathcal{X}$ be a ground  set  and $\mathcal{F}$ be a set of functions from $\mathcal{X}$ to the interval $[0,1]$. 
Let $U=\{x_1,\cdots ,x_n\} \subset \mathcal{X}$ be a fixed set, $R=\{r_1, \cdots, r_n\}$ be a fixed set of real numbers with $r_i \in [0, 1]$ for all $i\in[n]$, and $f \in \mathcal{F}$ be a fixed function. 
The set $U_f= \{x_i \in U \mid f(x_i) \ge r_i\}$ is called the induced subset of $U$ formed by $f$ and $R$. 
The set $U$ with associated values $R$ is shattered by $\mathcal{F}$ if $|\{U_f \mid f \in \mathcal{F} \}|= 2^n$. 
The \emph{pseudo-dimension} of $\mathcal{F}$ is the cardinality of the largest shattered subset of $\mathcal{X}$ (or $\infty$).
\end{definition}
Let $\mathcal{G}$ be the class of functions in $\mathcal{F}$ composed with $L$, so that 
\[\mathcal{G} := \{L \circ f : f \in \mathcal{F} \}.  \]
Without loss of generality, we normalize the range of the loss function $L$ to $[0,1]$. 
Then the performance of the empirical risk minimization and the number of necessary samples can be related to the pseudo-dimension using standard bounds as follows:
\begin{theorem}
\label{thm:uni:dim}
\cite{AnthonyB99}
Let $\mathcal{D}$ be a distribution over problem instances $\mathcal{S}$ and $\mathcal{G}$ be a class of functions $g: \mathcal{S} \rightarrow[0, 1]$ with pseudo-dimension $d_{\mathcal{G}}$. 
Consider $t$ i.i.d.\ samples $\mathcal{S}_{1}, \mathcal{S}_{2}, \ldots, \mathcal{S}_{t}$ from $\mathcal{D}$.  
There exists a universal constant $C$, such that for any $\eps>0$, if $t \geq\frac{C \cdot d_{\mathcal{G}}}{\eps^2}$, then with probability at least $\frac{9}{10}$, 
\[\left|\frac{1}{t} \sum_{i=1}^{t} g\left(\mathcal{S}_{i}\right)-\mathbb{E}_{\mathcal{S}\sim \mathcal{D}} \, g(\mathcal{S})\right| \leq \eps\]
for all $g \in \mathcal{G}$. 
\end{theorem}
From the triangle inequality, we then have the following corollary. 
\begin{corollary}\label{cor:uni:dim}
Let $t \geq\frac{C \cdot d_{\mathcal{G}}}{\eps^2}$, where $C$ is the universal constant from Theorem \ref{thm:uni:dim}. 
Consider a set of $t$ independent samples $\mathcal{S}_1, \ldots,\mathcal{S}_t$ from $\mathcal{D}$ and let $\hat{g}$ be a function in $\mathcal{G}$ that minimizes $\frac{1}t \sum_{i=1}^t g(\mathcal{S}_i)$. 
Then with probability at least $\frac{9}{10}$,
\[\E_{\mathcal{S}\sim\mathcal{D}} [\hat{g}(\mathcal{S})] \le \E_{\mathcal{S}\sim\mathcal{D}}[g^*(\mathcal{S})] + 2 \eps.\]
\end{corollary}
Hence, the main question is whether we can bound the pseudo-dimension of the given function class $\mathcal{G}$. 
To that end, we first observe that the pseudo-dimension can be related to the VC dimension of a related class of threshold functions. 
In particular, this relationship has been instrumental in showing a number of existing learning bounds, e.g., \cite{Feldman18,IzzoSZ21,ErgunFSWZ22,ChenEILNRSWWZ22, chen2022}.
\begin{lemma}[Pseudo-dimension to VC dimension, Lemma $10$ in \cite{Feldman18}]\label{lem:PD:to:VC}
For any $g \in \mathcal{G}$, let $B_g(x,y)  =  \text{sgn}(g(x)-y)$, i.e., the indicator function of the region on or below the graph of $g$. 
Then the pseudo-dimension of $\mathcal{G}$ is equivalent to the VC dimension of the subgraph class $B_{\mathcal{G}}=\{B_g \mid g \in \mathcal{G}\}$.
\end{lemma}
We can also relate the VC dimension of a given function class to the computational complexity of the function class, i.e., the time complexity of computing a function in the class.
\begin{lemma}[Theorem 8.14 in \cite{AnthonyB99}]
\label{lem:VC:bound}
Let $h: \R^a \times \R^b \rightarrow \{ 0,1\}$ define the class
\[ \mathcal{H} = \{x \rightarrow h(\theta, x) : \theta \in \R^a\}. \]
Suppose that $h$ can be computed by an algorithm that takes as input the pair $(\theta, x) \in \R^a \times \R^b$ and returns $h(\theta, x)$ after no more than $t$ of the following operations:
\begin{itemize}
\item arithmetic operations $+, -,\times,$ and $/$ on real numbers,
\item jumps conditioned on $>, \ge ,<, \le ,=,$ and $=$ comparisons of real numbers, and
\item output $0,1$.
\end{itemize}
Then the VC dimension of $\mathcal{H}$ is at most $O(a^2 t^2 + t^2 a \log a)$.
\end{lemma}
We can now prove Theorem \ref{thm:learning} by roughly instantiating Lemma \ref{lem:VC:bound} using the computational complexity of any function in the function class $\mathcal{G}$.  
\begin{proof}[Proof of Theorem~\ref{thm:learning}]
The proof essentially follows from the proofs of similar theorems in the sample complexity bound sections of the papers \cite{Feldman18,IzzoSZ21,ErgunFSWZ22,ChenEILNRSWWZ22, chen2022} but we nonetheless repeat it here for completeness. By Theorem \ref{thm:uni:dim} and Corollary \ref{cor:uni:dim}, we observe that it suffices to bound the pseudo-dimension of the class $\mathcal{G} = L \circ \mathcal{F}$. 
By Lemma \ref{lem:PD:to:VC}, the pseudo-dimension of $\mathcal{G}$ is equivalent to the VC dimension of threshold functions defined by $\mathcal{G}$. 
Moreover by Lemma \ref{lem:VC:bound}, the VC dimension of the threshold functions defined by $\mathcal{G}$ is polynomial in the complexity of computing a member of the function class. 

Hence, Lemma \ref{lem:VC:bound} implies that the VC dimension of $B_{\mathcal{G}}$ defined in Lemma \ref{lem:PD:to:VC} is polynomial in the number of arithmetic operations needed to compute the threshold function associated to some $g \in \mathcal{G}$, which is polynomial in $T(n,d)$ by our definition. 
Therefore, the pseudo-dimension of $\mathcal{G}$ is also polynomial in $T(n,d)$ and the desired claim follows. 
\end{proof}
\section{A Simple Learning-Augmented Online Algorithm}
\label{app:simple}
In this section, we revisit a simple and possibly folklore learning-augmented online algorithm for online covering LPs. 
Its robustness ratio is asymptotically as good as the state-of-the-art online algorithm and its consistency ratio is also asymptotically as good as the advice if the advice is feasible. 
We recall that a similar approach also applies to a wide range of learning-augmented online minimization problems, including online SDPs, online LPs and SDPs with box constraints, online submodular cover \cite{gupta2020online}, and online covering with $\ell_q$-norm objectives \cite{shen2020online}.

The idea of the algorithm is to maintain two candidate solutions, the advice $x'$ and the solution obtained by a state-of-the-art online algorithm \cite{gupta2014approximating}. Let $\mathcal{O}$ be a state-of-the-art online algorithm and let $x^{(i)}$ be its solution obtained at round $i$. In the beginning, $x^{(0)}$ is initialized according to the online algorithm $\mathcal{O}$. We follow the online algorithm $\mathcal{O}$ until its objective $c^T x^{(i)}$ is larger than the objective of the advice $c^T x'$. Once the advice is infeasible, we ignore the advice and use the online algorithm. Transitioning either from the online algorithm to the advice or from the advice to the online algorithm pays only a constant factor. Once we transition, to maintain monotonicity of the online solution, we pick the coordinate-wise maximum between the advice and the solution by $\mathcal{O}$. The algorithm works as follows.

\begin{algorithm}[!htb]
\caption{A Simple Algorithm for Learning-Augmented Online Covering LPs} \label{alg:simple}
\begin{algorithmic}[1]
\State $i_a \gets \infty$. \Comment{$i_a$ is the round that we transition to the advice}
\For{$i = 1, 2, ...$} \Comment{each arriving row or constraint}
    \State Update $A$ by adding a new row $i$.
    \State Run $\mathcal{O}$ for round $i$ and obtain $x^{(i)}$.
    \If{$A x' \ge \mathbf{1}$ and $c^T x' \ge c^T x^{(i)}$} \Comment{use the online solution} \label{line:alg-simple-case-1}
        \State{$x \gets x^{(i)}$.}
    \EndIf
    \If{$A x' \ge \mathbf{1}$ and $c^T x' < c^T x^{(i)}$ and $i < i_a$} \Comment{transition to the advice} \label{line:alg-simple-case-2}
        \State $i_a \gets i$. \Comment{(only once at round $i_a$)}
        \For{$j \in [n]$}
            \State $x_j \gets \max\{x'_j, x^{(i-1)}_j\}$.
        \EndFor
    \EndIf
    \If{$A x' \ge \mathbf{1}$ does not hold} \Comment{use only the online solution}
        \If{$i < i_a$} \Comment{no transition to advice} \label{line:alg-simple-case-3}
            \State $x \gets x^{(i)}$.
        \EndIf
        \If{$i \ge i_a$} \Comment{has transitioned to advice} \label{line:alg-simple-case-4}
            \For{$j \in [n]$} \Comment{transition back to the online solution}
                \State $x_j \gets \max\{x'_j, x^{(i)}_j\}$.
            \EndFor
        \EndIf
    \EndIf
\EndFor
\end{algorithmic}
\end{algorithm}

\begin{theorem}
For the learning-augmented online covering LP problem, there exists an online algorithm that generates $x$ such that
\[c^T x \le \min\{O(c^T x'), O(\log k) \opt\}\]
when $x'$ is feasible, and $c^T x \le O(\log k) \opt$ when $x'$ is infeasible. Here, $k \le n$ is the row sparsity of $A$, i.e., the maximum number of non-zero entries of each row.
\end{theorem}

\begin{proof}
We use Algorithm \ref{alg:simple} with the $O(\log k)$-competitive online algorithm for $\mathcal{O}$.

\begin{theorem} [\cite{gupta2014approximating}]
There exists an $O(\log k)$-competitive algorithm for the online covering LP problem.
\end{theorem}

Algorithm \ref{alg:simple} considers four possible cases.
\begin{enumerate}
    \item The advice is feasible but $\mathcal{O}$ is better (line \ref{line:alg-simple-case-1}).
    \item The advice is feasible and it is better than $\mathcal{O}$ (line \ref{line:alg-simple-case-2} or $x$ is not updated).
    \item The advice is infeasible and never used (line \ref{line:alg-simple-case-3}).
    \item The advice is infeasible and was used when it was feasible (line \ref{line:alg-simple-case-4}).
\end{enumerate}

Clearly, Algorithm \ref{alg:simple} always returns a feasible solution $x$ and $x$ is updated in a non-decreasing manner.

We first consider the case when $x'$ is feasible. For the first case, we never use the advice, so
\[c^T x \le O(\log k)\opt \le \min\{O(c^T x'), O(\log k) \opt\}.\]
For the second case, either the advice becomes better than the online solution at round $i$ or $x$ is not updated because the advice becomes better than the online solution at a previous round $i_a < i$. For the previous case, by the design of Algorithm \ref{alg:simple}, we have that
\[c^T x \le c^T (x' + x^{(i-1)}) \le 2 c^T x' \le \min\{O(c^T x'), O(\log k) \opt\}.\]
For the later case, $x$ is not updated after round $i_a$ so the claim still holds.

The remaining cases are when $x'$ is infeasible. For the third case, we never use the advice, so $c^T x \le O(\log k)\opt$. For the last case, let $\opt_i$ denote the value of $\opt$ at round $i$. We know that the advice was feasible at some previous round $i' \ge i_a$ but became infeasible at round $i'+1 \le i$. Round $i'$ belongs to the second case, so at that round
\[c^T x \le 2 c^T x'.\]
In round $i'+ 1$, we transition back from the advice to the online algorithm, so
\[c^T x \le 2 c^T x' + O(\log k) \opt_{i'+1} \le \min\{O(c^T x'), O(\log k) \opt_{i'+1}\} \le O(\log k) \opt_{i'+1}.\]
In round $i$, $c^T x'$ does not change but $\opt$ changes. By the property of the online algorithm $\mathcal{O}$, we always have
\[c^T x \le  O(\log k) \opt_i.\]
\end{proof}


\end{document}